\documentclass[11pt]{article}

\usepackage[utf8]{inputenc} 
\usepackage[T1]{fontenc}
\usepackage{lmodern}
\usepackage{geometry}
\usepackage{physics}

\usepackage{bbm}
\usepackage{algorithmic} 
\usepackage{microtype}  
\usepackage{array}
\usepackage{multirow}
\usepackage{amsmath} 
\usepackage{amssymb}
\usepackage{amsthm} 
\usepackage{thmtools}

\usepackage[procnumbered,ruled,vlined,linesnumbered]{algorithm2e}

\SetCommentSty{mycommfont} % set comment style for algorithms
\usepackage{xcolor}  
\usepackage{xspace}
\usepackage{enumitem}    
\usepackage{tikz}   

\usepackage{comment} 

\usepackage{graphicx}
\graphicspath{{Figures/}}
\usepackage{caption}

\usepackage{subcaption}

\usepackage{xcolor}
\usepackage{nameref}
\definecolor{ForestGreen}{rgb}{0.1333,0.5451,0.1333}
\definecolor{DarkRed}{rgb}{0.65,0,0}
\definecolor{Red}{rgb}{1,0,0}
\usepackage[linktocpage=true,
pagebackref=true,colorlinks,
linkcolor=DarkRed,citecolor=ForestGreen,
bookmarks,bookmarksopen,bookmarksnumbered]
{hyperref}
\usepackage{cleveref}

\usepackage[procnumbered,ruled,vlined,linesnumbered]{algorithm2e}
\SetCommentSty{mycommfont} % set comment style for algorithms
\SetKwInput{KwGlobalVar}{Global variables}

\declaretheorem[numberwithin=section]{theorem}
\declaretheorem[numberlike=theorem]{lemma}

\declaretheorem[numberlike=theorem]{proposition}

\declaretheorem[numberlike=theorem]{corollary}
\declaretheorem[numberlike=theorem]{claim}
\declaretheorem[numberlike=theorem]{observation}

\declaretheorem[numberlike=theorem]{assumption}

\crefname{algorithm}{Algorithm}{Algorithms}
\Crefname{algorithm}{Algorithm}{Algorithms}

    \newtheoremstyle{TheoremNum}
        {\topsep}{\topsep}              %%% space between body and thm
        {\itshape}                      %%% Thm body font
        {}                              %%% Indent amount (empty = no indent)
        {\bfseries}                     %%% Thm head font
        {.}                             %%% Punctuation after thm head
        { }                             %%% Space after thm head
        {\thmname{#1}\thmnote{ \bfseries #3}}%%% Thm head spec
    \theoremstyle{TheoremNum}

\theoremstyle{definition}
\declaretheorem[numberlike=theorem]{definition}

\newcommand{\ot}{\tilde{O}}
\newcommand{\cbar}{c_T}
\newcommand{\roc}{reducing-or-covering\xspace}
\newcommand{\psetpair}{partition-set pair\xspace}
\newcommand{\ssetpair}{separators-set pair\xspace}

\newcommand{\Roc}{Reducing-or-Covering\xspace}
\newcommand{\Psetpair}{Partition-Set Pair\xspace}
\newcommand{\Ssetpair}{Separators-Set Pair\xspace}

\newcommand{\phibar}{\bar\phi}

\newcommand{\xorterminal}{expanders-or-terminal pair\xspace}
\newcommand{\Xorterminal}{Expanders-or-Terminal Pair\xspace}

\newcommand{\Gorig}{G^{\textrm{orig}}}

\newcommand{\calT}{\mathcal{T}\xspace}  
\newcommand{\calS}{\mathcal{S}\xspace}

\newcommand{\cL}{\mathcal{L}\xspace}

\newcommand{\poly}{\operatorname{poly}} 
\newcommand{\polylog}{\operatorname{polylog}}

\newcommand{\vol}{\operatorname{vol}}

\SetKwFor{RepTimes}{repeat for}{times}{end}

\newcommand{\eat}[1]{}

\newcommand{\textinput}{\operatorname{orig}}   
\newcommand{\textdeg}{\operatorname{deg}}

\newcommand{\cl}{\operatorname{cl}}
\newcommand{\isolate}{\operatorname{isolate}}

\newcommand{\textlocal}{\operatorname{local}}

\newcommand{\calG}{\mathcal{G}}

\title{Deterministic Small Vertex Connectivity in Almost Linear Time} 

\author{
        Thatchaphol
        Saranurak\thanks{\texttt{thsa@umich.edu}. University of Michigan, USA}
        \and
        Sorrachai Yingchareonthawornchai\thanks{\texttt{sorrachai.yingchareonthawornchai@aalto.fi}. Aalto University, Finland}
}

\date{}   
\begin{document}

\maketitle
        \pagenumbering{gobble}
        \begin{abstract}
In the \emph{vertex connectivity} problem, given an undirected $n$-vertex
$m$-edge graph $G$, we need to compute the minimum number of vertices
that can disconnect $G$ after removing them. This problem is one
of the most well-studied graph problems. From 2019, a new line of
work {[}Nanongkai \emph{et al.}~STOC'19;SODA'20;STOC'21{]} has used
randomized techniques to break the quadratic-time barrier and, very
recently, culminated in an almost-linear time algorithm via the recently
announced maxflow algorithm by Chen \emph{et al}. In contrast, all known deterministic algorithms are much slower. The fastest algorithm {[}Gabow
FOCS'00{]} takes $O(m(n+\min\{c^{5/2},cn^{3/4}\}))$ time where $c$
is the vertex connectivity. It remains open whether there exists a subquadratic-time deterministic algorithm for any constant 
$c>3$.

In this paper, we give the first deterministic almost-linear time
vertex connectivity algorithm for all constants $c$. Our
running time is $m^{1+o(1)}2^{O(c^{2})}$ time, which is almost-linear
for all $c=o(\sqrt{\log n})$. This is the first deterministic algorithm
that breaks the $O(n^{2})$-time bound on sparse graphs where $m=O(n)$, which is known for more
than 50 years ago {[}Kleitman'69{]}.

Towards our result, we give a new reduction framework to \emph{vertex
expanders} which in turn exploits our new almost-linear time construction
of mimicking network for vertex connectivity. The
previous construction by Kratsch and Wahlstr\"{o}m {[}FOCS'12{]}
requires large polynomial time and is randomized. 

\end{abstract}
  
        \newpage        
    	\setcounter{tocdepth}{2}  
    	\tableofcontents        
        \newpage
        
        \pagenumbering{arabic}

        \section{Introduction}

Vertex connectivity is one of most fundamental measures for robustness
of graphs. For any $n$-vertex $m$-edge graph $G=(V,E)$, the vertex
connectivity $\kappa_{G}$ of $G$ is the minimum number of vertices
that can disconnect $G$ after removing them. Fast algorithms for
computing vertex connectivity has been extensively studied \cite{Kleitman1969methods,Podderyugin1973algorithm,EvenT75,Even75,Galil80,EsfahanianH84,Matula87,BeckerDDHKKMNRW82,LinialLW88,CheriyanR94,NagamochiI92,CheriyanT91,Henzinger97,HenzingerRG00,Gabow06,Censor-HillelGK14}. (See \cite{NanongkaiSY19} for a discussion.) Since
2019, a new line of work \cite{NanongkaiSY19,ForsterNYSY20,li2021vertex} has used randomized techniques to break
the long-standing quadratic-time barrier and, very recently, culminated
in an almost-linear time algorithm via the recently announced maxflow
algorithm by \cite{chen2022maximum}. However, these new algorithms
are all Monte-Carlo and, thus, can make errors.

In contrast, all known deterministic algorithms are much slower. Improved
upon previous deterministic algorithms \cite{even1975network,HenzingerRG00}, in 2000, Gabow \cite{Gabow06} gave the previously fastest deterministic algorithm with
$O(m(n+\min\{c^{5/2},cn^{3/4}\}))$ time where $c$ is the vertex
connectivity. While linear-time algorithms are known when $c\le3$
\cite{Tarjan72,HopcroftT73}, it remains open whether there exists a subquadratic-time deterministic algorithm for any constant $c>3$. 

In this paper, we give the first deterministic almost-linear time
algorithm for computing vertex connectivity for all constants $c$.
To state our result precisely, we need some terminology. We say $(L,S,R)$
is a \emph{vertex cut} if $L,S,R$ partition a vertex set $V$ but
there is no edge between $L$ and $R$. The size of a vertex cut $(L,S,R)$
is $|S|$. Note that a vertex cut of size less than $c$ certifies
that $\kappa_{G}<c$. If there is no such cut, then $\kappa_{G}\ge c$,
and we say that $G$ is $c$-connected. Our main result is as follows:
\begin{theorem}
\label{thm:main}There is a deterministic algorithm that, given an
undirected graph $G$, in $m^{1+o(1)}2^{O(c^{2})}$ time either outputs
a vertex cut of size less than $c$ or concludes that $G$ is $c$-connected.
\end{theorem}

When focusing on sparse graphs where $m=O(n)$, the $O(n^{2})$-time
bound was shown for more than 50 years ago \cite{Kleitman1969methods}. \Cref{thm:main} is
the first that deterministically breaks this barrier.

\paragraph{Related Works.}
Significant effort has been devoted to  devising deterministic algorithms that is as fast as their randomized counterparts.
Examples include the line of work on deterministic minimum edge cut algorithm \cite{KawarabayashiT15,HenzingerRW17,saranurak2021simple,LiP20deterministic} where, recently, Li \cite{li2021deterministic} showed a deterministic almost-linear time algorithm for minimum edge cuts even in weighted graphs.
Another example are deterministic  expander decomposition algorithms and their applications to deterministic Laplacian solvers and approximate max flow \cite{chuzhoy2020deterministic}. 

\paragraph{Historical Note.}
%\thanks{
%
The previous version of this paper  \cite{gao2019deterministic} contains two results: deterministic subquadratic-time algorithms for balanced cuts and vertex connectivity. The paper was split into two newer papers. The newer version of the balanced cut algorithms is \cite{chuzhoy2020deterministic}, which gave the first almost-linear time deterministic balanced cut algorithm. Our paper is the new version of the vertex connectivity algorithms from \cite{gao2019deterministic}.%}

\subsection{Our Approach}

Towards our main result (\Cref{thm:main}), we develop two main algorithmic
components. Below, we explain how their combination immediately implies
\Cref{thm:main}. Suppose our goal is just to check whether $G$ is
$c$-connected for $c=O(1)$.

\paragraph{``Almost'' Reduction to Vertex Expanders.}

Our starting point is the observation that $c$-connectivity can be
checked fast on \emph{vertex expanders}. For any vertex cut $(L,S,R)$
of $G$, its \emph{expansion} is $h(L,S,R)=\frac{|S|}{|S|+\min\{|L|,|R|\}}$.
When $h(L,S,R)$ is small, we say that the cut is \emph{sparse}. We
say $G$ is a $\phi$-vertex expander if $\min_{(L,S,R)}h(L,S,R)\geq\phi$
and we simply call it a vertex expander when $\phi\ge1/n^{o(1)}$.
Observe that for any vertex cut $(L,S,R)$ of size less than $c$
in a vertex expander must be very unbalanced, i.e., $\min\{|L|,|R|\}\le cn^{o(1)}$.  Suppose we are given a vertex $x\in L$ where $|L|\le cn^{o(1)}$,
the local flow algorithms  can find the cut of size less than
$c$ in $\text{poly}(cn^{o(1)})$ time \cite{NanongkaiSY19} or in $c^{O(c)}n^{o(1)}$  time \cite{chechik2017faster}. Thus, by simply calling the
local flow algorithm from every vertex, one can check $c$-connectivity
of a vertex expander in almost-linear total time. However, in general the graph
might contain a very sparse cut. This is exactly the hard instance
of all previous deterministic algorithms. 

Now, our new framework essentially allows us to assume that the input
graph is a vertex expander. More precisely, give any graph $G=(V,E)$,
our subroutine $\textsc{ExpandersOrTerminal}(G,c,\phi)$ computes
in $m^{1+o(1)}/\phi$ time a collection ${\cal G}$ of $\phi$-vertex
expanders and a terminal set $T\subset V$ with the following guarantee: 
\begin{enumerate}
\item $G$ is $c$-connected iff every $H\in{\cal G}$ is $c$-connected
and the \emph{Steiner connectivity} of $T$ is at least $c$ (i.e.
there is no vertex cut of size less than $c$ separating any terminal
pair $x,y\in T$). 
\item ${\cal G}$ has almost-linear total size, i.e.,~$\sum_{H\in{\cal G}}|V(H)|=n^{1+o(1)}$.
\item $|T|\le\phi n^{1+o(1)}$. 
\end{enumerate}
See details in \Cref{def:expander-or-terminal family} and \Cref{thm:fast xorterminal pair}. We will set  $\phi=1/n^{o(1)}$ so that $|T|\ll n$ is small.

We employ the reduction as follows. Given $G$, we call $\textsc{ExpandersOrTerminal}(G,c,\phi)$
and check $c$-connectivity of all vertex expanders $H\in{\cal G}$
in almost-linear total time, since their total size is almost-linear.
If any $H\in{\cal G}$ is not $c$-connected, then we are done. Otherwise,
we still need to check the Steiner connectivity of $T$. Therefore,
in almost-linear time, the framework allows us to ``focus'' on a
small terminal set $T$ and this is where our second algorithmic component
can help.

\paragraph{Mimicking Network for Vertex Connectivity.}

Given any $G=(V,E)$ and terminal set $T$, our second subroutine
$\ensuremath{\textsc{VertexSparsify}(G,T,c)}$ computes a small graph
$H$ of size proportional to $T$ and $H$ preserves \emph{all} minimum
vertex cuts between terminals $T$ up to size $c$. More precisely,
let $\mu_{G}(A,B)$ be the minimum number of vertices that, after
removing from $G$, there is no path from $A$ to $B$ left (we allow
removing vertices in $A$ and $B$). The guarantees of $H$ is
that, for any $A,B\subseteq T$, $\min\{\mu_{G}(A,B),c\}=\min\{\mu_{H}(A,B),c\}$.
We also require that $\kappa_{H}\ge\min\{\kappa_{G},c\}$ (otherwise,
there might be a new minimum vertex cut in $H$). We call $H$ a \emph{$c$-vertex
connectivity mimicking network} or, for short a \emph{$(T,c)$-sparsifier}
for $G$. Below, we show that $\textsc{VertexSparsify}$ can be implemented
in almost-linear time. See \Cref{def:mimicking network} and \Cref{thm:vertex sparsifier} for details.
\begin{theorem}
\label{thm:overview sparsifier}Our subroutine $\ensuremath{\textsc{VertexSparsify}(G,T,c)}$
computes a $(T,c)$-sparsifier $H$ for $G$ of size $|E(H)|=|T|2^{O(c^{2})}$
in $m^{1+o(1)}2^{O(c^{2})}$ time.
\end{theorem}

How do we exploit this sparsifier? Let $T$ be the terminal set from
$\textsc{ExpandersOrTerminal}$ and $H=\ensuremath{\textsc{VertexSparsify}(G,T,c)}$.\footnote{For a technical reason, we actually need to compute $H=\ensuremath{\textsc{VertexSparsify}(G,T',c)}$ where $T' = \cup_{v\in T} N^{c}_G(v)$ and $N^{c}_G(v)$ is an arbitrary set of $c$ neighbors of $v$. See \Cref{sec:main vc alg} for details.}  
Since $H$ preserves all minimum vertex cuts between $T$ and also
does not introduce any new minimum cut, it suffices to check if $H$
is $c$-connected. By setting $\phi=1/2^{O(c^{2})}n^{o(1)}$, we have
$|T|\le n/2^{O(c^{2})}$ and so $|E(H)|\le n/2$. Thus, we have reduced
the vertex connectivity problem to a graph of half the size in almost-linear
time. By repeating this framework at most $O(\log n)$ rounds, 
we are done. That completes the explanation how the high-level components
fit together.

\subsection{New Mimicking Networks}

The new mimicking network in \Cref{thm:overview sparsifier} can be of
independent interest. 

Let us compare \Cref{thm:overview sparsifier} with related results
on mimicking network literature. Kratsch and Wahlstr\"{o}m \cite{kratsch2020representative} showed
an algorithm that computes a $(T,c)$-sparsifier (without the condition
that $\kappa_{H}\ge\min\{\kappa_{G},c\}$) of size $O(|T|^{3})$ in
some large polynomial time and their algorithm is randomized. In our applications it is crucial that the size is linear in $|T|$ and the
algorithm is deterministic. When we consider edge cuts instead of
vertex cuts, $c$-edge-connectivity mimicking networks of size $|T|\text{poly}(c)$
can be computed in  $m^{1+o(1)}c^{O(c)}$ time \cite{chalermsook2021vertex} or $n^{O(1)}$ time
\cite{liu2020vertex}.

\Cref{thm:overview sparsifier} can be viewed as an adaptation of
$c$-edge connectivity mimicking networks from \cite{liu2019vertex,chalermsook2021vertex} to $c$-vertex connectivity, but there are many technical obstacles we need to overcome.
Basically, this is because not all techniques for edge cuts generalize
to vertex cuts. Even when they do generalize, they require careful and
complicated definitions and arguments in order to carry out the approach.
For concrete examples, we observe that the ``cut enumeration'' approach
of \cite{chalermsook2021vertex} fails completely for vertex cuts.\footnote{In  \cite{chalermsook2021vertex}, they only show a fast algorithm for enumerating cuts $(S,V\setminus S)$ whose one side induced a connected graph (i.e.~$G[S]$ is connected), and then argue that this kind of cuts suffice for their construction. This is not true for vertex cuts. We would need to list vertex cuts $(L,S,R)$ where $G[L]$ is not connected too. While listing cut where $G[L \cup S]$ is connected suffices, it is not clear how to do it fast.}
So, instead, we take the approach based on covering sets and intersecting
sets of \cite{liu2019vertex}. To order to generalize this approach, we arrive at a
complicated and delicate definition of \emph{\roc} \psetpair (See \Cref{def:roc partition}). Nonetheless, while trying to overcome these technical obstacles in constructing mimicking networks, we believe that there is one technical lesson that might be useful beyond this work. 

\paragraph{Fast Closure Oracles via Hypergraphs.}
In our algorithm, we need to perform an operation that is analogous to contracting an edge $e$, but we do it on a vertex $v$. This operation is called \emph{neighborhood closure}, or just \emph{closure}, for short. The closure has been use as an important primitive in \cite{kratsch2020representative}.
Given a graph $G$ with a vertex $v$, \emph{closing} $v$ in $G$ is to add a clique between all neighbors of $v$, and remove $v$ from $G$. %
This operation clearly takes a lot of time when $v$ has large degree. 
We observe that if we ``lift'' the problem to hypergraphs, it turns out that the equivalent operation is as follows:
\emph{closing $v$ on a hypergraph} $H$ is to merge all hyperedges $e$ containing $v$ into one hyperedge, i.e.~insert $e' = \cup_{e \ni v} e$ and remove all $e$ that contains $v$. This simple equivalence is formalized in \Cref{obs:hypergraph clique equiv}.

Now, it turns out that the closure operation on hypergraph is very easy to (implicitly) implement, for example, by using the union-find data structure. 
More specifically, our algorithm in \Cref{sec:reducing set} will need to implement a data structure on a graph that handle closure operations in an online manner. The key technique that enables our algorithm to run in almost-linear time is to ``lifting'' the graph into a hypergraph and implementing closure operations on the hypergraph.

\subsection{Organization}
We describe necessary background and terminologies in \Cref{sec:prelim}. We describe the full vertex connectivity algorithm in \Cref{sec:main vc alg}.  Then, we explain in \Cref{sec:xorterminals} how to implement the reduction to vertex expanders, i.e.~the subroutine $\textsc{ExpandersOrTerminal}(G,c,\phi)$. For the rest of the paper,  we explain how to compute $c$-vertex connectivity mimicking network, i.e.~the subroutine
$\ensuremath{\textsc{VertexSparsify}(G,T,c)}$. In \Cref{sec:mimicking},  we explain a reduction from constructing mimicking networks to an object called \emph{covering set}. In \Cref{sec:covering set}, we show a reduction from constructing covering set to another object called \emph{\roc \psetpair}. We give fast implementations for computing a \roc \psetpair in \Cref{sec:reducing set}.  Finally, we discuss open problems in \Cref{sec:open problems}.

\section{Preliminaries} \label{sec:prelim}

Let $G = (V,E)$ be an undirected graph. Throughout this paper, we consider only  undirected graphs (unless stated otherwise).  For $x,y \in V$, we say that a path in $G$ is
$(x,y)$-\textit{path} if it starts with $x$ and ends with $y$. For $S, T \subseteq V$, we say that a path
$P$ is an $(S,T)$-\textit{path} in $G$ if it is an $(s,t)$-path for
some $s\in S, t \in T$. We also denote $E_G(S,T)$ to
be the set of edges whose one endpoint is in $S$ and the other
endpoint is in $T$. For any $x \in V$, we denote $N_G(x)$ to be the
set of neighbors of $x$ in graph $G$. For any $S \subseteq V$, we
denote $N_G(S)$ to be the set of neighbors of some vertex in $S$ that
is not in $S$. We also denote $N_G[S] = S \cup N_G(S)$ and $\vol_G(S)
= \sum_{x \in S}\textdeg_G(x)$ where $\textdeg_G(x)$ is the degree of
$x$ in $G$, which is the number edges incident to $x$.  For any $X
\subseteq V$, we denote $G[X]$ as subgraph of $G$ induced by $X$.  For set
notations, we denote set difference as $S - T$. We denote $G - S$ to
be the graph $G$ after removing all vertices in $S$.   We also denote
$V(G)$ and $E(G)$ to be the set of vertices in $G$ and the set of
edges in $G$, respectively. 

Let $S \subseteq V$ and $A ,B \subseteq V$. We say that $S$ is a \textit{separator} in $G$ if
$G - S$ is not connected.  We say that $S$ is an $(x,y)$-separator in $G$ for
some $x,y \in V$ if $G -S$ does not have an $(x,y)$-path, and $x,y \not
\in S$. We say that $S$ is an $(X,Y)$-\textit{separator} in $G$ for some disjoint sets $X, Y \subseteq V$ if $G - S$ does not have an $(X,Y)$-path and $X \cap S = Y \cap S = \emptyset$.  We denote $\kappa_G(x,y)$ to be the minimum size of $(x,y)$-separator in $G$ or $n-1$ if such a
separator does not exist. \textit{Vertex connectivity} of $G$ is $\kappa_G =
\min_{x,y \in V}\kappa_G(x,y)$.  We also say that $S$ is an $(A,B)$-\textit{weak separator} if
$G-S$ does not have an $(A,B)$-path (but possibly $A \cap S \neq
\emptyset$ and $B \cap S \neq \emptyset$).  We denote $\mu_G(A,B)$ to be the minimum
size of $(A,B)$-weak separator in $G$. Note that $\mu_G(A,B) \leq
\min\{|A|,|B|\}$ because $A$ and $B$ are trivial $(A,B)$-weak
separators. For the purpose of vertex connectivity, we can assume WLOG
that $G$ is simple (i.e., no-self loops,
and no parallel edges). 

A \textit{vertex cut} in $G = (V,E)$ is a triple $(L,S,R)$ that forms
a partition of $V$ such that $L \neq \emptyset, R\neq\emptyset$ and
$E_G(L,R) =\emptyset$.  A
\textit{vertex expansion} of $(L,S,R)$ is $h_G(L,S,R) =
\frac{|S|}{\min\{|L \cup S|, |S \cup R|\}}$. A \textit{vertex
  expansion} of a graph $G$ is $h(G) = \min_{(L,S,R)}h_G(L,S,R)$. Let $\phi \in (0,1)$. We
say that a vertex cut $(L,S,R)$ is $\phi$-\textit{sparse} if $h_G(L,S,R) < \phi$.
We say that $G$ is a $\phi$-\textit{vertex expander} if $G$ does not have a
$\phi$-sparse vertex cut, i.e., $h(G) \geq \phi$.

\begin{definition} \label{def:localvc}
  A $(c,\nu)$-\textit{LocalVC algorithm} takes as inputs an initial
  vertex $x$ in a graph $G = (V,E)$, and two parameters $\nu, c$ such
  that $c\cdot \nu \leq  |E|/\tau$ (for some constant $\tau \geq 1$) and outputs either:
  \begin{itemize}
  \item $\bot$ certifying that there is no vertex cut $(L,S,R)$ such
    that $x \in L, \vol_G(L) \leq \nu$ and $|S| < c$, or
  \item a vertex set $L$ such that $|N_G(L)| < c$. 
  \end{itemize}
\end{definition}

\begin{theorem}  [\cite{NanongkaiSY19, chechik2017faster}]
 There are deterministic $(c,\nu)$-LocalVC algorithms that takes $O(\nu c^{O(c)})$ %%
 time and $O(\nu^2 c)$ time, respectively.  
\end{theorem}

\begin{theorem}[\cite{NagamochiI92}] \label{thm:nagamochi}  
  Given a graph $G = (V,E_G)$ with $n$ vertices and $m$ edges and a parameter $c > 0$, there is a linear-time algorithm that outputs another graph $H = (V,E_H)$ satisfying
  the following properties: 
  \begin{itemize}
  \item $G$ is $c$-connected if and only if $H$ is $c$-connected. Furthermore, for all $S \subseteq V$ such that $|S| < c$, $S$ is a separator in $G$ if and only if $S$ is a separator in $H$, and 
  \item $H$ has arboricity $c$.  In particular, $|E_H| \leq nc$ and for all $S \subseteq V$, $|E_H(S,S)| \leq c|S|$. 
  \end{itemize}
\end{theorem}

In the language of vertex sparsifier, the following terms will be use throughout the paper. 

\begin{definition} \label{def:tc equivalent}
  Let $G$ and $H$ be graphs that contain the same terminal set $T$ and $c > 0$ be an integer. We say that $G$ and $H$ are $(T,c)$-\textit{equivalent} if
for all pair $A,B \subseteq T$, we have $\min\{\mu_H(A,B),c\} =
\min\{\mu_G(A,B),c\}$. 
\end{definition}

\begin{definition}\label{def:cut recoverable}
  Let $H$ and $G$ be graphs and $c > 0$ be an integer. %
  \begin{itemize}
      \item $H$ is \textit{cut-recoverable} for $G$ if  $V(H) \subseteq V(G)$ and  every separator in $H$ is a separator in $G$,
      \item $H$ is $c$-\textit{cut-recoverable} for $G$ if  $V(H) \subseteq V(G)$ and  every separator of size $<c$ in $H$ is a separator in $G$, and
      \item $H$ is $c$-\textit{mincut-recoverable} for $G$ if $V(H) \subseteq V(G)$ and for all $s,t \in V(H)$ every min $(s,t)$-separator of size $<c$ in $H$ is an $(s,t)$-separator in $G$.
  \end{itemize}
\end{definition}
\begin{proposition} \label{pro:recoverable imply monotone}
If $H_1$ is cut-recoverable for $G$, then $\kappa_{H_1} \geq \kappa_G$.
If $H_2$ is $c$-cut-recoverable for $G$ or $c$-mincut-recoverable, then $\kappa_{H_2} \geq \min\{c, \kappa_G\}$. 
\end{proposition}
\begin{proof}
If $S^*$ be a min separator in $H_1$, then $S^*$ is a separator in $G$, and thus $\kappa_{H_1} = |S^*| \geq \kappa_G$.  Now consider $H_2$. If $\kappa_{H_2} \geq c$, then $\kappa_{H_2} \geq \min\{c,\kappa_G\}$, and we are done. Now, assume $\kappa_{H_2} < c$. Let $S^*$ be a min separator in $H$. Since $H_2$ is $c$-cut-recoverable (or $c$-mincut-recoverable) and $|S^*| < c$, $S^*$ is a separator in $G$, and thus $c > \kappa_{H_2} = |S^*| \geq \kappa_G$. Therefore, $\kappa_{H_2} \geq \kappa_G = \min\{c,\kappa_G\}$. %
\end{proof}

Observe that cut-recoverable property is transitive: If $A$ is cut-recoverable for $B$ and $B$ is cut-recoverable for $C$, then $A$ is cut-recoverable for $C$.  Using \Cref{def:tc equivalent} and \Cref{def:cut recoverable}, the algorithm in \Cref{thm:nagamochi} outputs a graph $H$ such that (1) $H$ and $G$ are $(V,c)$-equivalent, (2) $H$ is $c$-cut recoverable for $G$, and (3) $H$ has arboricity $c$.

\section{Vertex Connectivity Algorithm} \label{sec:main vc alg}

In this section, we prove \Cref{thm:main}. We first consider a warm up problem when the input graph is already a $\phi$-expander, and then discuss  our key tools to handle general graphs.  

\paragraph{Warm up.}  Suppose that $G$ is $\phi$-vertex expander. For intuitive purpose,  we can think of $\phi = 1/\log n$ (we
use $\phi = 1/n^{o(1)}$ in the real algorithm).   We show that we can decide $c$-connectivity in
$\ot(n\phi^{-2})$ time.  Since $h(G) \geq \phi$, any vertex cut
$(L,S,R)$ where $|S| < c$ must be such that $|L| \leq c\phi^{-1}$ or $|R|
\leq c\phi^{-1}$. Therefore, $G$ is $c$-connected if and only if there is
a set $L \subseteq V$ where $|L| \leq c\phi^{-1}$ and $|N(L)| < c$
where $N(L)$ is the neighbors of $L$. Observe that $\vol(L) \leq |L|^2
+ |L|c = O(c^2\phi^{-2})$.   This is exactly where the \emph{local vertex connectivity }(LocalVC)
algorithm introduced in \cite{NanongkaiSY19} can help us. This algorithm
works as follows: given a vertex $x$ in a graph $G$ and parameters
$\nu$ and $c$, either (1) certifies that there is no set $L \ni x$
where $\vol(L)\le\nu$ and $|N(L)|<c$, or (2) returns a set $L$ where
$|N(L)|<c$. See \Cref{def:localvc} for a formal definition. There are
currently two deterministic algorithms for this problem: a
$\tilde{O}(\nu^{2}c)$-time algorithm \cite{NanongkaiSY19} and a $O(\nu c^{O(c)})$-time algorithm by a slight adaptation of the algorithm by Chechik et al.~\cite{chechik2017faster}.
As $c=O(1)$, we will use the $O(\nu c^{O(c)})$-time algorithm here.
From the above observation about the set $L$, it is enough to run
the LocalVC algorithm from every vertex $x$ with a parameter
$\nu = O(c^{-2}\phi^{-2})$ to decide if such $L$ exists. This takes $O(n\phi^{-2})$
total time to decide $c$-connectivity of $G$ if $G$ is
$\phi$-expander. When $\phi^{-1} = \log n$, we have a $\ot(n)$ time
algorithm for deciding $c$-connectivity of  $\phi$-vertex
expanders. That is, we have the following. 
\begin{proposition} \label{pro:expander unbalanced}
 Given an $\phi$-vertex expander $G = (V,E)$, and parameters $c > 0$ and $\phi \in (0,1)$,   there is an $\ot(|V|\frac{ c^{O(c)}}{\phi^2})$-time algorithm that outputs either a min separator of size $<c$ or certifies that $G$ is $c$-connected. %
\end{proposition}

\paragraph{General Graphs.} In general, $G$ is not necessarily a $\phi$-vertex expander. To handle the general case, we compute a $(c,\phi)$-\textit{\xorterminal} of $G$. %

\begin{definition} \label{def:expander-or-terminal family}
Given a graph $G = (V,E)$ and $c, \phi$ as parameters, and let $\calG$
be a set of graphs,  $T \subseteq V$ be  a set of vertices in $G$, we say
that a pair $(\calG,T)$ is a $(c,\phi)$-\textit{\xorterminal} for $G$ if every $H \in \calG$ is a $\phi$-vertex expander and $c$-mincut-recoverable for $G$ (\Cref{def:cut recoverable}) and the pair can
characterize $c$-connectivity of $G$: $G$ is $c$-connected if and only if
\begin{enumerate}
\item every graph $H \in \calG$ is $c$-connected, and
\item for all $x,y \in T$, $\kappa_G(x,y) \geq c$. %
\end{enumerate}
We also say that $T$ is a \textit{terminal} set of $G$.
\end{definition}

Our technical tool is that we can compute in
almost-linear time a \textit{small} $(c,\phi)$-\xorterminal in that the size of
the terminal set $T$ is small (e.g., $\leq n\phi$) and at the
same time the total size of $\phi$-vertex-expanders in $\calG$ is
almost linear. In \Cref{sec:xorterminals}, we prove the following:

\begin{theorem} [Fast Algorithm for \Xorterminal] \label{thm:fast xorterminal pair} 
  Given a graph $G = (V,E)$ where $m = |E|, n = |V|$ and
  $c,\phi$ as inputs where $\phi < 1/(2c\log^2 n)$ and min-degree of $G \geq c$, there is an
  $O(m^{1+o(1)}/\phi)$ time algorithm, denoted as $\textsc{ExpandersOrTerminal}(G,c,\phi)$, that outputs a separator of size $<c$, or  a $(c,\phi)$-\xorterminal  $(\calG,T)$ for $G$ such that 
  \begin{enumerate}
  \item $|T| \leq n^{1+o(1)}\phi$,  %
   \item $\sum_{H \in \calG} |V(H)| = O(n^{1+o(1)})$.%
  \end{enumerate}
\end{theorem}

By applying \Cref{thm:fast xorterminal pair} on a graph $G$ with
appropriate parameters $c$ and $\phi$, we obtain a
$(c,\phi)$-\xorterminal $(\calG,T)$ for $G$ such that $|T| \leq n^{1+o(1)}\phi$ and total size of $\phi$-vertex expanders is almost linear. Since every graph $H \in
\calG$ is a $\phi$-vertex expander, we can decide $c$-connectivity of
all graphs in $\calG$ by applying \Cref{pro:expander unbalanced} on
every $\phi$-expander in $\calG$. This takes $O(n^{1+o(1)}c^{O(c)}
\phi^{-2})$ time. By  \Cref{def:expander-or-terminal family}, the final step for deciding $c$-connectivity of $G$ is to verify if $\kappa_G(x,y) \geq
c$ for all $x,y \in T$.   However, it is not known how to compute $\min_{x,y \in T} \kappa_G(x,y)$ in \textit{deterministically} near-linear time\footnote{In fact, the problem generalizes vertex connectivity because when $T = V$ it becomes vertex connectivity. In our case, however, we make progress because we can guarantee that $|T| \leq n^{1+o(1)}\phi$.}.

Instead of computing $\kappa_G(x,y)$ for all $x,y \in T$, we compute
\emph{$c$-vertex connectivity mimicking network}.  That is, we want to sparsify the
graph while preserving small cuts with respect to terminal set
$T$. The intuition is that if $G$ is not $c$-connected because there
is a vertex cut that separates $x \in T$ from $y \in T$, then  there must
be some $A, B \subseteq T$ where $|A|,|B|\ge c$ such that $\mu_G(A,B) < c$. Therefore, it
is enough to compress $G$ into another graph $H$ so that the vertex
cut that corresponds to $\mu_G(A,B) < c$ is preserved and at the same time $H$ does not have a new vertex cut of size less than $c$. %

More precisely,  let $G$ and $H$ be graphs that contain the same terminal set $T$.  We say that $G$ and $H$ are $(T,c)$-\textit{equivalent} if
for all pair $A,B \subseteq T$, we have $\min\{\mu_H(A,B),c\} =
\min\{\mu_G(A,B),c\}$. We say that $H$ is $c$-\textit{cut-recoverable} for $G$ if  $V(H) \subseteq V(G)$ and  every separator of size $<c$ in $H$ is a separator in $G$. By \Cref{pro:recoverable imply monotone}, $\kappa_H \geq \min\{c,\kappa_G\}$. We now define $c$-connectivity mimicking
network.

\begin{definition} [$c$-Vertex-Connectivity Mimicking
  Network] \label{def:mimicking network}
  We say that a graph $H$ is a \textit{pairwise} $(T,c)$-\textit{sparsifier} for $G = (V,E)$ that contains a terminal set $T \subseteq V$  if 
  \begin{enumerate} [noitemsep,nolistsep]
  \item $G$ and $H$ are $(T,c)$-equivalent, and 
  \item $H$ is $c$-cut-recoverable for $G$. In particular, $\kappa_H \geq \min\{c,\kappa_G\}$. %
  \end{enumerate}
\end{definition}

This leads to our second new technical
tool. In \Cref{sec:mimicking}, we prove the following:  
\begin{theorem}\label{thm:vertex sparsifier}
  Let $G = (V,E)$ be an undirected graph with $n$ vertices and $m$
  edges and a terminal set $T
  \subseteq V$ and $c > 0$ be a parameter. There is an
  $O(m^{1+o(1)} 2^{O(c^2)})$-time algorithm, denoted as
  $\textsc{VertexSparsify}(G,T,c)$, that outputs a $(T,c)$-sparsifier $H$
  for $G$ where $|E(H)| \leq c|V(H)|$ and $|V(H)| = O(|T| 2^{O(c^2)})$. 
\end{theorem}

Recall that the terminal set $T$ in $(c,\phi)$-\xorterminal is of size
at most $n^{1+o(1)}\phi$ (\Cref{thm:fast xorterminal pair}). After applying \Cref{thm:vertex sparsifier}
using parameters $(G,T,c)$, we obtain a $(T,c)$-sparsifier $H$ for
$G$ where $|E(H)| \leq c|V(H)|$ and $|V(H)| = O(|T| 2^{O(c^2)}) =
O(n^{1+o(1)} \phi 2^{O(c^2)})$. The key claim is that $G$ is
$c$-connected if and only if $H$ is $c$-connected.  Furthermore,  if
we set $\phi^{-1} = 10n^{o(1)}2^{O(c^2)}$, then we have that $|V(H)|
\leq n/10$ and $|E(H)| \leq nc/10$. Therefore, we can recurse on graph $H$ to decide
$c$-connectivity of $G$. This can repeat at most $O(\log n)$ time, and
by the choice of $\phi$, each iteration takes $O( m^{1+o(1)} +
n^{1+o(1)} c^{O(c)} + m^{1+o(1)} 2^{O(c^2)}) = O(m^{1+o(1)} 2^{O(c^2)})$.

We are now ready to prove \Cref{thm:main}. We describe the algorithm and analyze correctness and its running time. We start with the algorithm description. Recall that  $\textsc{ExpandersOrTerminal}$ is the algorithm in \Cref{thm:fast xorterminal pair}; and $\textsc{VertexSparsify}$ is the algorithm in \Cref{thm:vertex sparsifier}. We use $\phi^{-1} = 10n^{o(1)}2^{O(c^2)}$. If the graph is $\phi$-vertex expander, we can decide $c$-connectivity efficiently using \Cref{pro:expander unbalanced}. The algorithm is described in \Cref{alg:main vertex conn alg}.  %

\begin{algorithm}[H]
  \DontPrintSemicolon
  \KwIn{ A graph $G = (V,E)$, and connectivity parameter $c$} 
  \KwOut{ A separator of size $<c$ or $\bot$ certifying that $G$ is $c$-connected.}%
  \BlankLine
  \lIf{ $|V| \leq 100 c$}
    { output $\bot$ if $\kappa_G \geq c$, and output a separator of size $<c$ otherwise.} 
  \lIf{\normalfont{min}-degree of $G < c$} {\Return{\normalfont{the} set of neighbors of the min degree vertex}. \label{line:mindeg}} 
  Call $\textsc{ExpandersOrTerminal}(G,c,\phi)$ where
  $\phi^{-1}= 10n^{o(1)}2^{O(c^2)}$\;
  \lIf{\normalfont{a} separator $Z$ is returned} {\Return{$Z$.}} %
  Otherwise, $(c,\phi)$-expanders-or-terminal pair  $(\calG,T)$ is returned.\;
  \lIf{$\exists H \in \mathcal{G}$ \normalfont{s.t.} $H$ is not $c$-connected}
  {\Return{\normalfont{the} corresponding min separator.} \label{line:expander set}}
  $T \gets \bigcup_{v  \in T} N^{c}_G(v)$ where   $N^{c}_G(v)$ is an arbitrary set of $c$ neighbors of $v$ \;
  $H \gets \textsc{VertexSparsify}(G,T,c)$ \; 
  \Return{\normalfont{\textsc{Main}}$(H, c)$} 
\caption{\textsc{Main}$(G, c)$}
\label{alg:main vertex conn alg}
\end{algorithm}

\paragraph{Correctness.} We prove that \Cref{alg:main vertex conn alg}
outputs a separator of size $<c$ if and only if $G$ is not $c$-connected. We first prove that  if $G$ is
$c$-connected, then   \Cref{alg:main vertex conn alg} returns $\bot$. By \Cref{thm:fast xorterminal pair}, we have that 
$\textsc{ExpandersOrTerminal}(G,c,\phi)$ must return an
$(\phi,c)$-\xorterminal $(\calG,T)$ for $G$. By
\Cref{def:expander-or-terminal family}, every graph $H \in \calG$ is
$c$-connected.  By \Cref{thm:vertex sparsifier},
$\textsc{VertexSparsify}(G,T,c)$, that outputs a $(T,c)$-sparsifier
$H$ for $G$.  By \Cref{def:mimicking network}, $\kappa_H \geq
\min\{c,\kappa_G \} \geq c$. Therefore, $H$ is also
$c$-connected. \Cref{alg:main vertex conn alg} then recurses on
$H$. By the choice of $\phi$,  we have $|V(H)| \leq n/10$. The process
repeats for at most $O(\log n)$ time, and eventually the graph in the
base case is $c$-connected. Therefore, \Cref{alg:main vertex conn alg}
returns $\bot$.

Next, we prove that if $G$ is not $c$-connected, then the
\Cref{alg:main vertex conn alg} returns a separator of size $<c$. By \Cref{thm:fast xorterminal pair},  if
$\textsc{ExpandersOrTerminal}(G,c,\phi)$ returns a separator, then $G$ is not $c$-connected, and we are done.  Now assume that
$\textsc{ExpandersOrTerminal}(G,c,\phi)$ returns an
$(\phi,c)$-\xorterminal $(\calG,T)$ for $G$ (\Cref{def:expander-or-terminal family}).  If one of the graph in
$\calG$ is not $c$-connected, then the algorithm returns the corresponding min separator in $H$ of size $<c$. Since every graph in $\calG$ is $c$-mincut-recoverable (\Cref{def:cut recoverable}), the same separator is also a separator in $G$ and we
are done. Now assume that every graph in $\calG$ is $c$-connected.
Since  $(\calG,T)$ is a $(\phi,c)$-\xorterminal, there is $x,y \in T$
such that $\kappa_G(x,y) < c$.  By \Cref{thm:vertex sparsifier},
$\textsc{VertexSparsify}(G,T,c)$, that outputs a $(T,c)$-sparsifier
$H$ for $G$. Observe that, prior to running $\textsc{VertexSparsify}(G,T,c)$, $T$ is updated to be $\bigcup_{v \in T}N_G^c(v)$ where $N^{c}_G(v)$ be an arbitrary set of $c$ neighbors of $v$. 

\begin{claim}
 $H$ is not $c$-connected.   
\end{claim}
\begin{proof}
  Since $\kappa_G(x,y) < c$, there is a vertex cut $(L,S,R)$ in $G$
  where $x \in L\cap T$ and $y \in R \cap T$ and $|S| < c$. Since
  $N^c_G(x) \subseteq T$ and $N^c_G(y) \subseteq T$, we have that $|T
  \cap (L\cup S)| \geq c+1$ and $|T \cap (S \cup R)| \geq c+1$. Let $A
  = T \cap (L \cup S)$ and  $B = T \cap (S \cup R)$. We have that $S$ is an $(A,B)$-weak separator in $G$ and thus $\mu_G(A,B)
  \leq |S| < c$.   Since $H$
  and $G$ are $(T,c)$-equivalent, we have that
  $$ \min\{\mu_H(A,B),c\} = \min\{\mu_G(A,B),c\} = \mu_G(A,B) < c.$$
  Therefore, $\mu_H(A,B) < c$, and thus there is a separator $S'$ such
  that $H - S'$ does not have a path from $A$ to $B$ (note that $S'$ may contain $A$ and $B$).  Since $|S'| <
  c$, and $|A| \geq c+1, |B| \geq c+1$, $A$ is not contained entirely
  in $S'$, and also $B$ is not contained entirely in $S'$. Therefore,
  there are $x' \in A - S'$ and $y' \in B - S'$ such that there is no
  path from $x'$ to $y'$ in   $H - S'$. So, $S'$ is a separator of
  size $<c$ in $H$. 
\end{proof}
Since $H$ is not $c$-connected and the algorithm recurses on $H$, and
repeat for $O(\log n)$ time. It eventually hits the base case where
the input graph is not $c$-connected. The base case will return a separator $S$ of size $< c$. Since $H$ is a $(T,c)$-sparsifier for $G$, $H$ is $c$-cut-recoverable for $G$. Since $c$-cut-recoverable property is transitive (i.e., if $A$ is $c$-cut-recoverable for $B$ and $B$ is $c$-cut-recoverable for $C$, then $A$ is $c$-cut-recoverable for $C$), we conclude that $S$ is also a separator in $G$. Therefore, \Cref{alg:main vertex conn alg} must output a separator of size $<c$.

\paragraph{Running Time.} Let $n$ be the number of vertices and $m$ be the number of edges of the input graph $G$. By \Cref{line:mindeg}, we can assume that $m \geq nc$. By design, we set $\phi^{-1} = 10n^{o(1)} 2^{O(c^2)}$. By \Cref{thm:fast xorterminal pair}, $\textsc{ExpandersOrTerminal}(G,c,\phi)$ takes $O(m^{1+o(1)}/\phi) = O(m^{1+o(1)} 2^{O(c^2)})$ time. If a separator is returned, we are done. Now, assume $(c,\phi)$-\xorterminal $(\calG,T)$ of $G$ is returned. By \Cref{thm:fast xorterminal pair}, every graph in $\calG$ is $\phi$-vertex expander where $\sum_{H \in \calG} |V(H)| = O(n^{1+o(1)})$ and also $|T| \leq n^{1+o(1)}\phi$. By \Cref{pro:expander unbalanced}, we can solve $c$-connectivity of every $\phi$-expander in $\calG$ in $O( \sum_{H \in \calG} |V(H)| \cdot c^c \phi^{-2}) = O(n^{1+o(1)} 2^{O(c^2)})$ time. By \Cref{thm:vertex sparsifier}, $\textsc{VertexSparsify}(G,T,c)$, that outputs a $(T,c)$-sparsifier $H$ for $G$ in $O(m^{1+o(1)} 2^{O(c^2)})$ time. Finally, we prove that $|E(H)| \leq m/10$ and $|V(H)| \leq n/10$. By \Cref{thm:vertex sparsifier}, we have $|E(H)| \leq |V(H)|c$ and $|V(H)| = O(|T| 2^{O(c^2)})$. Since $|T| \leq n^{1+o(1)} \phi$, and $\phi^{-1} = 10n^{o(1)}2^{O(c^2)}$ (with appropriate parameters), we have that $|V(H)| \leq n/10$, and also $|E(H)| \leq |V(H)|c \leq m/10$. Therefore, we repeat at most $O(\log n)$ time where each time takes total $O(m^{1+o(1)}2^{O(c^2)})$.

\section{Expanders or Terminals} \label{sec:xorterminals}

This section is devoted to proving \Cref{thm:fast xorterminal pair} which is the basis of our approach.   We organize the proof into four subsections. We give an overview of the proof in \Cref{sec:overview}. In \Cref{sec:recursive vertex cuts}, we prove structural lemma for vertex cuts. In \Cref{sec:slow xor alg}, we show a recursive algorithm that computes the object from \Cref{def:expander-or-terminal family}. In \Cref{sec:fast implementaion of xor}, we show that the algorithm can be implemented in almost linear-time using the notion of most balanced sparse cuts.

\subsection{Overview} \label{sec:overview}

 Our new insight for proving \Cref{thm:fast xorterminal pair} is the following structural lemma about vertex cuts: 
 
 \begin{lemma} [Informal] \label{lem:structure vc glgr informal}
  For any vertex cut $(L,S,R)$ of $G$, consider forming the graph $G_L$ by 
  \begin{enumerate}[noitemsep, nolistsep]
      \item  removing all vertices of $R$,
      \item replacing $R$ with a clique $K_R$ of size $c$, and 
      \item adding a biclique between $S$ and $K_R$,
  \end{enumerate}
  as well as $G_R$ symmetrically. Then, $G$ is $c$-connected if and only if 
  \begin{enumerate} [noitemsep,nolistsep]
      \item for all pair $x,y \in S$, $\kappa(x,y) \geq c$, and  %
      \item both $G_L$ and $G_R$ are $c$-connected. 
  \end{enumerate}
 \end{lemma}
 
  See \Cref{sec:recursive vertex cuts} for the formal statement and its proof. The structural lemma suggests a divide-and-conquer algorithm for computing $(c,\phi)$-\xorterminal for $G$. In the base case, if $G$ is already a $\phi$-expander, then we return $(\{G\},\{\})$. Otherwise, $G$ has $\phi$-sparse cuts. Among of all $\phi$-sparse cuts, we select the one $(L,S,R)$ which maximizes $\min\{|L|, |R|\}+|S|$ with some $n^{o(1)}$ approximation factor from the most balanced one.  From an approximate most balanced cut $(L,S,R)$, we construct $G_L$ and $G_R$ as defined in \Cref{lem:structure vc glgr informal},
  and recurse on $G_L$ and $G_R$. 
  Let $(\calG_L,T_L)$ be a $(c,\phi)$-\xorterminal for $G_L$ and let $(\calG_R,T_R)$ be a $(c,\phi)$-\xorterminal for $G_R$. Then, by applying \Cref{lem:structure vc glgr informal}, we can prove that $(\calG_L \cup \calG_R, T_L \cup T_R \cup S)$ is also a $(c,\phi)$-\xorterminal for $G$. We refer to \Cref{sec:slow xor alg} for more details.  To bound the running time, suppose that we can compute $n^{o(1)}$-approximate most balanced cut in $O(m^{1+o(1)}/\phi)$ time (see \Cref{lem:balancedorexpander}). 
  Given this subroutine, we can show that the recursion depth is $O(\log n)$ using a standard argument, and therefore the total running time is $O(m^{1+o(1)}/\phi)$. Furthermore, we should expect the size of terminal set to be $n^{1+o(1)}\phi$ because we only add the new terminals from  $\phi$-sparse cuts. 

\subsection{A Divide-and-Conquer Lemma for Vertex Cuts} \label{sec:recursive vertex cuts}

\begin{definition} [$c$-Left and $c$-Right Graphs] \label{def:cleftrightgraphs}
  Let $(L,S,R)$ be a vertex cut in graph $G = (V,E)$ and $c >0$ be a
  parameter such that $|R| \geq |L| > c$. We define
  $c$-\textit{left graph}  $G_L$ (w.r.t. $(L,S,R)$) as the graph $G$
  after the following operations:
  \begin{enumerate} [noitemsep,nolistsep]
  \item Let $K_R \subseteq R$ be an arbitrary subset of $c$ vertices
    of $R$, 
  \item remove vertices in $R - K_R$, 
  \item for all pair $x,y \in K_R$, add edge $(x,y)$, 
  \item for all $x \in S$ and $y \in K_R$, add edge $(x,y)$, 
  \item (optional) remove all edges in $E(S,S)$. 
  \end{enumerate}
  Intuitively, $G_L$ is $G$ after replacing $R$ with a clique of size
  $c$ followed by adding edges between the clique and
  $S$. Furthermore, each vertex $K_R$ corresponds to one of the $c$
  distinct vertices in $R$.
  
  Similarly, we define $c$-\textit{right graph} $G_R$
  (w.r.t. $(L,S,R)$) as the graph $G$ after the following operations:
  \begin{enumerate} [noitemsep,nolistsep]
  \item Let $K_L \subseteq L$ be an arbitrary subset of $c$ vertices
    of $L$, 
  \item remove vertices in $L - K_L$, 
  \item for all pair $x,y \in K_L$, add edge $(x,y)$, 
  \item for all $x \in S$ and $y \in K_L$, add edge $(x,y)$, 
  \item (optional) remove all edges in $E(S,S)$. 
  \end{enumerate}
  Intuitively, $G_R$ is $G$ after replacing $L$ with a clique  of size
  $c$ followed by adding edges between the clique and
  $S$. Furthermore, each vertex $K_L$ corresponds to one of the $c$
  distinct vertices in $L$. 
\end{definition}

Observe that, by definition, $V(G_L) \subseteq V(G)$ and $V(G_R) \subseteq V(G)$.

\begin{lemma} [Divide-and-Conquer Lemma for Vertex Cuts] \label{lem:dacforvc}
  Given a vertex cut $(L,S,R)$ in $G = (V,E)$ and a parameter $c >0 $
  where $|R| \geq |L| \geq c$, 
  $G$ is $c$-connected if and only if it satisfies the following
  properties: %
  \begin{enumerate} 
  \item $\kappa_G(x,y) \geq c$ for all $x,y \in S$, and 
  \item Both of its $c$-left and $c$-right graphs (i.e., $G_L$ and $G_R$) are $c$-connected.
  \end{enumerate}
  Furthermore, $G_L$ and $G_R$ are $c$-mincut-recoverable for $G$. 
\end{lemma}

We divide the proof of \Cref{lem:dacforvc} into two lemmas. Let
$(L,S,R)$ be a vertex cut in $G = (V,E)$,  $G_L$ be a $c$-left graph,
and $G_R$ be a $c$-right graph of $G$ w.r.t. $(L,S,R)$, respectively.

\begin{lemma} \label{lem:forward}
  If $G$ is not $c$-connected, then at least one of the properties in \Cref{lem:dacforvc} is false. 
\end{lemma}
\begin{lemma} \label{lem:backward}
  $G_L$ and $G_R$ are $c$-mincut-recoverable for $G$.  Therefore, if at least one of the properties in \Cref{lem:dacforvc} is false,  then $G$ is not $c$-connected.   %
\end{lemma}

We now prove each lemma in turn. 
\begin{proof}[Proof of \Cref{lem:forward}]
  First, observe that we can assume $|S| \geq c$ and $\kappa_G(x,y) \geq c$  for all $x,y\in S$. This is because if $|S| < c$, then $G_L$ and $G_R$ are not $c$-connected (since $S$ is a separator in both $G_L$ and $G_R$) and we are done. 
  Also, if there is a pair of
  vertices $x,y \in S$ such that  $\kappa_G(x,y) < c$, then the first property of \Cref{lem:dacforvc} is false and we are done too.
  
  Now, since $G$ is not $c$-connected, $G$ has a vertex cut $(L^*,S^*, R^*)$ where $|S^*| < c$. 
  We will prove that $S^*$ is a separator in $G_L$ or $G_R$.
  \begin{claim} \label{claim:scaplstar}
  We have $S \cap L^* = \emptyset$  or $S \cap R^* = \emptyset$.  Also, the two sets $S \cap L^*$ and $S \cap R^*$ cannot be both empty.%
  \end{claim} 
  \begin{proof}
  We prove the first statement. Suppose the two sets are both non-empty. There are $u \in S \cap L^*$ and $v \in S \cap R^*$, and
   thus $u \in S, v \in S$. So, $\kappa_G(u,v) \geq c$. Since $u \in
   L^*$ and $v \in R^*$, we have $\kappa_G(u,v) \leq |S^*| < c$, a
   contradiction.  Next, we prove the second statement. Suppose the two sets are both
   empty. Then, $S \subseteq S^*$, and thus $c \leq |S| \leq |S^*| < c$, a contradiction. 
 \end{proof}
  We assume WLOG that $S \cap R^* = \emptyset$ and $S \cap L^* \neq
  \emptyset$. The case where $S \cap L^* =
  \emptyset$ is symmetric.  Since $S \cap R^* = \emptyset$, we have $L
  \cap R^* \neq \emptyset$ or $R \cap R^* \neq \emptyset$.  It remains
  to prove the following claim, see \Cref{fig:crossing-diagram2} for
  illustration.
  
\begin{figure}[!h]
\centering
\includegraphics[width=0.75\linewidth]{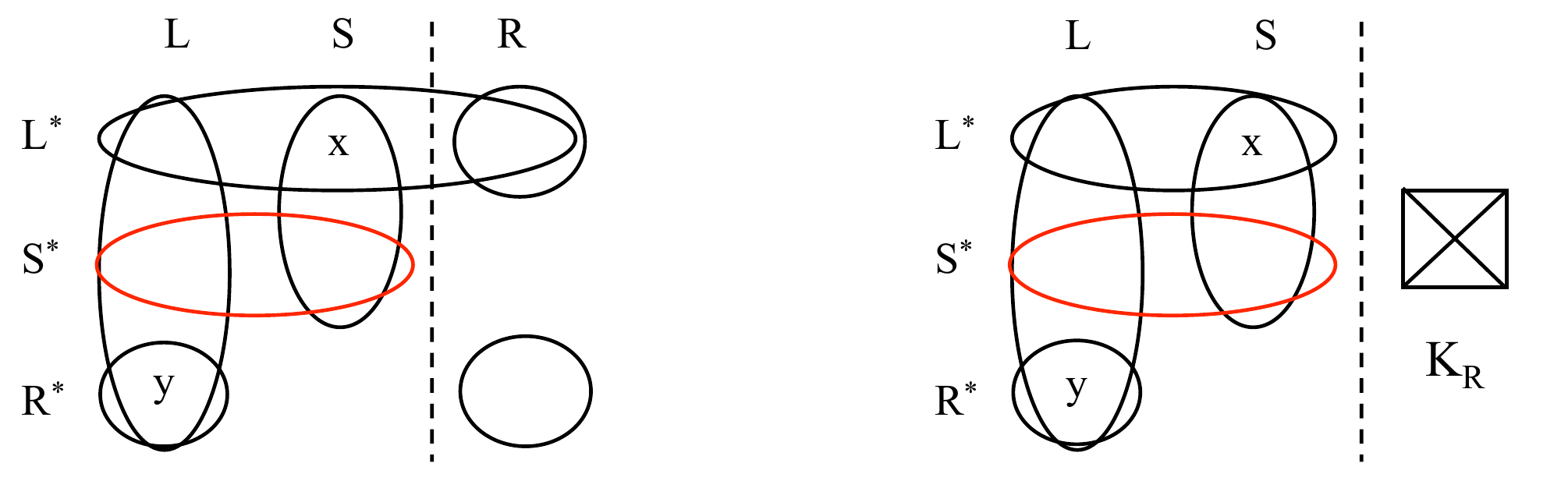}
  \caption{A crossing-diagram for two vertex cuts $(L,S,R)$ and
    $(L^*,S^*,R^*)$ before and after transformation from $G$ to
    $G_L$. If $L \cap R^* \neq \emptyset$, we show that $S^*$ is a vertex cut in $G_L$}
  \label{fig:crossing-diagram2}
\end{figure}

  \begin{claim}
    If $L \cap R^* \neq \emptyset$, then $S^*$ is a separator in
    $G_L$, and if $ R \cap R^* \neq \emptyset,$ then $S^*$ is a
    separator in $G_R$. 
  \end{claim}
  \begin{proof}
   We prove the case $L \cap R^* \neq \emptyset$.  Let $x \in S\cap
   L^*$ and $y \in L \cap R^*$. Recall that $K_R$
   is a clique that is added to $G_L$, and denote $K_R$ as the set
   of vertices of $K_R$. Observe that $( (L^* - R) \cup
   K_R, S^*, R^* - R)$ forms a partition of vertex set in $G_L$. It
   remains to prove that there is no edge from  $(L^* - R) \cup
   K_R$ to $R^* - R$ in $G_L$. If true, then the partition is a
   vertex cut in $G_L$ and we are done. Indeed,  since $(L^*,S^*, R^*)$ is a vertex cut in $G$, there is no edge from  $(L^* -R, R^*- R)$ in $G_L$. We next prove that there is no edge
  from $K_R$ to $R^* -R$.
  Let $v$ be an arbitrary vertex in $K_R$. By \Cref{def:cleftrightgraphs}, $N_{G_L}(v) \subseteq S
  \cup K_R$, which is disjoint from $L$.  But observe that $R^* -R \subseteq L$.%
  Therefore, there is no edge from $K_R$ to $R^* - R$, and we are done. The proof
  for the case $R \cap R^* \neq \emptyset$ is symmetric. 
  \end{proof}
  \textbf{Remark.} The same proof goes through if we further remove
  edges in $E_{G_L}(S,S)$ in $G_L$ (and also $E_{G_R}(S,S)$ in $G_R$). 
 \end{proof}

 \begin{proof}[Proof of \Cref{lem:backward}]
   We prove the result for $G_L$ (the argument for $G_R$ is symmetric). Fix $s,t \in V$ such that $\kappa_{G_L}(s,t) < c$.  Let $(L',S',R')$ be a vertex cut in $G_L$ where $|S'| = \kappa_{G_L}(s,t) < c, s \in L', t \in R'$. Recall
    that $K_R$ is the clique in $G_L$. %
    \begin{claim}
     $K_R \subseteq L'$ or $K_R \subseteq R'$.
   \end{claim}
   \begin{proof}
    We first show that $K_R \subseteq L'
    \cup S'$ or $K_R \subseteq S' \cup R'$. Since
    $|K_R| = c > |S'|$, we have $K_R \not \subseteq
    S'$. Since $K_R$ is a clique, $K_R$ cannot be in both
    $L'$ and $R'$. Therefore, $K_R \subseteq L'  \cup S'$ or $K_R \subseteq S' \cup R'$.

     Next we show that $K_R \cap S' = \emptyset$. Suppose
     otherwise. Let $v$ be an arbitrary vertex in $S' \cap
     K_R$. Since $K_R \subseteq L'  \cup S'$ or $K_R
     \subseteq S' \cup R'$, we have that  $N(v)
    \subseteq L' \cup S'$ or $N(v) \subseteq R' \cup  S'$.  If $N(v)
    \subseteq L' \cup S'$, then $(L' \cup \{v\}, S' - \{v\}, R')$ is a
     vertex cut. If $N(v) \subseteq R'
    \cup S'$, then $(L', S'- \{v\},R' \cup \{v\})$ is a vertex
    cut. Either way, $S' -\{v\}$ is an $(s,t)$-separator which is smaller than
    min $(s,t)$-separator, a contradiction.
  \end{proof}

      Since    $K_R \subseteq L'$ or $K_R \subseteq R'$, by
      replacing $K_R$ with the vertex set $R$, we have that $S'$ is
      also a separator in $G$. Furthermore, for any $x \in L'$  and $y \in R'$, we have that $S'$ is also
      $(x,y)$-separator in $G$.  This completes the proof of \Cref{lem:backward}.

      \textbf{Remark.} In addition, the same proof goes through even if we further remove
 edges in $E_{G_L}(S,S)$ in $G_L$ (and also $E_{G_R}(S,S)$ in $G_R$) because of the following claim: 
      \begin{claim} \label{claim:krkl}
    If $K_R \subseteq L'$, then $S \subseteq L' \cup S'$.
    Similarly, if $K_R \subseteq R'$, then $S \subseteq R' \cup S'$. 
  \end{claim}
  \begin{proof}
    We prove the first statement. The second statement is similar. Suppose $S \not \subseteq L' \cup S'$. There is a vertex $r \in S \cap R'$. Since there are biclique edges between $S$ and $K_R$, there is an edge between $r$ and a vertex in $K_R$. Since $r \in R'$ and $K_R \subseteq L'$, there is an edge between $L'$ and $R'$ in $G_L$, a contradiction. %
  \end{proof}
    Observe that $G$ can be obtained from $G_L$ by replacing $K_R$ with the vertex set $R$ and  adding original edges $E'$ inside $S$. By \Cref{claim:krkl}, every edge $(u,v) \in E'$ does not end with $L'$ and $R'$ simultaneously.  Using the same argument as above, $S'$ is  a separator in $G$ (even we add edges inside $S$ back). %
\end{proof}

\subsection{An Algorithm for \Xorterminal} \label{sec:slow xor alg}

Observe that if $G$ is a $\phi$-vertex expander, then $(\{G\},\{\})$ is
$(c,\phi)$-\xorterminal of $G$. If $G$ is not a
$\phi$-expander, there is a $\phi$-sparse cut $(L,S,R)$. Then, we can
recurse on the $c$-left and $c$-right graphs of $G$
w.r.t. $(L,S,R)$ (\Cref{def:cleftrightgraphs}). At the end of recursion, we obtain a set of
$\phi$-expanders and a set of $\phi$-sparse cuts. By \Cref{lem:dacforvc}, we can
show the pair of sets above form a $(c,\phi)$-expander-or-terminal
family.  \Cref{alg:expander or terminal family} shows the detail of
the algorithm. %

\begin{algorithm}[H]
  \DontPrintSemicolon
  \KwIn{ A graph $G = (V,E)$, and connectivity parameter $c$ and $\phi
  \in (0,1/10), \bar \phi \geq \phi$ where $\phi \leq (2c)^{-1}$.} 
  \KwOut{A $(c,\phi)$-\xorterminal of $G$. }
  \BlankLine
   \lIf{$G$ \normalfont{is} a $\phi$-vertex expander}{
     \Return{$(\{G\},\{\})$}.}
   Let $(L,S,R)$ be a $\bar \phi$-sparse cut  in $G$.\; %
    Let $G_L$ and $G_R$ be $c$-left and $c$-right graphs of $G$
    w.r.t. $(L,S,R)$.\;
      $(\calG_L, T_L) \gets
      \textsc{SlowExpandersOrTerminals}(G_L,c,\phi,\bar \phi)$\;
       $(\calG_R, T_R) \gets
       \textsc{SlowExpandersOrTerminals}(G_R,c,\phi, \bar \phi)$\;
        \Return{$(\calG_L \cup \calG_R,  T_L \cup T_R  \cup  S )$}.  
\caption{\textsc{SlowExpandersOrTerminals}$(G,c,\phi,\bar \phi )$}
\label{alg:expander or terminal family}
\end{algorithm}

\begin{lemma} \label{lem:correctness expander or terminal}
  \Cref{alg:expander or terminal family} takes a graph $G = (V,E)$ where $m = |E|, n = |V|$ and
  $c,\phi$ as inputs and outputs a $(c,\phi)$-\xorterminal $(\calG,T)$ of
  $G$.   %
\end{lemma}

\begin{proof}
  We prove that \Cref{alg:expander or terminal
  family} outputs a $(c,\phi)$-\xorterminal of every
connected graph $G$ by induction on number of vertices $N$. Observe that
the parameters $c,\phi,\phibar$ are fixed throughout the algorithm. If $|V(G)| \leq 1/(2c)$, then $G$ is always a
$\phi$-expander since $\phi \leq 1/(2c)$, and thus the pair
$(\{G\},\{\})$ is $(c,\phi)$-\xorterminal of $G$.   For
all $N \geq 1/(2c)$, we now assume that \Cref{alg:expander or terminal family} outputs a
$(c,\phi)$-\xorterminal of every connected graph of
at most $N-1$ vertices, and prove that \Cref{alg:expander or terminal family} outputs a
$(c,\phi)$-\xorterminal of every connected graph of
at most $N$ vertices.

Let $G$ be a graph with $N$ vertices. If $G$ is a $\phi$-expander, we
are done. Otherwise, we are given a $\phibar$-sparse cut
$(L,S,R)$. Let $G_L$ and $G_R$ be the $c$-left and $c$-right graphs of
$G$ w.r.t. $(L,S,R)$, respectively.  We prove that $n_{G_L} \leq N-1$
and $n_{G_R} \leq N-1$.  Indeed, since $(L,S,R)$ is $\phibar$-sparse, $|L|+1
\geq \phibar^{-1} \geq 2c$, and thus $|R| \geq |L| \geq 2c -1 >
c$. Therefore, $n_{G_L} = c + |S|+|R| \leq |L|-1+|S|+|R| =
N-1$. Similarly, we have $n_{G_R} = |L|+|S|+c \leq |L|+|S|+|R|-1 =
N-1$.  By inductive hypothesis,  $(\calG_L,T_L)$ is
$(c,\phi)$-\xorterminal for $G_L$, and similarly
$(\calG_R,T_R)$ is $(c,\phi)$-\xorterminal for $G_R$. Thus, every
graph in $\calG_L$ and in $\calG_R$ is a $\phi$-vertex expander.  

It remains to prove that $(\calG_L \cup \calG_R, T_L \cup T_R \cup S)$
is a $(c,\phi)$-\xorterminal of $G$.  We first prove that if $G$ is $c$-connected, then both conditions in
\Cref{def:expander-or-terminal family} hold. Since $G$ is
$c$-connected, we have that for all $x,y \in V(G)$, $\kappa_G(x,y) \geq c$. In
particular, for all $x,y \in T_L \cup T_R \cup S$, $\kappa_G(x,y) \geq
c$. Next, we prove that every graph in $\calG_L \cup \calG_R$ is
$\phi$-vertex expander and $c$-connected.  Since $G$ is $c$-connected, \Cref{lem:dacforvc} implies
that $G_L$ and $G_R$ are $c$-connected and $|S| \geq c$.  Since  $(\calG_L,T_L)$ is
$(c,\phi)$-\xorterminal for $G_L$, and similarly
$(\calG_R,T_R)$ is $(c,\phi)$-\xorterminal for
$G_R$,  every graph $H \in \calG_R \cup \calG_L$ is $\phi$-expander, and $c$-connected.

Finally, we prove that if $G$ is not $c$-connected, then at least one of the
conditions \Cref{def:expander-or-terminal family} is violated for
$(\calG_L \cup \calG_R, T_L \cup T_R \cup S)$.  If $\kappa_G(x,y) <c$
for some $x,y \in S$, then we are done. Now assume that for all $x,y \in S, \kappa_G(x,y) \geq
c$.  So, \Cref{lem:dacforvc} implies that  $G_L$ is not
$c$-connected or $G_R$ is not $c$-connected.  Now, we assume WLOG that
$G_L$ is not $c$-connected.   Since  $(\calG_L,T_L)$ is $(c,\phi)$-\xorterminal of $G_L$ and $G_L$ is not
$c$-connected, at least the one of the conditions \Cref{def:expander-or-terminal family} is violated for
$(\calG_L, T_L)$ for $G_L$. If there is $H \in \calG_L$ that is not
$c$-connected, then we are done.  Now, assume that there is a pair
$x,y \in T_L$ such that $\kappa_{G_L}(x,y) < c$. Since $G_L$ is $c$-mincut recoverable,  we have
$ \kappa_G(x,y) < c$, and this completes the proof.
\end{proof}

\subsection{A Fast Implementation of \Cref{alg:expander or terminal
    family}} \label{sec:fast implementaion of xor}
We start with an important primitive for computing sparse cut. 

\begin{lemma} [\cite{LongS22}] \label{lem:balancedorexpander}
Let $G=(V,E)$ be an $n$-vertex $m$-edge graph. Given a parameter
$0<\phibar\le1/10$ and $1\le r\le\left\lfloor \log_{20}n\right\rfloor $,
there is a deterministic algorithm, $\textsc{BalancedCutOrExpander}(G,\phibar,r)$, that computes a vertex cut $(L,S,R)$
(but possible $L=S=\emptyset$) of $G$ such that $|S|\le\phibar|L\cup
S|$ which further satisfies 
\begin{itemize}
\item either $|L\cup S|,|R\cup S|\ge n/3$; or
\item $|R\cup S|\ge n/2$ and $G[R\cup S]$ is a $\phi$-vertex expander
for some $\phi=\phibar/\log^{O(r^{5})}m$.
\end{itemize}
The running time of this algorithm is
$O(m^{1+o(1)+O(1/r)}\log^{O(r^{4})}(m)/\phi)$. %
\end{lemma}
 
When $G[S \cup R]$ is a $2\phi$-vertex expander, our $c$-right graph $G_R$ is still an $\phi$-vertex expander. 

\begin{claim} \label{claim:GR still a expander}
  If $G[S \cup R]$ is a  $2\phi$-vertex expander, then the $c$-right
  graph $G_R$ is $\phi$-vertex expander. 
\end{claim}
\begin{proof}
    Suppose there is a $\phi$-sparse cut $(L',S',R')$ in $G_R$. By
    removing the clique in $G_R$, we obtain another cut $(L'',S'',R'')$
    in $G$ where $L'' = L' -V(K_L)$ $S'' = S' - V(K_L)$ and $R'' = R -
    V(K_L)$. The expansion $h(L'',S'',R'') = \frac{|S''|}{|L''|+|S''|}
    \leq \frac{|S'|}{|L'|+|S'| -c} \leq 2\frac{|S'|}{|L'|+|S'|} <
    2\phi$. So, $(L'',S'',R'')$ is $2\phi$-sparse in $G$, a  contradiction. 
\end{proof}

The full algorithm is described in \Cref{alg:fast expander or terminal family}. Here, we denote $G^{\textrm{orig}}$ as the original input graph. We assume that the min-degree of $G^{\textrm{orig}}$ is $\geq c$. Thus, \Cref{def:cleftrightgraphs} implies that  any $c$-left and $c$-right graph have min-degree at least $c$. Therefore, we can assume that min-degree is at least $c$. During the recursion, if we obtain a balanced sparse cut $(L,S,R)$ such that $|S| < c$, then  we know that the original input graph $G^{\textrm{orig}}$ is not $c$-connected. Furthermore, if we want to output the corresponding cut,  we use the fact that $G_L$ and $G_R$ are $c$-mincut-recoverable. Therefore, it is enough to compute one min $s,t$ separator where $s$ and $t$ are on the different side of $(L,S,R)$ for an arbitrary $s \in L$ and $t \in R$. This takes $O(mc)$ time by Fold-Fulkerson max-flow algorithm. %

\begin{algorithm}[H]
  \DontPrintSemicolon
  \KwGlobalVar{the original input graph $G^{\textrm{orig}}$ whose min degree is $\geq c$.} %
  \KwIn{ A graph $G = (V,E)$, and connectivity parameter $c$ and $\phi
  \in (0,1)$ where $\phi \leq (2c)^{-1}$.} 
  \KwOut{A $(c,\phi)$-\xorterminal of $G$. }
  \BlankLine
  $r \gets \log \log |V|$\;
  $\phibar \gets 2\cdot \phi \cdot \log^{O(r^5)}|E|$\;
  $(L,S,R) \gets \textsc{BalancedCutOrExpander}(G, \phibar,r)$\;
  \If{$L = S = \emptyset$}{
     By \Cref{lem:balancedorexpander}, $G$ is a $2\phi$-vertex expander. \;
    \Return{$(\{G\},\{\})$}
  }
  \If{$|S| < c$} { \textbf{terminate}  and output min $(s,t)$-separator for an arbitrary $s \in L$ and $t \in R$. \label{line:declare not connected}} %

    By \Cref{lem:balancedorexpander}, $(L,S,R)$ is $\phibar$-sparse.\;
       Let $G_L$ and $G_R$ be $c$-left and $c$-right graphs of $G$
       w.r.t. $(L,S,R)$.\;
       Remove the edges in $E_{G_L}(S,S)$ from $G_L$. \label{line:remove edges in S}\;
             $(\calG_L, T_L) \gets   \textsc{FastExpandersOrTerminals}(G_L,c,\phi)$\;
    \If{$\min\{|L\cup S|, |R \cup S| \} \geq |V|/3$}  {
      Remove the edges in $E_{G_R}(S,S)$ from $G_R$. \label{line:remove edges in S R}\;
         $(\calG_R, T_R) \gets
         \textsc{FastExpandersOrTerminals}(G_R,c,\phi)$\;
         \Return{$(\calG_L \cup \calG_R,  T_L \cup T_R  \cup  S )$}.  
       }\Else{
         By \Cref{claim:GR still a expander}, $G_R$ is a $\phi$-vertex expander.\;
        \Return{$(\calG_L \cup \{G_R\}, T_L \cup S)$.}   
      }

\caption{\textsc{FastExpandersOrTerminals}$(G,c,\phi)$}
\label{alg:fast expander or terminal family}
\end{algorithm}

\begin{lemma}  \label{lem:gorig expander or terminal}
  Let $G^{\textrm{orig}}$ be the original input graph with $n$
  vertices and $m$ edges, and $c, \phi$ be parameters where $G^{\textrm{orig}}$ has min degree $\geq c$. Then, \Cref{alg:fast expander or terminal family}, when applying with the
  inputs $(G^{\textrm{orig}}, c,\phi)$ either
 outputs a separator of size $<c$ in $G^{\textrm{orig}}$ or a
$(c,\phi)$-\xorterminal $(\calG, T)$ for $G^{\textrm{orig}}$  satisfying:
\begin{enumerate}
\item $\sum_{H \in \calG} |V(H)| \leq n(1+3\phibar)^{O(\log n)}$, and
\item $|T| \leq \phibar n(1+3\phibar)^{O(\log n)}$,
\end{enumerate}
 where $\phibar = \phi \cdot m^{o(1)}$. The algorithm runs in $
 m^{1+o(1)}(1+4\phibar)^{O(\log n)} \cdot \phi^{-1}$
 time. 
\end{lemma}
 
\Cref{lem:gorig expander or terminal} with appropriate parameters immediately implies \Cref{thm:fast xorterminal pair}. The rest of this section is
devoted to proving \Cref{lem:gorig expander or terminal}. 

  \paragraph{Correctness.}   Consider the execution of  \Cref{alg:fast
    expander or terminal family} on $(G^{\textrm{orig}},  c,\phi)$. First we prove that if the condition in \Cref{line:declare not connected}
  is true at some point, then $G^{\textrm{orig}}$ is indeed not
  $c$-connected. Consider the recursion tree $\calT$ starting with
  $G^{\textrm{orig}}$. Each internal node can be represented as a
  graph $G$ along with a $\phibar$-sparse cut $(L,S,R)$, and we
  recurse on left subproblem $G_L$ and right subproblem $G_R$. The
  leaf nodes represent the instance to be a $\phi$-expander.   Let $G'$ be the first subproblem in the
  recursion tree that \Cref{line:declare not connected} is
  executed. Since $|S| < c$, we have that $G'$ is not
  $c$-connected. Let $P$ be a path from the root node to the
  subproblem $G'$ in the recursion tree $\calT$.   By applying
  \Cref{lem:dacforvc} along the path $P$, we derive that
  $G^{\textrm{orig}}$ is not $c$-connected, and we are done. From now, we assume that \Cref{line:declare not connected} is always
  false (i.e., $|S| \geq c$ every time). By the description of
  \Cref{alg:fast expander or terminal family} along with \Cref{lem:balancedorexpander,claim:GR still a expander} , the execution
  of \Cref{alg:fast expander or terminal family} is identical to that
  of \Cref{alg:expander or terminal family} (where depending on the
  conditions, $G_L$ and $G_R$ may or may not have edges inside
  $S$). By \Cref{lem:correctness expander or terminal}, \Cref{alg:fast
    expander or terminal family} outputs a $(c,\phi)$-\xorterminal for
  $G^{\textrm{orig}}$. This completes the correctness proof. 

  Next, we bound the size of the $(c,\phi)$-\xorterminal and the
  running time. We first show that the recursion depth is $O(\log
  n)$. 
  
  \begin{claim} \label{claim:logn recursion depth}
  The depth of the recursion tree $\calT$ is $O(\log n)$.  
  \end{claim}
  \begin{proof}
  Let $G$ be any internal node of the recursion tree $\calT$. Let $n_0,n_L,n_R$ be the number
  of vertices in $G,G_L$ and $G_R$, respectively. We prove that the
  size is reduced by a constant factor each time.  Let $(L,S,R)$ be the $\phibar$-sparse
  cut in $G$. If $\min\{|L \cup S|, |R\cup S|\} \geq n_0/3$, then we
  have $n_L \leq 2n_0/3$ and $n_R \leq 2n_0/3$.  Otherwise, we have that
  $G_R$ is a  $\phi$-expander and its right  subproblem becomes a leaf
  node. It remains to prove that $n_L \leq 2n_0/3$ in this
  case. Since $|S| \leq \phibar n_0$ and $|R \cup S| \geq n_0/2$, we
  have $|R| \geq n_0/4$, and thus $|L \cup S| = |L|+|S| \leq
  3n_0/4+\phibar n_0 \leq 0.85n_0$.  
\end{proof}

  \paragraph{Sizes.} If the pair $(\calG, T)$ is returned, then  \Cref{line:declare not connected} is always
  false (i.e., $|S| \geq c$ every time). We now bound the total number of vertices in all graphs in
  $\calG$ and  the size of the
  terminal set $T$.  Let $\calT$ be the recursion tree of starting with
  $G^{\textrm{orig}}$ where the leaf nodes represent $\phi$-expanders,
  and internal nodes are represented by a graph $G$ and its $\phibar$-sparse cut
  $(L,S,R)$. Fix an arbitrary internal node in the recursion tree $G$
  and $(L,S,R)$. Let $G_L$ and $G_R$ be $c$-left and $c$-right graphs
  of $G$ w.r.t. $(L,S,R)$. Let $n_0, n_L$ and $n_R$ be the number of
  vertices in $G, G_L$ and $G_R$, respectively. By definition of $G_L$
  and $G_R$, $$ n_L +
  n_R = n_0+|S|+2c \leq n_0+3|S| \leq n_0(1+3\phibar).$$ The first
  inequality follows because $n_L = |L|+|S| + c$, and $n_R =
  c+|S|+|R|$. The last inequality follows because $(L,S,R)$ is $\phibar$-sparse. Therefore, at
  level $i$, the total number of vertices is at most
  $n(1+3\phibar)^{i-1}$. Let $\ell$ be the recursion depth. We have that
  the total number of vertices at the last level is at most
  $O(n(1+3\phibar)^{\ell+1})$.  To bound the number of terminals,
  observe that it is at most $\phibar$ times total number of vertices
  generated by the recursion. Thus, $|T| = O( \phibar
  n(1+3\phibar)^{\ell+1})$. The results follow by \Cref{claim:logn
    recursion depth}. 
  
  \paragraph{Running Time.}  Consider the same recursion tree $\calT$.   We first bound the total number of edges over
  all the algorithm.  Let $m_0, m_L$ and $m_R$ be the number of
  edges in $G, G_L$ and $G_R$, respectively. We assume that  $\min\{|L \cup S|,
  |R\cup S|\} \geq n_0/3$ because otherwise $G_R$ is a $\phi$-vertex
  expander, and we only go for the left recursion. By \Cref{line:remove edges in S} and \Cref{line:remove edges in S R}, both $G_L$ and $G_R$ do not have edges
  in $E(S,S)$. Therefore,
  \begin{align*}
      m_L + m_R \leq m_0 +2(|S|c + c^2) \leq m_0 +  4|S|c \leq m_0+ 4c
    n_0\phibar \leq m_0(1 +4 \phi\bar).  
  \end{align*}
   The first inequality follows because the set of new edges are those
   in the clique in both $G_L$ and $G_R$ and biclique edges. Note that
   there is no edge inside $S$ for both graphs. The second inequality
   follows because $|S| \geq c$. The third inequality follows because
   $(L,S,R)$ is $\phibar$-sparse. The last inequality follows because
   min-degrees of $G_L$ and $G_R$ are at least $c$.

   Therefore, at level $i$, the total number of edges is at most
   $O(m(1+4\phibar)^{i-1})$.   Summing over all levels, we have that
   the total number of edges is $O(m(1+4\phibar)^{\ell+1})$ where
   $\ell$ is the recursion depth.  By \Cref{claim:logn recursion depth}, the recursion depth
   $\ell = O(\log n)$. Finally, the running time follows because each
   node $v$ in the recursion tree takes $O(m'^{1+o(1)}/\phi)$ time
   where $m'$ is the  number of edges of the graph at node $v$.

        \section{$c$-Vertex Connectivity Mimicking Networks} \label{sec:mimicking}

This section and the rest  are devoted to proving \Cref{thm:vertex sparsifier}. Let $G = (V,E)$ be a graph and $T \subseteq V$ be a terminal set. Our approach for computing a $(T,c)$-sparsifier for $G$ is to compute a $(T,c)$-\textit{covering set}.  

\begin{definition}
  A vertex set $Z \subseteq V$ is $(T,c)$-\textit{covering} if
  for all $A,B \subseteq T$ such that   $\mu(A,B) \leq c$,  $Z$ contains some min $(A,B)$-weak  separator.  
\end{definition}

Our technical result is the following. 

\begin{theorem} \label{thm:fast covering set}
Given $G = (V,E), T \subseteq V$, and a parameter $c > 0$, we can
compute $(T,c)$-covering set of size $O(|T| 2^{O(c^2)})$ in
$O(m^{1+o(1)} 2^{O(c^2)})$ time. 
\end{theorem}

We prove \Cref{thm:fast covering set} in \Cref{sec:covering set}. Once we have a $(T,c)$-covering set $Z$, we show that we can compute a $(T,c)$-sparsifier for $G$ of small size  efficiently.  

 \begin{lemma} \label{lem:reduction tc covering}
  Given a $(T,c)$-covering set $Z$ of $G$ and an integer $c > 0$, there is an $O(mc)$-time
  algorithm that outputs a $(T,c)$-sparsifier for $G$  where the
  number of vertices is $|T \cup Z|$ and the number of edges is at most $c|T \cup Z|$.  %
\end{lemma}

 \Cref{thm:fast covering set,lem:reduction tc covering} imply \Cref{thm:vertex sparsifier}. Next, we prepare necessary tools for proving \Cref{lem:reduction tc covering}. We start with the key primitive operation for constructing $(T,c)$-sparsifier for $G$, which we call \textit{vertex closure operation} in \Cref{sec:vertex closure}. Then, we show a fast offline vertex closure oracle in \Cref{sec:offline closure oracle}, which is of independent interest. Finally, we prove \Cref{lem:reduction tc covering} in \Cref{sec:proof of reduction tc covering}.   

\subsection{Vertex Closure Operation}
\label{sec:vertex closure}

We start with the definition. 
\begin{definition}
\textit{Vertex closure} operation for a vertex $v$ in graph $G$ is defined as follows: add  all edges between vertices in $N_G(v)$, and remove $v$ from $G$.  We denote $\cl(G, v)$ as the graph $G$ after closing a vertex $v$. 
\end{definition}

Intuitively, one can view vertex closure operation as an analogue of edge contraction operation.
If $x$ and $y$ are adjacent to each other via an edge $e$, after contracting $e$, there is no edge cut separating  $x$ and $y$ as they become the same vertex. 
Similarly, if $x$ and $y$ share a common neighbor via a vertex $v$, after closing $v$, there is no vertex cut separating $x$ and $y$ as they are directly connected by an edge now. 

We prove that vertex closure operation does not decrease the vertex connectivity.  
\begin{lemma} [Monotonicity Property of Vertex Closure Operation] \label{lem:monotonic vertex closure}
Let $v$ be an arbitrary vertex in $G$. Then, $\cl(G,v)$ is \textit{cut-recoverable}, i.e., if $S$ is a separator in $\cl(G,v)$, then $S$ is also a separator in $G$. In particular, $\kappa_{\cl(G,v)} \geq \kappa_G$.%
\end{lemma}  
\begin{proof}
  Suppose there is a vertex cut $(L,S,R)$ in $\cl(G,v)$. To undo contraction, we add vertex $v$ and remove edges in $\cl(G,v)$ that do not exist in $G$. Hence, $(L,S \cup \{v\},R)$ is a vertex cut in $G$.   Next, we prove that $N_G(v) \cap L = \emptyset$ or $N_G(v) \cap R = \emptyset$. Suppose otherwise. By definition of vertex contraction, there is an edge between $L$ and $R$ in $\cl(G,v)$, contradicting to the fact that $(L,S,R)$ is a vertex cut in $\cl(G,v)$.  Assume WLOG that $N_G(v) \cap L = \emptyset$ (the case $N_G(v) \cap R = \emptyset$ is similar). Since  $(L,S \cup \{v\},R)$ is a vertex cut in $G$ and $N_G(v) \cap L = \emptyset$, we have that $(L,S, R \cup \{v\})$ is a vertex cut in $G$.%
\end{proof}

It is convenient to consider a batch version of vertex closure operations. For any subset of vertices $X \subset V$, we denote $\cl(G,X)$ to be the graph $G$ after closing all vertices in $X$. Note that the resulting graph is the same regardless of the ordering of closure operations. 

\begin{proposition} \label{pro:batch closure}
For any subset of vertices $X \subset V$ in $G$, $\cl(G,X)$ is the same graph as the following transformation on $G$.  Let $S_1, \ldots, S_{\ell}$ be connected component of $G[X]$. For each $i$, we add a clique on $N_G(S_i)$, i.e., adding an edge $(u,v)$ for all pairs $(u,v)$ of vertices in $S_i$. Then, we remove $X$ from $G$.
\end{proposition}

\begin{proof}
It is enough to prove for individual connected component. Fix a connected component $S$ of $G[X]$. We prove that there is an edge $(u,v)$ for every pair of vertices $u$ and $v$ in $N_G(S)$ in $\cl(G,S)$. Fix $u,v \in N_G(S)$. Since $S$ is a connected component, there is a path $P$ from $u$ to $v$ using only vertices from $S \cup \{u,v\}$. For any sequence of closure operation in $S$, whenever an internal vertex of $P$ is closed,  we obtain a path from $u$ to $v$ using only vertices $S \cup \{u ,v\}$.  After closing the last remaining internal vertex of the path $P$, we will add an edge from $u$ to $v$. 
\end{proof}

\begin{lemma} [Offline Closure Oracle] \label{lem:static neighborhood oracle} 
We are given a graph $G = (V,E)$ with $n$ vertices and $m$ edges and a set of vertices $X \subseteq V$ and an integer $c > 0$. Let $V' = V -X$. Then, there is an $O(mc)$-time deterministic algorithm that outputs  a $(V',c)$-sparsifier for $\cl(G,X)$ with $|V'|$ vertices and at most $c|V'|$ edges.  
\end{lemma} 
We prove \Cref{lem:static neighborhood oracle} in the next section. 

\subsection{Proof of \Cref{lem:static neighborhood oracle}} \label{sec:offline closure oracle}
  We describe the algorithm, analyze running time, and argue its correctness. For any subset of vertex $S \subseteq V$, we define a $c$-\textit{partial clique} on $S$ as follows. If $|S| < c$, then we add an edge to every pair of vertices in $S$. Otherwise, we select an arbitrary $c$ vertices in $S$; for each selected vertex, we add an edge to every other vertex in $S$. 
  
  \paragraph{Algorithm.} The inputs include $G$, $X$ and $c$ as defined in the statement.
  \begin{enumerate}
      \item Starting with $G$. Let $Y_1,\ldots, Y_\ell$ be connected components of $G[X]$. For each connected component $Y_i$, we add a $c$-partial clique on $N_G(Y_i)$, and remove $Y_i$. We call the new graph $G'$.  
      \item Apply \Cref{thm:nagamochi} on $G'$ using $c$ as a parameter, and return the resulting graph. 
  \end{enumerate}
  
  \paragraph{Running Time.} It takes $O(m)$ time to find a set of connected components of $G[X]$. The running time for adding partial clique can be computed as follows. For each connected component $Y_i$, it takes $O(c|N_G(Y_i)|) = O(c \vol_G(Y_i))$ time to add a $c$-paritial clique. Therefore, the total time in this step is $O(\sum_{i} c|N_G(Y_i)|) = O(c \sum_i\vol(Y_i)) = O(mc)$. The last equality follows since $Y_1,\ldots Y_\ell$ are pairwise disjoint. The new graph has $O(mc)$ edges, and then we apply \Cref{thm:nagamochi} on the new graph which takes $O(mc)$ time.
  
  \paragraph{Correctness.}   Let $H = \cl(G,X)$.  Observe that $H$ and $G'$ have the same vertex set $V'$. 
  \begin{claim}  \label{claim:G' is V'c sparsifier}
  $G'$ is a $(V',c)$-sparsifier for $H$.
  \end{claim} 
  
  Therefore, after applying  \Cref{thm:nagamochi} on $G'$ with parameter $c$, we obtain a $(V',c)$-sparsifier for $H$ where the set of vertices is $V'$ and the number of edges is at most $|V'|c$. It remains to prove \Cref{claim:G' is V'c sparsifier}, which follows from the following two claims. 
  \begin{claim} \label{claim:v'c equivalent}
  For every $S \subseteq V'$ such that $|S| < c$, $S$ is a separator in $G'$ if and only if $S$ is a separator in $H$.  Therefore, $G'$ and $H$ are $(V',c)$-equivalent. 
  \end{claim}

  \begin{claim} \label{claim:monotone v'}
   If $S \subseteq V'$ is a separator in $G'$ such that $|S| < c$, then $S$ is a separator in $G$. Therefore, $G'$ is $c$-cut-recoverable for $G$.%
  \end{claim}
   The two claims are corollaries of the following claim. 
  \begin{claim} \label{claim:explicit cut}
   If $(L,S,R)$ is a vertex cut in $G'$ where $|S| < c$, then $(L \cup L', S, R \cup R')$ is a vertex cut in $G$ where $L' = \bigcup_{i \in U} Y_i$, and $R' = \bigcup_i Y_i - L'$ for some $U \subseteq \{1,\ldots,\ell\}$. 
  \end{claim}
  \begin{proof}
  Let $G_i$ be the graph after $i$ iterations of the first step in the algorithm. By design, we have $G_0 = G$, and $G_{\ell} = G'$. 
  We prove the following for all $i \leq \ell$:   At the end of iteration $i$, if $S$ is a separator in $G_{i}$ of size less than $c$, then $S$ is a separator in $G_{i-1}$. Let $(L,S,R)$ be a vertex cut in $G_i$ where $|S| < c$. Since we only add edges between $N(Y_i)$ in $G_{i-1}$ to get $G_i$, we have that $(L, S \cup Y_i,R)$ is a vertex cut in $G_{i-1}$. The key claim is $N(Y_i) \subseteq L \cup S$ or $N(Y_i) \subseteq S \cup R$. Indeed, suppose otherwise. Since $|S| < c$, there is one vertex $x$ (out of $c$ vertices) in $c$-partial clique such that $x \not \in S$ and we add an edge from $x$ to every vertex in $N(Y_i)$ (if $|N(Y_i)| < c$, then $x$ can be any vertex outside $S$). WLOG, assume that $x \in L$. By design, we  connect $x$ to every vertex in $N(Y_i)$ including a vertex in $R$. This implies that there is an edge between $L$ and $R$ in $G_{i-1}$, contradicting to the fact that $(L,S,R)$ is a vertex cut in $G_{i-1}$. Since $N(Y_i) \subseteq L \cup S$ or $N(Y_i) \subseteq S \cup R$, we have that either $(L \cup Y_i,S,R)$ or $(L,S, R \cup Y_i)$ is a vertex cut in $G_{i-1}$. Therefore, $S$ is a separator in $G_{i-1}$. 
  \end{proof}
  
  \Cref{claim:monotone v'} follows immediately from \Cref{claim:explicit cut}. It remains to derive \Cref{claim:v'c equivalent}.
  \begin{proof}[Proof of \Cref{claim:v'c equivalent}]
   By \Cref{pro:batch closure}, $G'$ is a subgraph of $H$, and thus a separator in $H$ is a separator in $G'$. It remains to prove that if $S$ where $|S| < c$ is a separator in $G'$, then $S$ is also a separator in $H$. Let $(L,S,R)$ be a vertex cut in $G'$ where $|S| < c$. By \Cref{claim:explicit cut},  $(L \cup L', S, R \cup R')$ is a vertex cut in $G$ where $L' = \bigcup_{i \in U} Y_i$, and $R' = \bigcup_i Y_i - L'$ for some $U \subseteq \{1,\ldots,\ell \}$. Observe that $\cl(G,X)$ will add only edges in the neighbors of $Y_i$. Since $Y_i \subseteq L'$ or $Y_i \subseteq R'$ for all $i$, $\cl(G,X)$ does not add an edge between $L$ and $R$, and thus $(L,S,R)$ is a vertex cut in $H$.
  \end{proof}
\subsection{Proof of \Cref{lem:reduction tc covering}} \label{sec:proof of reduction tc covering}
  Let $X = V - (Z \cup T)$ be a set of vertices to be closed, and let $H =
  \cl(G,X) = (V_{H},E_{H})$. First, observe  that $\kappa_H \ge \kappa_G$ by \Cref{lem:monotonic vertex closure}. Next, we claim that $H$ is $(T,c)$-equivalent
  to $G$.  If true, then we apply \Cref{lem:static neighborhood oracle} using the graph $G$ and the
  closure set $X$ to obtain a $(V_{H},c)$-sparsifier for $H$ whose number of
  vertices is $|T\cup Z|$ and number of edges is at most $c|T \cup Z|$. Since $T
  \subseteq V_H$, the sparsifier is also a $(T,c)$-sparsifier for $G$ as
  desired. By \Cref{lem:static neighborhood oracle}, the algorithm takes $O(mc)$ time.  %

  We now prove the claim.  We show that for all pair $A,B \subseteq
  T$,  we have that $\min\{c,\mu_G(A,B)\} = \min\{c, \mu_{\cl(G,X)}(A,B)\}$. Fix a pair $A,B \subseteq T$. If
  $\mu_G(A,B) \geq c$, then $\mu_{\cl(G,X)}(A,B) \geq \mu_G(A,B) \geq c$ (by 
  monotonicity property of closure operations, \Cref{lem:monotonic vertex closure}),
  and we are done.  Otherwise,  $\mu_G(A,B) < c$.  In this case, we prove that $\cl(G,X)$ contains a min
  $(A,B)$-weak separator in $G$. 
  %  %
  %
  $(T,c)$-covering set. Hence, $S$ is never closed, i.e., $S \cap X =
  \emptyset$. Since $S$ and $T$ are not closed, we conclude that $S$ is an $(A,B)$-weak separator in
  $\cl(G,X)$. In fact, $S$ is a min $(A,B)$-weak separator by
  monotonicity property of closure operations.  Therefore, $c > \mu_G(A,B)
  = |S| = \mu_{\cl(G,X)}(A,B)$. 

\section{Computing a $(T,c)$-Covering Sets} \label{sec:covering set}

This section is devoted to proving \Cref{thm:fast covering set}. We show a reduction from $(T,c)$-covering sets to $(T,c)$-\textit{\roc} \textit{\psetpair}. We set up notations.  Denote $\mu(A,B)$ to be the size of a minimum $(A,B)$-weak separator. Let $\mathcal{S}_{A,B}$ be the set of all min $(A,B)$-weak separators. Define $\mu^T(A,B) = \min_{S \in \mathcal{S}_{A,B}} |S -
T|$.  %
We say that a set $S$ \textit{covers}
another set $T$ if $T \subseteq S$. 

\begin{definition} \label{def:roc partition}
Given a graph $G = (V,E)$ with terminal set $T$, and parameter $c >0$, a \psetpair  $(\Pi = \{Z,X_1,\ldots,X_\ell\}, C \subseteq V)$ is $(T,c)$-\textit{\roc} if $\Pi$ is a partition of $V$ such that $Z$ is an $(X_i,X_j)$-separator for all $i \neq j$ and for all $A,B \subseteq T$ if $\mu(A,B) \leq c$, then one of the followings is true for some min $(A,B)$-weak separator $S$:
\begin{enumerate}
    \item $C$ covers $S$, %
    \item $\Pi$ \textit{splits} $S$, i.e., $|S \cap N_G[X_i]| \leq |S| -1$ for all $i$,
        \item $\Pi$ $T$-\textit{hits} $S$, i.e., $| (S\cap X_i) - T| \leq \mu^T(A,B)-1$ for all $i$. %
\end{enumerate}
We also say that $Z$ is the \textit{reducer} of the partition $\Pi$, and $X_1,\ldots,X_{\ell}$ is the \textit{non-reducer sequence} of the partition $\Pi$.
\end{definition}

We first mention here that there exists an almost-linear time algorithm for computing a $(T,c)$-reducing-or-covering $(\Pi,C)$ such that $|Z|,|C|$ and the total size of boundaries $\sum_i |N(X_i)|$ are small. This is formalized in  \Cref{thm:reducing-or-covering set system} below and will be proved in \Cref{sec:reducing set}.
We will use it as a key subroutine for constructing a $(T,c)$-covering set in this section.

\begin{theorem}  \label{thm:reducing-or-covering set system}
  Given a graph $G$ with terminal set $T$ and two parameters $c > 0,\phi \in (0,1)$ where $G$ has arboricity $c$ and $k = |T|$, there is an $O(m^{1+o(1)}\phi^{-4} \cdot
  2^{O(c^2)})$-time  algorithm that outputs  a
 $(T,c)$-\textit{\roc \psetpair} $(\{Z
 ,X_1,\ldots, X_\ell\},C)$ of $G$ such that
  \begin{itemize}
   \item $|Z| = O( (k + \phi n^{1+o(1)})c^2)$,
   \item  $|C| = O((k+\phi n^{1+o(1)})2^{O(c^2)})$, and
   \item  $\sum_{i=i}^\ell |N(X_i)| = O((k+\phi n^{1+o(1)})c^2)$.
   \end{itemize}
\end{theorem}

\Cref{def:roc partition} might not be very intuitive at first.
Let us explain the terms ``split'' and ``$T$-hits'' more illustratively here.  Let $\Pi$ be a partition of $V$ from \Cref{def:roc partition}.  Let $S$ be an $(A,B)$-weak separator in $G$ of size $\leq c$ for some $A,B \subseteq T$. Observe that $\Pi$ splits $S$ if and only if there exists $i$ and $x,y \in S$ such that $N(X_i)$ is $(x,y)$-separator.  If $\Pi$ does not split $S$, then $S \subseteq N[X_i]$ for some $i$. 
In this case, note that $(S\cap X_i) - T = S - T - N(X_i)$. So $\Pi$ $T$-hits $S$  means that $N(X_i)$ contains enough number of non-terminal vertices of $S$ so that $|S\cap X_i - T| = |S - T - N(X_i)| \le \mu^T(A,B)-1$. In particular, $N(X_i)$ must contain at least one vertex from $S-T$, otherwise $|S - T - N(X_i)|\ge \mu^T(A,B)$.

Now, we explain why we say that $(\Pi,C)$ is ``reducing-or-covering''. 
For $i \leq \ell$, define $G_i = G[N[X_i]]$ and $T_i = N(X_i) \cup (X_i \cap T)$. 
If the set $C$ does not cover $S$, then the partition $\Pi$ must ``reduces'' $S$ in the following sense: either $S$ is split into different $G_i$ so that $|S \cap N[X_i]| \leq c-1$ for all $i$ or there is $i$ such that $S \subseteq N[X_i]$ and $|S- T_i| \leq \mu^T_G(A,B)-1$.  In other words, $S$ either becomes a smaller cut in $G_i$ or $S$ has the same size in one smaller graph $G_i$ but contains fewer non-terminal vertices in some $G_i$ with the new terminal set $T_i$. 

The above discussion suggests a recursive algorithm for constructing a $(T,c)$-covering set from a $(T,c)$-reducing-or-covering $(\Pi,C)$ where the algorithm recurses on each $G_i$ such that, for every $A,B\subseteq T$, some $(A,B)$-weak separator is ``reduced'' in $G_i$ in the above sense.
The correctness of this recursion scheme is captured by \Cref{lem:a reduction to reducing set system} below. To state it, we need to define a notion of $(T,c,\cbar)$-\textit{covering set}. 

\begin{definition}
Given a graph $G = (V,E)$ with terminal set $T$ and parameter $c > 0$, a vertex set $Z \subseteq V$ is $(T,c,\cbar)$-\textit{covering} if
  for all $A,B \subseteq T$ such that $\mu^T(A,B) \leq \cbar$ and
  $\mu(A,B) \leq c$,  $Z$ covers some min $(A,B)$-weak separator. We
  say that $Z$ is $(T,c)$-\textit{covering} if it is $(T,c,c)$-covering. 
\end{definition}

Observe that if $\cbar > c$, then an empty set is $(T,c,\cbar)$-covering.  Also, the terminal set $T$ is a $(T,c,0)$-covering set. Another trivial $(T,c)$-covering set is $V$. 
The structural lemma below is the key for our recursive algorithm.
\begin{lemma} \label{lem:a reduction to reducing set system}
Given a graph $G$ with terminal set $T$ and $c \geq \cbar > 0$, let $(\Pi = \{Z ,X_1,\ldots, X_\ell\},C)$ be a  $(T,c)$-\roc \psetpair of $G$.  For each $i \in
 [\ell]$, define $G_i = G[N[X_i]]$, $T_i = (T \cap X_i) \cup N_G(X_i)$
 and let $Y_i$ be a $(T_i,c-1,\cbar)$-covering set in $G_i$, $\bar Y_i$
 be a $(T_i,c,\cbar-1)$-covering set in $G_i$. Then, $Z \cup C \cup \bigcup_{i \leq
   \ell} (Y_i \cup \bar Y_i) \cup T$ is $(T,c,\cbar)$-covering set in $G$.  
\end{lemma}

In \Cref{sec:proof of fast covering set}, we will show the recursive algorithm based on \Cref{lem:a reduction to reducing set system} for constructing a $(T,c)$-covering set from the subroutine from \Cref{thm:reducing-or-covering set system} for constructing a $(T,c)$-reducing-or-covering partition-set pair. This would prove \Cref{thm:fast covering set}, the main theorem of this section. 
Then, we will present the proof of \Cref{lem:a reduction to reducing set system} in \Cref{sec:proof of a reduction to reducing set system}. %

\subsection{Proof of \Cref{thm:fast covering set}} \label{sec:proof of fast covering set}

\paragraph{Algorithm.} We describe the algorithm for computing
$(T,c,\cbar)$-covering set in $G$ in \Cref{alg:tcccovering}.  Let
$\phi$ be a parameter to be selected. 

\begin{algorithm}[H]
  \DontPrintSemicolon
  \KwIn{ A graph $G = (V,E)$, a terminal set $T \subseteq V$ and parameters $c,\cbar, \phi \in (0,1)$} 
  \KwOut{ A $(T,c,\cbar)$-covering set in $G$.}
  \BlankLine
  \lIf{$\cbar > c$ \normalfont{or }$c = 0$}{\Return{$\{\}$.}}
  \lIf{$\cbar = 0$}{\Return{$T$.}}
  Apply \Cref{thm:nagamochi} on $G$ to obtain a $(V,c)$-sparsifier for $G$
  with arboricity $c$. \label{line:tccnagamochi}\;
  Let $(\{Z ,X_1,\ldots, X_\ell\},C)$ be the
  $(T,c)$-\roc \psetpair of $G$ obtained by applying
  \Cref{thm:reducing-or-covering set system} using $G,T,c,\phi$ as
  inputs.\;
   \For{$i \in [\ell]$}
  { $T_i \gets (T \cap X_i) \cup  N_G(X_i)$\;
    $Y_i \gets \textsc{CoveringSet}(G[N_G[X_i]],T_i, c-1, \cbar,\phi)$\;
    $\bar Y_i \gets \textsc{CoveringSet}(G[N_G[X_i]],T_i, c, \cbar -1 ,\phi)$\;  
   }
  \Return{ $Z' = Z \cup C \cup \bigcup_{i \leq \ell} (Y_i \cup \bar Y_i)   \cup T $. }
\caption{\textsc{CoveringSet}$(G,T, c,\cbar,\phi)$}
\label{alg:tcccovering}
\end{algorithm}

\paragraph{Correctness.} We prove that the output $Z'$ is a
$(T,c,\cbar)$-covering set for $G$ by induction on $c$ and
$\cbar$. For the base case, if $\cbar = 0$, then the terminal set is a
covering set by definition. If $c = 0$ or $\cbar > c$, then an
emptyset is a covering set by definition. Now, suppose $Y_i$ is a
$(T_i,c-1,\cbar)$-covering set for $G_i = G[N[X_i]]$, and $\bar Y_i$
is a $(T_i,c,\cbar -1)$-covering set for $G_i$. Combining with the
fact that $(\{Z,X_1,\ldots, X_{\ell}\},C)$ is a
$(T,c)$-reducing-or-covering set of $G$, \Cref{lem:a reduction to
  reducing set system} implies that $Z'$ is a $(T,c,\cbar)$-covering set in $G$.

\paragraph{Size.} We next
bound the size of $Z'$. %
Let $\tau$ be the constant in the exponent of
$2^{O(c^2)}$, and $p(n)$ be a subpolynomial factor in
\Cref{thm:reducing-or-covering set system}, respectively.  Let
$s(n,k,c,\cbar)$ be the size of output of \Cref{alg:tcccovering} where
$n$ is the number of vertices, $k$ is the number of terminals,
and $c,\cbar$ are the parameters as inputs of
\Cref{alg:tcccovering}.
For each $i \leq \ell$, denote $n_i = |N[X_i]|$ and $k_i = |T_i|$.  By
definitions, we have $s(n,k,c,0) = k$ and $s(n,k,c,\cbar) = 0$ if $c < \cbar$ or $c = 0$.    Otherwise, by \Cref{thm:reducing-or-covering set system},  $s(n,k,c,\cbar)$ satisfies
the following recurrence relation. 
\begin{align} \label{eq:recurrence snkc}
   s(n,k,c,\cbar) \leq (k+\phi n \cdot p(n)) \cdot 2^{\tau c^2} +
  \sum_{i \leq \ell} s(n_i,k_i,c-1,\cbar) + \sum_{i\leq \ell} s(n_i,k_i,c,\cbar-1),
\end{align}
where
\begin{align} \label{eq:sum ki and sum ni}
\sum_{i \leq \ell}k_i \leq (k+ n\phi \cdot p(n))c^2 \mbox{ and }
  \sum_{i \leq \ell} n_i \leq  n + \sum_{i \leq \ell} k_i.
\end{align}

By careful inductive arguments, we prove the following claim in \Cref{sec:boring induction recurrence} that

\begin{claim} \label{claim:boring induction recurrence}
  $s(n,k,c,\cbar) \leq  (k+n \phi \cdot p(n))  (4+c+\cbar)^{3(c+\cbar)} \cdot (1+\cbar) \cdot  2^{(\tau   c^2+\cbar)}$.
\end{claim}

When $\cbar = c$, \Cref{claim:boring induction recurrence} implies
that $s(n,k,c,c) \leq k \cdot 2^{O(c^2)}+n \phi  \cdot p(n) \cdot
g(c)$  for some function of $c$.  Therefore, by setting $\phi =
(10\cdot p(n) \cdot g(c))^{-1}$, we have that $s(n,k,c,c) \leq k \cdot
2^{O(c^2)} + n/10$. Now, we repeat the same algorithm for $O(\log n)$
time, we obtain the final covering set of size $O(k\cdot
2^{O(c^2)})$. 

\paragraph{Running Time.} By \Cref{line:tccnagamochi}, we can assume that $m \leq
nc$ and that $G$ has arboricity $c$. That is, for all $S \subseteq V$,
we have $|E_G(S,S)| \leq c|S|$. We next bound the total size of
subproblems. Let $m$ be the size of the input
graph. For all $i$, let $m_i$ be the size of $G[N[X_i]]$. Denote
$\tilde X_i = N[X_i]$, we have $m_i = |E_G(\tilde X_i, \tilde X_i)|
\leq c|N_G[X_i]|$.  Therefore,
\begin{align}
  \sum_i m_i \leq c \sum_i |N[X_i]| \leq c(\sum_i|X_i| +
  \sum_i|N(X_i)|) \leq cn + cn^{1+o(1)}\phi. 
\end{align}

Next, we establish the recurrence relation of the running  time. By \Cref{thm:reducing-or-covering set system}, it takes $\ot(m^{1+o(1)} \phi^{-4} 2^{O(c^2)})$ time to
compute a $(T,c)$-reducing-or-covering set system.  Let
$f(m,\phi,c,\cbar)$ be the running time of \Cref{alg:tcccovering}. We
have that  $f(n,k,c,\cbar)$ satisfies
the following recurrence relation. 
\begin{align}
   f(m,k,c,\cbar) \leq m \cdot p(m) \phi^{-4} 2^{\tau c^2} + \sum_i
  f(m_i,\phi,c-1,\cbar) + \sum_i f(m_i,\phi,c,\cbar -1), 
\end{align}
where $p(m)$ is a subpolynomial factor, $\sum_i m_i \leq m(1+\phi
p(m))$, and the base cases take linear time. By the choice of $\phi$,
and the similar inductive arguments, we can show that $f(n,k,c,\cbar) =
\ot(m^{1+o(1)}2^{O(c^2)})$. 

\subsection{Proof of \Cref{lem:a reduction to reducing set system}} \label{sec:proof of a reduction to reducing set system}

This section is devoted to proving \Cref{lem:a reduction to reducing
  set system}. For all $A, B \subseteq T$
such that $\mu(A,B) \leq c, \mu^T(A,B) \leq \cbar$, if $C$ covers some $(A,B)$-min weak separator, then we are done. Otherwise, \Cref{thm:reducing-or-covering set system} implies
that there is a min $(A,B)$-weak separator $S$ in $G$ such that $\Pi$ $T$-hits $S$ or $\Pi$ splits $S$.  For $i \in [\ell]$, define $S_i = S \cap
N[X_i], A_i = A \cap N[X_i]$ and $B_i = B \cap N[X_i]$.  We say that
$S$ is $(T,c,\cbar)$-\textit{small} in component $i$ if (1) $|S_i| \leq
c-1$, or (2) $|S_i| = c$ but $|(S \cap X_i) - T| \leq \cbar-1$.  By definition, $S$ is $(T,c,\cbar)$-small for every component $i$.

Intuitively, for all $i$, we should expect $S_i$ to be covered by either
$Y_i$ or $\bar Y_i$ since $Y_i$ and $\bar Y_i$ are
$(T,c-1,\cbar)$-covering and $(T,c,\cbar-1)$-covering in $G_i$,
respectively. However, $Y_i$ or $\bar Y_i$ may cover a different set $S'_i \neq S_i$ that has the same key properties as $S_i$ for our purpose. We show that we can combine $S'_i$ from each component to obtain a single cut $S'$ such that $S'$ is covered by $Y_i$ and $\bar Y_i$, and at the same time $S'$ is a min $(A,B)$-weak separator in $G$. We next formalize the intuition. We start with the following lemma. %

\begin{lemma} [Swapping Lemma]  \label{lem:swapping Si}
  If $S$ is a min $(A,B)$-separator in $G$, and $S$ is
  $(T,c,\cbar)$-small in component $i$, then there exists a separator $S'_i$ in
  $G_i$ that is covered by $Y_i \cup \bar Y_i$ and $(S - S_i) \cup S_i'$ is a min $(A,B)$-separator in $G$. 
\end{lemma}
 
We defer the proof of \Cref{lem:swapping Si} to the end of this section.  Our proof strategy is to apply \Cref{lem:swapping Si} for each component. %
 We are now ready to prove  \Cref{lem:a reduction to reducing set system}.

\begin{proof} [Proof of \Cref{lem:a reduction to reducing set system}] 
 Let $S$ be the min $(A,B)$-weak separator in $G$ such that $S$ is $(T,c,c_T)$-small for every component $i$ (as discussed above).  We assume WLOG that $S$ is not covered by $Z \cup C \cup \bigcup_{i \leq
   \ell} (Y_i \cup \bar Y_i) \cup T$.  We reorder the indices so that $S_1$ is $S_i$ such that $|S \cap
 X_i - T|$ is maximized over all $i$. Observe that $|S \cap X_1 - T| >
 0$ because otherwise $S$ must have been covered by $Z \cup \bigcup_{i \leq
   \ell} (Y_i \cup \bar Y_i) \cup T$. Also,
 \begin{align} \label{eq:small non terminal Si}
   \text{for all } i > 1,  |S \cap X_i - T| \leq \cbar-1.
 \end{align}
 Observe that $S$ is $(T,c,\cbar)$-small in component 1. Since $S$ is a min $(A,B)$-weak
 separator in $G$ and  $S$ is $(T,c,\cbar)$-small in
 component 1, \Cref{lem:swapping Si} implies that there exists $S_1'$ that is covered by
 $Y_1 \cup \bar Y_1$ and $S \cup S'_1 - S_1$ is a min $(A,B)$-weak separator in $G$.   %
 Therefore, we update $S \gets (S- S_1) \cup S'_1$. Observe that $S'_1$
 is covered by $Y_1 \cup \bar Y_1$, and $S$ is a min $(A,B)$-weak
 separator.

 Next, we repeat the same process for $i = 2, \ldots, \ell$ . That is,
 for each $i \in \{2,\ldots,\ell\}$, we set $S \gets (S- S_i) \cup S'_i$ where
 $S'_i$ be the separator as stated in \Cref{lem:swapping Si}. If such
 $S'_i$ always exists, then at the end of the iteration $S$ is a min
 $(A,B)$-cut that is covered by $Z \cup \bigcup_{i \leq
   \ell} (Y_i \cup \bar Y_i)$ and we are done. It
 remains to verify that at the beginning of each iteration $i$, $S$ is a min
 $(A,B)$-weak separator and $S$ is $(T,c,\cbar)$-small in component
 $i$. We use induction on the number of iterations. The first
 iteration where $i = 2$ follows from the above discussion. Now assume that at
 iteration $i$ where $2 \leq i \leq \ell-1$,  $S$ is a min
 $(A,B)$-weak separator and $S$ is $(T,c,\cbar)$-small in component
 $i$. Therefore, \Cref{lem:swapping Si} implies that there exists $S_i'$ that is covered by
 $Y_i \cup \bar Y_i$ and $(S- S_i) \cup S'_i$ is min $(A,B)$-weak
 separator. So, we update $S \gets (S-S_i) \cup S'_i$. Observe that for all $j < i$, $S_j$ can be different due to the change in the boundary vertices $\partial X_j := N_G(X_j)$, but $S_j$ is still covered by $Y_j\cup \bar Y_j \cup Z$. It remains to show that $S$ is $(T,c,\cbar)$-small in
 component $i+1$.   Since $\mu(A,B) \leq c$, we have $|S_{i+1}| \leq |S| = \mu(A,B)
 \leq c$. Also,  $|S \cap X_{i+1} - T| \leq \cbar -1$
 because of \Cref{eq:small non terminal Si} and the fact that the
 component $X_{i+1}$ has not been touched yet.\footnote{This is the reason why we showed a separate argument for component 1, otherwise it is less convenient to show that \Cref{eq:small non terminal Si} holds for component $i+1$.} Therefore, $S$ is $(T,c,\cbar)$-small in
 component $i+1$. This completes the proof. 
 
\end{proof}

It remains to prove \Cref{lem:swapping Si}. 
\begin{proof} [Proof of \Cref{lem:swapping Si}]
    The lemma is trivial if $S_i = \emptyset$. We now assume $S_i \neq
    \emptyset$. Let $S_{-i}= S - S_i$. Denote $\partial X_i = N_G(X_i)$ as a set of boundary
    vertices, and we say that $v \in V_{G_i}$ is a boundary vertex if
    $v \in \partial X_i$.  In $G_i$, we
    say that a boundary vertex $v \in \partial X_i$ is an
    $A$-\textit{proxy} if $A - A_i \neq \emptyset$, and there is an
    $(A - A_i,v)$-path in $G  - S_{-i}$ that does not use $X_i$.
    Similarly, $v$ is a $B$-\textit{proxy}  if $B - B_i \neq \emptyset$ and there is an $(B -
    B_i,v)$-path in $G  - S_{-i}$ that does not use $X_i$.  Let
    $\partial A_i$ be the set of $A$-proxy boundary vertices, and
    $\partial B_i$ be the set of $B$-proxy boundary vertices.

  Intuitively, we can think of $\partial A_i$ as a set of proxy nodes that
  represents all paths originated from $A$ outside the component $X_i$
  that could enter $X_i$ in $G - S_{-i}$, and similarly for
  $\partial B_i$.

  \begin{claim} \label{claim:new terminal set}
    $\partial A_i \cup A_i \neq \emptyset$ and  $\partial B_i \cup B_i
    \neq \emptyset$ and $S_i$ is $(\partial A_i \cup A_i,  \partial B_i \cup B_i)$-weak
    separator in $G_i$. 
  \end{claim}
  \begin{proof}
    We first prove that $\partial A_i \cup A_i \neq \emptyset$ and  $\partial B_i \cup B_i
    \neq \emptyset$. Since $S$ is a min $(A,B)$-weak separator, there must be
    an $(A,B)$-path $P$ in $G - S_{-i}$. Furthermore, every $(A,B)$-path
    in $G - S_{-i}$ must use $S_i$. By definition of
    $\partial A_i, \partial B_i$, $P$ must use
    either $A_i$ or $\partial A_i$, and $B_i$ or $\partial B_i$ in $G_i$.

    We prove the next claim. Suppose $S_i$ is not  a $(\partial A_i \cup A_i, \partial B_i \cup B_i)$-weak
    separator in $G_i$.  Let $P$ be an $(\partial A_i \cup A_i, \partial B_i
    \cup B_i)$-path in $G_i - S_i$ starting with a vertex $a$ and
    ending with a vertex $b$ (it is possible that $P$ is simply a
    vertex where $a = b$). %
    We define two paths $P_a, P_b$ as
    follows.  If $a \in \partial A_i$, then $P_a$ is an $(A -
    A_i,a)$-path in $G - S_{-i}$ that does not use $X_i$. Otherwise, $P_a$ is the vertex $a$.   If $b \in \partial B_i$, then $P_b$ is an $(B -
    B_i,a)$-path in $G - S_{-i}$ that does not use $X_i$. Otherwise, $P_b$ is the vertex $b$.
    Therefore, the path $P_a \rightarrow P \rightarrow P_b$ is an
    $(A,B)$-path in $G - S$, a contradiction. 
  \end{proof}

   Denote $A'_i = \partial A_i \cup A_i$ and $B'_i =  \partial B_i \cup B_i$. 
   
   \begin{proposition} \label{pro:mugi a'b'}
    $\mu_{G_i}(A',B') \leq c -1$ or  $\mu_{G_i}(A',B') =
   c$ and $\mu_{G_i}^{T_i}(A',B') \leq \cbar -1$
   \end{proposition}
   \begin{proof}
   Since $S$ is $(T,c,\cbar)$-small in component $i$, we
   have $|S_i| \leq c-1$ or $|S_i| = c$ but $|X_i \cap S - T| \leq
   \cbar -1$.  By \Cref{claim:new terminal set}, $S$ is a min $(A',B')$-weak separator in $G_i$. Therefore, the result follows. 
   \end{proof}
    We remark that \Cref{pro:mugi a'b'} is the key to motivate \Cref{def:roc partition}  to make this lemma work.  Since $Y_i$ is $(T_i,\cbar,c-1)$-covering set in $G_i$, and
   $\bar Y_i$ is $(T_i,\cbar-1,c)$-covering set in $G_i$, there is a
   min $(A',B')$-separator $S'_i$ in $G_i$ that is covered by $Y_i \cup
   \bar Y_i$.  It remains to prove
   the following. 

  \begin{claim} \label{claim:finalize swapping}
   $(S - S_i) \cup S'_i$ is a min $(A,B)$-weak separator in $G$. %
  \end{claim}
  \begin{proof}
  By \Cref{claim:new terminal set}, $S_i$ is a $(A',B')$-weak separator in $G_i$. Since $S'_i$ is a min $(A',B')$-weak separator in $G_i$, $|S'_i| \leq
  |S_i|$. Therefore, $|(S- S_i) \cup S'_i| \leq |S| = \mu(A,B)$. It remains to prove that $(S - S_i) \cup S'_i$ is an
  $(A,B)$-weak separator in $G$.   Since $S$ is a min $(A,B)$-weak
  separator in $G$, there is an $(A,B)$-path in $G -
  S_{-i}$. Furthermore, every $(A,B)$-path  contains a subpath $P'$ that uses
  $A'$ and $B'$ in $G[X_i]$. Since $S_i'$ is an $(A',B')$-weak separator
  in $G_i$, $P'$ must contain some vertex in $S_i'$. Therefore, $G -
  ((S- S_i) \cup S_i')$ has no $(A,B)$-path.%
   \end{proof}
\end{proof}

         \section{Computing a $(T,c)$-\Roc \Psetpair}%

\label{sec:reducing set}

 This section is devoted to proving \Cref{thm:reducing-or-covering set system}. Given a graph $G$ with terminal set $T$,  we say that a separator $T$-\textit{reduces} (or simply \textit{reduces}) the non-terminal part of another separator $S$ if it contains a non-terminal vertex of $S$. We also say that a separator \textit{splits} another separator $S$ if it is an $(x,y)$-separator in $G$ for some $x,y \in S$.  Let $\mathcal{S}$ be a set of separators in $G$.  We say that $\mathcal{S}$ $T$-\textit{reduces} (or simply \textit{reduces}) the non-terminal part of a separator $S$ if it contains a separator that reduces the non-terminal part of $S$. We also say that $\mathcal{S}$ \textit{splits} a separator $S$ if it contains a separator that splits $S$. %
 If $S$ is a min $(A,B)$-weak separator such that $|S - T| =
\mu^T(A,B)$, we say that $S$ is \emph{$T$-brittle} or just \emph{brittle} for short.

\begin{definition} \label{def:roc sspair}
 Given a graph $G$ with terminal set $T$ and $c >0$, a \ssetpair $(\mathcal{S},C \subseteq V)$, where  $\mathcal{S}$ is a set of separators in $G$, is $(T,c)$-\textit{\roc} if for all $A, B \subseteq T$ if $\mu_G(A,B) \leq c$, one of the followings is true:
 \begin{enumerate}
    \item $C$ covers some min $(A,B)$-weak separator,
    \item $\calS$ reduces the non-terminal part of some brittle min $(A,B)$-weak separator, or
    \item $\calS$ splits some brittle min $(A,B)$-weak separator. %
 \end{enumerate}
\end{definition}

\paragraph{Organization of \Cref{sec:reducing set}.} %
We first explain a divide and conquer lemmas for computing a $(T,c)$-\roc \ssetpair in \Cref{sec:divide and conquer reducing set}. In \Cref{sec:alg reducing set}, we show that divide-and-conquer lemmas naturally induce an algorithm that computes a $(T,c)$-\roc \ssetpair $(\calS,C)$ such that  total size of all separators in $\calS$ is $O(|T|c^2)$ and $|C| = O(|T|2^{O(c^2)})$.  In \Cref{sec:keeping track of reducing sets}, we explain how to modify the algorithm of \Cref{sec:alg reducing set} to output a $(T,c)$-\roc \psetpair with the properties as described in \Cref{thm:reducing-or-covering set system}. In \Cref{sec:fast impl few termminals}, we explain how to implement the algorithm in \Cref{sec:keeping track of reducing sets} in $\ot(|T|^3mc +|T| m 2^{O(c^2)})$ time, which is fast when $|T| \ll n$.  In \Cref{sec:fast impl expanders}, we show that if the input graph is a $\phi$-vertex expander, then the algorithm in \Cref{sec:keeping track of reducing sets} can be implemented in $\ot(mc\phi^{-1} + \phi^{-5}|T| 2^{O(c^2)})$; we will use the algorithm in \Cref{sec:fast impl few termminals} as a subroutine. We explain the implementation details in \Cref{sec:impl detials}.   Finally, in \Cref{sec:final reducing set},  we apply vertex expander decomposition as preprocessing and apply the algorithm in \Cref{sec:fast impl expanders} on each expander in the decomposition to compute the desired $(T,c)$-\roc \psetpair  with the properties as described in \Cref{thm:reducing-or-covering set system} in $\ot(m^{1+o(1)} \phi^{-4} \cdot 2^{O(c^2)})$ time.

\subsection{Divide-and-Conquer Lemmas} \label{sec:divide and conquer reducing set}

We prove divide-and-conquer lemma in this section. We recall the definition of Steiner cuts: A \textit{Steiner cut} $(L,S,R)$ is a vertex cut such that $L \cap T \neq \emptyset$ and $R \cap T \neq \emptyset$. It is minimum when $|S|$ is minimized.  Before stating the recursion lemma, we introduce the notion of \textit{left} and \textit{right} graphs of $G$.

\begin{definition} [Left and Right Graphs] \label{def:left right graphs}
Let $G$ be a graph with terminal set $T$ and  $(L,S,R)$ be a min
Steiner cut.   We define \textit{left graph} $G_L$ of $G$  and \textit{right
  graph} $G_R$ of $G$ as follows. We define $G_L$ as $G$ after adding clique edges to the
neighbors of $\hat R =  R \cup (S-T)$, followed by contracting\footnote{Equivalently, one can view it as replacing $\hat R$ with $t_R$ and add an edge to $t_R$ to every vertex in $N(\hat R)$.}  $\hat R$ into an arbitrary vertex $t_R \in R \cap T$. Symmetrically, we define $G_R$ as $G$ after adding clique edges to the neighbors of $\hat L = L
\cup (S -T)$, followed by contracting $\hat L$ into an arbitrary vertex
$t_L \in L \cap T$.  Let $T_L = (T \cap (L \cup S)) \cup \{t_R\}$ be the \textit{left
  terminal set}, and $T_R = (T \cap (S \cup R)) \cup \{t_L\}$ be the \textit{right terminal set}.  We define \textit{strictly left graph} $\check G_L$ as $\cl(G_L,t_R)$ and \textit{strictly right graph} $\check G_R$ as $\cl(G_R,t_L)$, and \textit{strictly left terminal set}  $\check T_L =  T \cap (L \cup S)$ and \textit{strictly right terminal set} $\check T_R = T \cap (S\cup R)$.  
\end{definition}

The divide-and-conquer lemma is the following. 
\begin{lemma} [Vertex Closure Recursion Lemma] \label{lem:closure recursion lemma} 
  Let $G$ be a graph and $T$ be a terminal set. Let $(L,S,R)$ be min Steiner cut. Define $G_L,T_L,G_R,T_R$ as in \Cref{def:left right graphs}.  If $(\calS_L,C_L)$ is a $(T_L,c)$-reducing-or-covering pair in $G_L$, and $(\calS_R,C_R)$ is a $(T_R ,c)$-reducing-or-covering pair in $G_R$, then, $(\calS_L \cup \{S\} \cup \calS_R,C_L \cup C_R)$ is a $(T,c)$-reducing-or-covering pair in
$G$.   %
\end{lemma}

By using a similar argument in the proof of \Cref{lem:closure
  recursion lemma}, we also have the following. 
\begin{lemma} \label{lem:edge case of isolating cut}
 Let $G$ be a graph and $T$ be a terminal set. Let $(L,S,R)$ be min Steiner cut. Define  strictly left graphs, strictly left terminal sets, strictly right graphs and strictly right terminal sets $\check G_L, \check T_L, \check G_R, \check T_R$ as in \Cref{def:left right graphs}.   Let $(\check \calS_L, \check C_L)$ be a  $(\check T_L, c)$-reducing-or-covering pair in $\hat G_L$, and $(\check \calS_R,\check C_R)$ be a $(\check T_R,c)$-reducing-or-covering pair in $\check G_R$, respectively. If  $S \subseteq T$, then   $(\check \calS_L \cup \{S\} \cup \check \calS_R, \check C_L \cup \check C_R)$ is a $(T,c)$-reducing-or-covering pair in $G$. 
\end{lemma} 

The rest of this section is devoted to proving \Cref{lem:closure recursion lemma}.
\begin{proof}[Proof of \Cref{lem:closure recursion lemma}]
  We show that for every $A,B \subseteq T$ such that $\mu_G(A,B) \leq c$ either $C_L \cup C_R$ covers some min $(A,B)$-weak separator or $\calS_L \cup \{S\} \cup \calS_R$ splits or reduces non-terminal part of some min brittle $(A,B)$-weak separator. Let $A, B \subseteq T$ be arbitrary two subsets of terminal set such that $\mu_G(A,B) \leq c$. If one of the following conditions hold, then we are done.  
  \begin{enumerate}
      \item $C_L \cup C_R$ covers some min $(A,B)$-weak separator,
      \item  $\calS_L \cup \{S\} \cup \calS_R$ reduces non-terminal part of some min brittle $(A,B)$-weak separator, 
      \item $S$ is an $(A,B)$-weak separator, or
      \item $S$ splits some min brittle $(A,B)$-weak separator.
  \end{enumerate}
    Assuming none of the above holds, we prove that $\calS_L$ or $\calS_R$ splits some brittle min $(A,B)$-weak separator. Let $S'$ be a brittle min $(A,B)$-weak separator, and $(L',S',R')$ be a corresponding vertex cut  where $A \subseteq L'\cup S'$ and $B \subseteq S' \cup R'$. Since $S$ does not reduce the non-terminal part of $S'$, we have $S' \cap S \subseteq T$. Since $S$ does not split $S'$, we have $S' \subseteq L \cup S$ or $S' \subseteq S\cup R$. WLOG, we assume $S' \subseteq L \cup S$. The case for $S' \subseteq S \cup R$ is similar.   Since $S$ is a Steiner cut, there is a terminal in $R$. Since $S' \subseteq L \cup S$, the terminal in $R$ must be in either $L'$ or $R'$. WLOG, we assume that it is in $L'$. That is, $T \cap L' \cap R \neq \emptyset$. We summarize our assumption for $S'$. %
  \begin{assumption} \label{ass:sprime} We have the following assumption for $S'$. 
    \begin{enumerate}
        \item   $S' \subseteq L \cup S$, 
        \item   $S' \cap S \subseteq T$,  and
        \item   $T \cap L' \cap R \neq \emptyset$.%
    \end{enumerate} 
  \end{assumption}

  From the assumption, we can intuitively expect $S'$ to be some separator in $G_L$ with terminal set $T_L$ and also expect that $\calS_L$ will split $S'$. The goal now is to prove that $\calS_L$ splits some brittle min $(A,B)$-weak separator in $G$. Note that it may not be $S'$. Define $A' := (A - R) \cup (S\cap S') \cup \{t_R\}$ and $B' := (B - R) \cup (S \cap S')$ where $t_R$ is the unique vertex in $T_L \cap R$. Note that $A', B' \subseteq T_L$.   
  
  \begin{claim} \label{claim: a-r b-r non empty}
    Both $A-R$ and $B-R$ are non-empty.
  \end{claim}  
  \begin{proof} Suppose $A - R$ is empty. Since $S' \subseteq L \cup S$, we have $S' \cap R = \emptyset$, and thus $A \subseteq L' \cap R$ and  $N_G(L' \cap R) \subseteq  S$. By definition of $S'$, we have $B \subseteq S'\cup R'$, and thus $B \cap (L' \cap R) = \emptyset$. Therefore, $S$ is an $(A,B)$-weak separator, a contradiction. The argument for $B -R$ is similar.  
  \end{proof}

  \begin{claim} \label{claim:weak A'B' in G_L is weak AB in G}
   Any $(A',B')$-weak separator in $G_L$ is an $(A,B)$-weak separator in $G$. Hence, $\mu_{G_L}(A',B') \geq \mu_G(A,B)$.
  \end{claim} 
  \begin{proof}
   We first set up notations. For any path $P$ in $G$ and a vertex set $X  \subseteq V$, we
     denote $\cl(G,P,X)$ to be the path $P$ after applying vertex
     closure operations on $X$ in $G$. By definition of $G_L$, before replacing
     $\hat R$ with $t_R$ we add clique edges to the neighbors  of
     $\hat R$. Therefore, any path in $P$ in $G$ becomes $\cl(G,P,\hat
     R)$ in $G_L$. 

    Let $S''$ be any $(A',B')$-weak separator in $G_L$. Suppose there is an $(A,B)$-path $P$ in $G - S''$. Our goal is to show that $\cl(G,P, \hat R)$ is an $(A',B')$-path in $G_L - S''$, which is a contradiction. Let $a \in A$, and $b \in B$ be the starting and ending vertices of path $P$, respectively. We consider three cases.  
    \begin{enumerate}
        \item The first case is $b \in B \cap R$. By \Cref{ass:sprime}(3),  $T \cap L' \cap R$ and $T \cap R' \cap R$ are both non-empty. Since $(L,S,R)$ is a min
   Steiner cut, we have  $L' \cap S = R' \cap S = \emptyset$, and thus $P$ must contain a vertex in  $S \cap S'
   \subseteq T$. Therefore,  $\cl(G,P,\hat R)$ contains a vertex in $S \cap S'$ in $G_L - S''$, a
   contradiction. 
        \item The second case is $a \in A- R, b \in B- R$. In this case, $a \not \in \hat R$ and $b \not \in \hat R$ since $a,b \in T$ and $\hat R$ contains only vertices in $R$ and $S - T$. Therefore, $\cl(G,P,\hat R)$ starts with $a$ and ends with $b$, and thus it is an $(A',B')$-path in $G_L - S''$, a contradiction. 
        \item The final case is $a \in A \cap R, b \in B - R$.  In this case, $a \in \hat R$ and $b \not \in \hat R$, and thus $\cl(G,P,\hat R)$ starts with a vertex in $N_G(\hat R)$ and ends at $b$. By definition of $G_L$, there is an edge from $t_R$ to every vertex in $N_G(\hat R)$. Therefore, the path obtained by appending $t_R$ with $\cl(G,P,\hat R)$ is an $(t_R,B-R)$-path in $G_L - S''$, a contradiction. 
    \end{enumerate} \qedhere

  \end{proof}
 
  \begin{claim} \label{claim:S' min A'B' weak sep}
    $S'$ is a min $(A',B')$-weak separator in $G_L$. Hence, $\mu_{G_L}(A',B') \leq |S'| = \mu_G(A,B) \leq  c$.
  \end{claim}
  \begin{proof}
   We claim that $N_G(\hat R) \subseteq (L\cup S) \cap (L'
  \cup S')$. If true, then  $S'$ is an $(A',B')$-weak separator in $G_L$ (since $B- R$ and $A-R$ are non-empty by \Cref{claim: a-r b-r non empty}), and thus 
 $\mu_{G_L}(A',B') \leq |S'| = \mu_G(A,B)$. We now prove the claim.  By \Cref{ass:sprime}(3),  $R \cap L' \cap T \neq \emptyset$. Since $(L,S,R)$ is a min Steiner cut in $G$, we have $S \cap R' = \emptyset$.   By \Cref{ass:sprime}(2),  $S' \subseteq L \cup
  S$, and thus we have $S' \cap R = \emptyset$.   By \Cref{ass:sprime}(1), $S \cap S'
  \subseteq T$. Since $\hat R = R \cup (S-T)$, the claim follows. 
  
  Next, we prove that $S'$ is
 a min $(A',B')$-weak separator in $G_L$. Let $S''$ be
  a min $(A',B')$-weak separator in $G_L$. By \Cref{claim:weak A'B' in G_L is weak AB in G}, we have $S''$
  is also an $(A,B)$-weak separator in $G$, and thus  $\mu_{G}(A,B) \leq
  |S''| = \mu_{G_L}(A',B')$. Therefore, $\mu_G(A,B) = \mu_{G_L}(A',B')
  = |S'|$. So, $S'$ is a min $(A',B')$-weak separator in $G_L$.
  \end{proof}
  
  \begin{claim} \label{claim:TL brittle T brittle}
 Any $T_L$-brittle min $(A',B')$-weak separator in $G_L$ is a $T$-brittle min $(A,B)$-weak separator in $G$. 
  \end{claim}
  \begin{proof}
   Let $S''$ be a min $T_L$-brittle min $(A',B')$-weak separator in $G_L$. We first show that $S''$ is a min $(A,B)$-weak separator in $G$.   By \Cref{claim:weak A'B' in G_L is weak AB in G}, $S''$ is an $(A,B)$-weak separator in $G$. So, $|S''| \geq \mu_G(A,B)$. By \Cref{claim:S' min A'B' weak sep}, $|S''| = \mu_{G_L}(A',B') \leq \mu_G(A,B)$. Therefore, $|S''| = \mu_G(A,B)$. Next, we show that  $|S'' - T| = \mu_G^T(A,B)$. If true, then $S''$ is a $T$-brittle min $(A,B)$-weak separator in $G$.   Since $S'$ is a $T$-brittle weak $(A,B)$-weak separator in $G$, $|S'-T| = \mu_G(A,B)$. Since $S''$ is a $T_L$-brittle weak $(A',B')$-weak separator in $G_L$, $|S'' - T_L| = \mu_{G_L}(A',B')$. 
    
    To do so, we show that $S'-T_L = S' -T$ and $S'' - T_L = S'' - T$. Indeed, since $S' \subseteq L \cup S$ and $S' \cap S \subseteq T$, we have that $S'$ does not contain any terminal in $T \cap R$, and thus  $S' - T_L = S' - T$. Since $S'' \subseteq V(G_L) =  L \cup (S - T) \cup \{t_R\}$ where $t_R$ is the unique terminal in $T \cap R$ and  additional terminals from $T_L$ are all in $R - \{t_R\}$,  $S'' - T_L = S'' - T$.  
    
    It remains to prove that $\mu_{G_L}^T(A',B') = \mu_G^T(A,B)$. If true, then we have 
    \begin{align*}
        |S'' - T| = |S'' - T_L| = \mu_{G_L}^T(A',B') = \mu_G^T(A,B).  
    \end{align*} Since $S'$ is a min $(A',B')$-weak separator in $G_L$ (\Cref{claim:S' min A'B' weak sep}), we have 
  \begin{align*}
  \mu_{G_L}^T(A',B') \leq |S' - T_L|  = |S' - T| =  \mu_{G}^T(A,B).   
  \end{align*}
    Since $S''$ is a min $(A,B)$-separator in $G$, we have 
    \begin{align*}
          \mu_{G_L}^T(A',B') = |S'' - T_L| = |S'' - T| \geq \mu_G^T(A,B). 
    \end{align*}
  \end{proof}
   
  \begin{claim} \label{claim:S_L split brittle}
  $\calS_L$ splits some $T_L$-brittle min $(A',B')$-weak separator in $G_L$
  \end{claim}
  \begin{proof}
     The following facts imply that  $\calS_L$ must split some $T_L$-brittle min $(A',B')$-weak separator in $G_L$ as desired.
   \begin{itemize}
       \item $(\calS_L,C_L)$ is a $(T_L,c)$-reducing-or-covering pair in $G_L$,
       \item $\mu_{G_L}(A',B') \leq c$,  
       \item $C_L$ does not cover any min $(A',B')$-weak separator in $G_L$,
       \item $\calS_L$ does not reduce non-terminal part of any min $T_L$-brittle $(A',B')$-weak separator in $G_L$.
   \end{itemize}
    The first item follows from the definition. We next show the second item. By \Cref{claim:S' min A'B' weak sep}, we have $\mu_{G_L}(A',B') = |S'| \leq c$. We prove the third item.  Suppose $C_L$ covers a min $(A',B')$-weak separator $Y$ in $G_L$. By \Cref{claim:weak A'B' in G_L is weak AB in G}, $Y$ is also a min $(A,B)$-weak separator in $G$, contradicting that $C_L$ does not cover any min $(A,B)$-weak separator in $G$.  Finally, the last item follows  from \Cref{claim:TL brittle T brittle}.
  \end{proof}
   \Cref{claim:S_L split brittle} and \Cref{claim:TL brittle T brittle} imply that $\calS_L$ splits some $T$-brittle min $(A,B)$-weak separator in $G$ as desired. This completes the proof of \Cref{lem:closure recursion lemma}.
 \end{proof}

\subsection{Computing a $(T,c)$-\Roc \Ssetpair} \label{sec:alg reducing set} %

\Cref{lem:closure recursion lemma} naturally gives us a recursive
algorithm for computing $(T,c)$-reducing-or-covering pair as described in \Cref{alg:sparsify1}.  

\paragraph{Algorithm.} Let $G$ be the input graph and $T$ be its terminal set. If
$T$ is small enough, we can output $(\emptyset, C)$ where $C$ is  the union of a small min $(A,B)$-weak separator for each $A,B\subseteq T$.   Let $(L,S,R)$ be the min Steiner cut. If it does not exist or $|S| \geq c+1$, we are done.  We say that $(L,S,R)$ is \textit{balanced} if $\min\{ |L \cap T|, | R\cap
T | \} \geq 4c^2$, \textit{isolating} if $\min\{ |L\cap T| , |R\cap T|\} =
1$. Otherwise, $(L,S,R)$ is \textit{unbalanced non-isolating}. Given a
terminal $t$, $(L,S,R)$ is $t$-\textit{isolating} if $t \in L$ and $T
- \{t\} \subseteq S \cup R$. For each terminal $t \in T$, define $\isolate_G(t)$ to be the size
of minimum $t$-isolating separator in $G$
\begin{definition} A vertex cut $(L,S,R)$ is \textit{good}
  $t$-isolating if it is a $t$-isolating cut min Steiner
  cut and $\isolate_{G_R}(t) > \isolate_{G}(t)$  or $S \subseteq T$.
\end{definition}

For the general case, we apply \Cref{lem:closure recursion lemma}. For efficiency point of view, we will apply \Cref{lem:edge case of isolating cut} if $S \subseteq T$ and $(L,S,R)$ is $t$-isolating. The full algorithm is described in \Cref{alg:sparsify1}. %

\begin{algorithm}[H]
  \DontPrintSemicolon
  \KwIn{ A graph $G = (V,E)$, a terminal set $T \subseteq V$ and
    parameters $c >0$} 
  \KwOut{ A $(T,c)$-\roc \ssetpair in $G$.}
  \BlankLine
  \If{ $|T| \leq  4c^2$}
    { 
     Let $C$ be the union of a min-$(A,B)$-weak
        separator for all $A,B \subseteq T$ such that $\mu_G(A,B) \leq
        c$.  \label{line:base 4c2}\;
    \Return{ $(\{\},C)$}
    }
    \lIf{$G[T]$ \normalfont{is} a complete graph}{\Return{$(\{\},T).$}}
    Let $(L,S,R)$ be a min Steiner cut. \label{line:min steiner cut slow}\;
    \lIf{ $|S| \geq c+1$}{\Return{$(\{\},T).$}}
    \If{$(L,S,R)$ \normalfont{is} a $t$-isolating cut}{
        Replace $(L,S,R)$ with a good $t$-isolating
        cut. \label{line:maximal isolating cut}\;
      }
     
        Let $G_L,  G_R,  T_{L},  T_{R}$ be the left and right graphs, the terminal sets as defined in \Cref{def:left right graphs} \label{line: construct GL GR}\;
      \If{$S\subseteq T$ \normalfont{ and } $(L,S,R)$ is
        $t$-isolating} {
         Let $G_L,  G_R,  T_{L},  T_{R}$ be the strictly left and right graphs, the strict terminal sets as defined in \Cref{def:left right graphs} \label{line: construct GL GR check}\;

       }
      $(\calS_L, C_L) \gets \textsc{ROC}(G_L, T_{L} , c)$\;
      $(\calS_R, C_R) \gets \textsc{ROC}(G_R, T_R, c)$\;
      \Return{ $(\calS_L \cup \{S\} \cup \calS_R, C_L\cup C_R)$ }
\caption{\textsc{ROC}$(G,T, c)$}
\label{alg:sparsify1}
\end{algorithm}

\paragraph{Analysis.}  We  describe the recursion tree of \Cref{alg:sparsify1} and its structure. Each internal
node in the recursion tree can be represented as $(G,T)$ and its left
and its right children are $(G_L,T_{L})$ and $(G_R,T_{R})$ respectively where $G_L$ and $G_R$ are left and right graphs of $G$ with respect to the min Steiner cut $(L,S,R)$ found in the current subproblem $(G,T)$. If $S \subseteq T$ and $(L,S,R)$ is $t$-isolating, then we further close the unique terminal $t_R$ in $G_L$ and another unique terminal $t_L$ in $G_R$ to obtain strictly left and strictly right graphs respectively. 

We state the key properties of this algorithm into two lemmas. 
\begin{lemma} \label{lem:recursion tree: total number of internal nodes}
 Given a graph $G$ with terminal set $T$ and a parameter $c > 0$,  then over all the
 executions of \Cref{alg:sparsify1}, the number of balanced min Steiner cuts, isolating, and
 unbalanced non-isolating are at most $O(|T|/c),
   O(|T|c),$ and $O(|T|c)$, respectively.   
 \end{lemma}
 
\begin{lemma} \label{lem:abstract reducing set size}
 \Cref{alg:sparsify1} returns a $(T,c)$-\roc \ssetpair $(\calS,C)$ such that $|\bigcup_{S \in \calS}S| = O(|T|c^2)$ and $|C| = O(|T| 2^{O(c^2)})$.  
\end{lemma}

We explain how to obtain $(T,c)$-\roc \psetpair from \Cref{alg:sparsify1} in \Cref{sec:keeping track of reducing sets}. We will describe fast implementation when the number of terminal is small in \Cref{sec:fast impl few termminals} and when the input graph is a vertex expander in \Cref{sec:fast impl expanders}. The final algorithm for computing $(T,c)$-\roc \psetpair is described in \Cref{sec:final reducing set}.  The rest of the section is devoted to proving \Cref{lem:recursion tree: total number of internal nodes} and \Cref{lem:abstract reducing set size}. We first start with useful properties of left and right graphs.

\begin{proposition} [Monotonicity of Left and Right
  Graphs] \label{pro:monotonicity GLGR}
  Every separator in $G_L,G_R, \check G_L, \check G_R$ (\Cref{def:left right graphs}) is a separator in
  $G$. %
\end{proposition}
\begin{proof} 
 We prove the results for the left graph $G_L$. The proof for $G_R$
 is similar. Let $t_R$ be an arbitrary terminal in $R \cap T$. Since
 $(L,S,R)$ is a min Steiner cut in $G$, there must be a path from $t_R$
 to every vertex in $S$ in $G[S \cup R]$ (otherwise, a subset of $S$ is a Steiner cut,
 contradicting the minimality of $S$).  Let $\tilde R = \hat R -
 \{t_R\}$. Since there is a path from $t_R$ to every vertex in $S$ without
 using $L$, we have $\cl(G,\tilde R) \subseteq G_L$.  Next, we prove
 that every separator in  $G_L$ is a separator in $G$. Let $S'$ be a separator in
 $G_L$. Since $\cl(G, \tilde R) \subseteq G_L$, $S'$ is also a
 separator in $\cl(G, \tilde R)$. By monotonicity  of vertex closure operation (\Cref{lem:monotonic vertex closure}), $S'$ is also
a separator in $G$. %
\end{proof}

We are now ready to prove \Cref{lem:recursion tree: total number of internal nodes}.

\begin{proof} [Proof of \Cref{lem:recursion tree: total number of internal nodes}] We define the potential function that we use to keep track of progress
of the recursion tree. For each terminal $t \in T$, let $\psi_G(t) =  \max\{0, c +1 -
\isolate_G(t)\}$ if $G$ is not in the base case in the algorithm (i.e., it does not terminate at line 1, 2 and 4 of the algorithm), and
$\psi_G(t) = 0$ otherwise. Define $\psi_G(T) = \sum_{t \in T}\psi_G(t)$. Let $\pi_G(T) = \max\{0, |T| - 4c^2\}$.  At any time, we have a collection of
graphs and their terminal set  $\mathcal{G} = \{(G,T)\}$ at leaf nodes
(that are not necessarily the base case at the moment)  in the recursion tree. Initially, we have only one node which is the
input graph and its terminal set. We keep track of the progress via
the following potential function
$$\Pi(\calG) = \sum_{(G,T) \in \calG} \pi_G(T) + \psi_G(T).$$

Initially, $\calG$ contains only the input graph and its terminal
set. Thus, $\Pi(\calG) \leq |T| + |T|c = O(|T|c)$. At any time,
$\Pi(\calG) \geq 0$. At any time,  the recursion tree changes at a
leaf node that is not a base. Suppose the leaf node represents a graph
$G$ with $k$ terminals where $k \geq 4c^2$.   There are three cases.

\paragraph{Case 1: Balanced.}  We prove that the number of balanced
min Steiner cut is at most $O(|T_{\text{input}}|/c)$ where
$T_{\text{input}}$ is the terminal set at the root of the recursion
tree. In this case, $(G,T)$ creates
two subproblems $(G_L, T_L)$ and $(G_R, T_R)$  %
 with $k_1$, and $k_2$ terminals respectively where
$\min\{k_1, k_2\} \geq 4c^2$. By definition, the only duplicate
terminals happens at the min Steiner cut $S$ where $|S| \leq c$, and
thus we have $k_1 + k_2 \leq k+c$.  The change in the potential
$\Pi(G)$ happens at the term $\pi_G(T) + \psi_G(T)$ and its two
new subproblems $\pi_{G_L}(T_L), \pi_{G_R}(T_R), \psi_{G_L}(T_L), \psi_{G_R}(T_R)$. That is the
change due to $\pi$ is $$\pi_G(T) - \pi_{G_R}(T_{R}) - \pi_{G_L}(T_{L}) = (k- 4c^2) - ( (k_1 - 4c^2) + (k_2 -4c^2) ) \geq -c + 4c^2
\geq 3c^2,$$ and the change
due to $\psi$ is $$ \psi_G(T) - \psi_{G_L}(T_L) - \psi_{G_R}(T_R)
 \geq (k - k_1 -k_2 - 2) c \geq -2c^2.$$ Intuitively, $\psi$ increases
 because there are terminals in $S\cap T$ that are replicated. Each
 terminal can have $\psi$ value by at most $c$. Since  $|S \cap T|
 \leq c$, the total increases due to $\psi$ is at most $O(c^2)$. Therefore, the net drop in
 potential is at least $3c^2 - 2c^2 = c^2$.  Since the initial potential is at most
 $O(|T_{\text{input}}|c)$, the number of balanced min Steiner cuts is
 at most 
 $O(|T_{\text{input}}|/c)$.  
 
\paragraph{Case 2: Isolating.} In this case, we obtain a min Steiner
cut $(L,S,R)$ that is good $t$-isolating for some $t$ and $|S| \leq
c$. WLOG, $L
\cap T = \{t\}$. Thus, $G_L$ can have at most $c+2$ terminals. So, $G_L$
becomes a base case and the thus the potential for $G_L$ is
zero. Hence, the change of the potential function due to $\pi$
is 
$$ \pi(G,T) - \pi(G_L,T_L ) - \pi(G_R,T_R)
 = 0.$$  That is, there is no change in the potential due to
$\pi$. Next, we bound the change of potential due to $\psi$. Observe 
that $G_R$ contains the same number of terminals, and $t_L = t$ where $t_L$ is the unique terminal in $T \cap L$ in $G_R$. For each of 
terminal $t' \in T_R - \{t_L\}$, we have $\isolate_{G_R}(t') \geq
\isolate_G(t')$ because of monotonicity of left and right graphs of
$G$ (\Cref{pro:monotonicity GLGR}).  It remains to bound the change of
$\psi$ on $t = t_L$. There are two cases:  
\begin{itemize}
\item The first case is when $S \subseteq T$.  In this case, $t_L$ is
  closed because we recurse on strictly right graph $\check G_R$ of $G$. Therefore, the terminal $t_L$  disappears in the remaining graph.     Since $\isolate_G(t_L) > 0$
  and $t_L$ disappears in $G_R$,  the drop in potential due to $\psi$ is at
  least 1. 
\item Otherwise,  since $(L,S,R)$ is a good $t$-isolating cut in $G$, and $S
- T \neq \emptyset$,  we have $\isolate_{G_R}(t) > \isolate_G(t)$ and the change due to $\psi$ is
$$ \psi_G(T) - \psi_{G_L}(T_L \cup \{t_R\}) - \psi_{G_R}(T_R \cap
\{t_L\}) \geq 1. $$     
\end{itemize}

 Since the initial potential is at most
 $O(|T_{\text{input}}|c)$, the number of good $t$-isolating min Steiner cuts is
 at most   $O(|T_{\text{input}}|c)$.
 
\paragraph{Case 3: Unbalanced Non-isolating.} In this case, we obtain
a min Steiner cut $(L,S,R)$ such that $\min\{|L \cap T|, |R \cap T| \}
\in [2,4c^2-1]$.  We assume WLOG $|L \cap T| \leq |R\cap
T|$. So, the subproblem $(G_L,T_L)$ becomes a base case,
and thus $\pi_{G_L}(T_L) = \psi_{G_L}(T_L)
= 0$. Observe that the number of terminals $G_R$ is strictly smaller
than that of $G$. So, $\pi_G(T) - \pi_{G_R}(T_R) \geq
1$. By \Cref{pro:monotonicity GLGR} (monotonicity of left and right
graphs of $G$), we have
$\isolate_{G_R}(t') \geq \isolate_G(t')$ for all $t' \in T \cap (S
\cup R)$. Furthermore, $t_L$ corresponds to one of the terminal in $L \cap
T$, and thus $\isolate_{G_R}(t_L) \geq \isolate_G(t)$ for any $t \in L
\cap T$.  Therefore, $\psi_G(T) - \psi_{G_R}(T_R)  \geq
0.$ Since the initial potential is at most
 $O(|T_{\text{input}}|c)$ where $T_{\text{input}}$ is the original input terminal set, the number of unbalanced non-isolating min
 Steiner cuts is  at most   $O(|T_{\text{input}}|c)$.
\end{proof}

Next, observe that if the min Steiner cut is larger than $c$, then we are done. 
\begin{proposition} \label{pro:min Steiner cut}
  Let $(L,S,R)$ be a minimum Steiner cut. If $|S| \geq c+1$, then the terminal set $T$ is $(T,c)$-covering. 
\end{proposition}
\begin{proof}
Since min Steiner cut is at least $c+1$, every $(A,B)$-weak separator of size at most $c$ must be either $A$ or $B$. Since $A \cup B \subseteq T$, we have that  $T$ is a $(T,c)$-covering set. 
\end{proof}

We are now ready to prove \Cref{lem:abstract reducing set size}. 
\begin{proof}[Proof of \Cref{lem:abstract reducing set size}]
We use induction on the recursion tree from leaves to the
root. 
We consider the base case (leaf nodes). If $|T| \leq 4c^2$, then
clearly the union is of all min $(A,B)$-weak separator for all $A,B
\subseteq T$ is $(T,c)$-covering. If min Steiner cut is at least
$c+1$ or $G[T]$, then \Cref{pro:min Steiner cut} implies that $T$ is $(T,c)$-covering
set. For the internal nodes, it follows immediately from
\Cref{lem:closure recursion lemma} if $S \not \subseteq T$, and by \Cref{lem:edge case of isolating cut} otherwise. 

Next, we bound the sizes of $(T,c)$-\roc \ssetpair. The total size of all separators in $\calS$ is bounded by the total size of min Steiner cuts in the internal nodes of the recursion tree. By
\Cref{lem:recursion tree: total number of internal nodes}, the number
of internal nodes is $O(|T|c)$,  and thus, the number of leaves is also at
most $O(|T|c)$. Since each Steiner cut has size at most $c$, the total
contribution of the size due to internal nodes is at most $O(|T|c^2)$. The size of $C$ is bounded by the total size of cuts in the base cases. Since
each leave represents a base case, the total contribution due to leaf nodes are at most $O(|T|c \cdot c \cdot 3^{4c^2}) = O(|T|\cdot 2^{O(c^2)})$.  
\end{proof}

\subsection{A Partition from the Recursion Tree} \label{sec:keeping track of reducing sets}

 One possible way to obtain $(T,c)$-\roc \psetpair from a $(T,c)$-\roc \ssetpair is as follows:
\begin{observation} \label{obs:ez conversion}
  Given a $(T,c)$-reducing-or-covering \psetpair $(\mathcal{S},C)$, define $Z$ to be the union of all separators in $S$, and $X_1, \ldots, X_\ell$ to be connected components of $G - Z$. Then, $(\{Z,X_1,\ldots,X_{\ell}\},C)$ is $(T,c)$-\roc \psetpair.
\end{observation}
\begin{proof}
Fix $A,B$ such that $\mu(A,B) \leq c$. If $C$ covers some min $(A,B)$-weak separator, or $\mathcal{S}$ splits some brittle min $(A,B)$-weak separator, then we are done. We now assume otherwise. By \Cref{def:roc sspair}, there exists  a brittle min $(A,B)$-weak separator denoted as $S$ in $G$ whose non-terminal part is reduced by  $\mathcal{S}$. Since $S$ is not split by $\mathcal{S}$,  $S \subseteq N[X_i]$ for some $i$.  Since $S$ is reduced by $\mathcal{S}$, there is a non-terminal vertex in $S$ that is in $Z$. Therefore, $| (S\cap X_i) - T | \leq |S - T| -1 \leq \mu^T(A,B)-1$. 
\end{proof}
However, the total neighbors $\sum_i|N(X_i)|$ in \Cref{obs:ez conversion} can be too large (because of overlapping neighbors), and thus the condition 3 of \Cref{thm:reducing-or-covering set system} does not hold.  To handle this issue, we show  that all the base cases in the recursion of the algorithm for computing $(T,c)$-\roc \ssetpair (\Cref{sec:alg reducing set}) correspond to components $X_1, \ldots, X_{\ell}$ whose total neighbors is small.  We now make the statement precise.

\begin{definition}  \label{def:output partition}
Let $G^{\textrm{orig}}$ be the original input graph at the root of the recursion tree of \Cref{alg:sparsify1}. Given a recursion tree at the end of \Cref{alg:sparsify1}, we define $X_1, \ldots X_{\ell}$ where $\ell$ is the number of leaf nodes in the recursion tree as follows. For each leaf node $i$, we start with $X_i = V(G^{\textrm{orig}})$. Then, we move along the path from root to leaf $i$ and keep filtering the vertex set in the following sense. At the current node $v$, let  $(L_v,S_v,R_v)$ be a min Steiner cut obtained in the subproblem at $v$.  If we go left, then we set $X_i \gets X_i \cap L_v$. Otherwise, we set $X_i \gets X_i \cap R_v$. We repeat until we reach the leaf node. 
\end{definition} 

\begin{lemma}  \label{lem:output partition}
 Let $(\calS,C)$ be a $(T,c)$-\roc obtained from \Cref{alg:sparsify1} with $G^{\textrm{orig}}, T, c$ as inputs, and let $X_1,\ldots,X_\ell$ be vertex sets of $G^{\textrm{orig}}$ according to \Cref{def:output partition}. Let $Z = \bigcup_{S \in \calS} S$. Then, the \psetpair $((Z, X_1,\ldots, X_{\ell}),C)$ is $(T,c)$-\roc for $G^{\textrm{orig}}$ where
   \begin{itemize}
       \item $|Z| = O(|T|c^2)$,
       \item $|C| = O(|T|\cdot 2^{O(c^2)})$, and
       \item $\sum_i |N_{\Gorig}(X_i)| = O(|T|c^2)$. 
   \end{itemize}     
\end{lemma}

The rest of the section is devoted to proving \Cref{lem:output partition}. We will discuss efficient implementation in next subsections. 

 \paragraph{Analysis.} Our goal is to establish the following claims, which imply \Cref{lem:output partition}. 
 \begin{claim} \label{claim:zx partition}
  $(Z,X_1,\ldots, X_{\ell})$ forms a partition of $V(G^{\textrm{orig}})$ such that $Z$ is an $(X_i,X_j)$ for all $i \neq j$.
 \end{claim}
 
 \begin{claim} \label{claim:reducing covering set}
The \psetpair  $((Z,X_1,\ldots,X_{\ell}),C)$ is an $(T,c)$-\roc.%
 \end{claim}
 
 \begin{claim} \label{claim:small total boundaries}
  $\sum_i |N_{\Gorig}(X_i)| = O(|T|c^2)$, $|Z| = O(|T|c^2)$ and $|C| = O(|T|\cdot 2^{O(c^2)})$. 
 \end{claim}

 For the purpose of analysis, we also associate $X \subseteq V(G^{\textrm{orig}})$ to each node in the recursion tree. That is, the input to the subproblem is of the form $(G,T,c,X)$.  We call $X$ $\textit{core set}$ of the subproblem. %
 Recall that $G,T,c$ are the same as in the input for \Cref{alg:sparsify1}. We define $X$ for each node in the recursion tree. Initially, $X = V(G^{\textrm{orig}})$ at the root node. At the current node with $(G,T,c,X)$, we obtain a min Steiner cut $(L,S,R)$ of $G$ with terminal set $T$. According to \Cref{alg:sparsify1}, we recurse on left graph $G_L$ of $G$ with terminal set $T_L$ and right graph $G_R$ of $G$ with terminal set $T_R$ (if $(L,S,R)$ is a isolating Steiner cut and $S \subseteq T$, then we define $G_L$ to be strictly left graph and $G_R$ to be strictly right graph of $G$ instead).   In this analysis, we recurse left  with $(G_L,T_L,c,X_L)$ as inputs and right with $(G_R,T_R,c,X_R)$ as inputs where $X_L := X \cap L$ and $X_R := X \cap R$. Observe that $X_L \subseteq V(G_L)$ and $X_R \subseteq V(G_R)$, respectively.

 \begin{lemma} \label{lem:S is LR sep}
   If $(G,T,c,X)$ is the current node, and $(L,S,R)$ is the min Steiner cut found in the algorithm, then $S$ is an $(L,R)$-separator in $\Gorig$. 
 \end{lemma}
 \begin{proof}
   Observe that the graph $G$ in the recursion tree is obtained by a sequence of transformations in \Cref{def:left right graphs} starting from  $G^{\textrm{orig}}$. By \Cref{pro:monotonicity GLGR}, $S$ is also a separator in $G^{\textrm{orig}}$. Since $L$ and $R$ are on different sides of the separator $S$ in the graph $G$, $S$ must separate $L$ and $R$ in $\Gorig$ as well.%
 \end{proof}

We are now ready to prove each of the claims above. 
\begin{proof} [Proof of \Cref{claim:zx partition}] 
We prove that $(Z,X_1,\ldots,X_\ell)$ is a partition of $V(\Gorig)$. We show that $X_i \cap X_j = \emptyset$ for all $i\neq j$. Suppose there are $i$ and $j$ such that $X_i \cap X_j \neq \emptyset$. Let $v \in X_i \cap X_j$. Let $p$ be the LCA (longest common ancestor) node of leaf $i$ and leaf $j$ in the recursion tree. Let $X_p$ be the core set of the subproblem at $p$. So, $v \in X_p$. Since $p$ is not the leaf, we must obtain a min Steiner cut $(L,S,R)$ at $p$. Since $v \in X_i$ and $v \in X_j$, we have $v \in L \cap R$, so $L \cap R \neq \emptyset$, a contradiction.  Next, we prove that $Z \cap X_i = \emptyset$ for all $i$. By design, $X_i$ is obtained by filtering from $V(\Gorig)$ with either left or right part of min Steiner cuts along the path from root to leaf. Therefore, $X_i$ does not contain any vertex from min Steiner cuts.

It remains to prove that $Z$ is an $(X_i,X_j)$-separator in $\Gorig$ for all $i \neq j$.  Suppose there is an $(X_i,X_j)$-path $P$ in $\Gorig - Z$ for some $i,j$. Let $(L^*,S^*,R^*)$ be the min Steiner cut obtained at the LCA $p$ of leaf $i$ and leaf $j$. Since $S^* \subseteq Z$, $P$ is also an $(X_i,X_j)$-path in $\Gorig - S^*$. Finally, we prove that $S^*$ is an $(X_i,X_j)$-separator in $\Gorig$ which leads to a contradiction. By recursion tree and \Cref{def:output partition}, $X_i \subseteq L^*$ and $X_j \subseteq R^*$. By \Cref{lem:S is LR sep}, $S^*$ is $(L^*,R^*)$-separator in $\Gorig$. Since $X_i \subseteq L^*$ and $X_j \subseteq R^*$, $S^*$ is an $(X_i,X_j)$-separator in $\Gorig$. 

\end{proof}

\begin{proof} [Proof of \Cref{claim:reducing covering set}]
Fix $A,B \subseteq T$ such that $\mu_{\Gorig}(A,B) \leq c$. If $C$ covers some min $(A,B)$-separator or $\mathcal{S}$ reduces the non-terminal part of some brittle min $(A,B)$-weak separator, then we are done. Now, we assume that $\calS$ splits some brittle min $(A,B)$-weak separator $S'$. We prove that $|S' \cap N_{\Gorig}[X_i]| \leq |S'| -1$ for all $i$. To do so, we show that $N_{\Gorig}[X_i]$ does not cover $S'$ for all $i$. Since $\calS$ splits $S'$, there could be many separators in $\calS$ that splits $S'$. Let $S \in \calS$ be the min Steiner cut $(L,S,R)$ in one of the subproblem in the recursion tree where $x \in L$ and $y \in R$ (or symmetrically $y \in L$ and $x \in R$). Such an $S$ exists by selecting the first level in the recursion tree that $S'$ is split (if $S'$ is not split at root and its non-terminal part is never reduced, then $S'$ must be split at either left subtree or right subtree). Let $x \in N_{\Gorig}[X_i]$ for some $i$. We prove that $y \not \in N_{\Gorig}[X_i]$. Indeed, since $x$ and $y$ are on different side of $S$ in the recursion, we have $y \not \in X_i$. By \Cref{lem:S is LR sep}, $S$ is an $(L,R)$-separator in $\Gorig$. Since $X_i \subseteq L$ and $y \in R$, $y \not \in N_{\Gorig}(X_i)$. This completes the proof.  
\end{proof}

\begin{proof} [Proof of \Cref{claim:small total boundaries}]

The bounds for $Z$ and $C$ follow from \Cref{lem:abstract reducing set size}. 
We next bound the sum of total neighbors. At any time, we have a collection of graphs and their terminal set  $\mathcal{G} = \{(G,T,c,X)\}$ at leaf nodes (that are not necessarily the base case at the moment)  in the recursion tree. Initially, we have only one node which is the
input graph and its terminal set. We keep track of the progress via
the following potential function
$$\Pi(\calG) = \sum_{(G,T,c,X) \in \calG} |N_{\Gorig}(X)|.$$
Notice the neighbors of each set corresponds to those in the original input graph $\Gorig$. 

Initially, $\calG$ contains only the input graph and its terminal
set. Thus, $\Pi(\calG) = 0$. At any time,  the recursion tree changes at a
leaf node that is not a base case. We claim that the potential can increase at most $4c$ whenever a min Steiner cut is obtained at any subproblem. If this is the case, then at the end of the recursion we have $\Pi(\calG) = \sum_{(G,T,c,X) \in \calG} |N_{\Gorig}(X)| = \sum_{i \leq \ell}|N_{\Gorig}(X_i)| = O(kc^2)$ (because the number of internal nodes is at most $O(kc)$ by \Cref{lem:recursion tree: total number of internal nodes}). 

It remains to prove the claim. We set up notations and make observation.  At the current node with $(G,T,c,X)$, suppose we obtain a min Steiner cut $(L,S,R)$. 
The tree will create two children $(G_L,T_L,c,X_L)$ and $(G_R,T_R,c,X_R)$ as described above. Suppose $|N_{\Gorig}(X)| = k_0$. 
\begin{claim} \label{claim:neighbor split small}
 $|N_{\Gorig}(X_L)| + |N_{\Gorig}(X_R)| \leq k_0 + 3c$.  
\end{claim}
\begin{proof}%
We first show that 
\begin{align} \label{eq:ngorigxlxr}
    N_{\Gorig}(X_L) \subseteq N_{\Gorig}(X) \cup S \textrm{ and } N_{\Gorig}(X_R) \subseteq N_{\Gorig}(X) \cup S.
\end{align}
 Indeed, since $X_L = X \cap L \subseteq X$, we have $N_{\Gorig}(X_L) \subseteq N_{\Gorig}(X) \cup X$. Since $(L,S,R)$ is a vertex cut in $G$ and $X \subseteq V(G)$, we have $N_{\Gorig}(X_L) \subseteq N_{\Gorig}(X) \cup S$. The proof for $N_{\Gorig}(X_R)$ is similar.  Second, $S$ is an $(X_L,X_R)$-separator in $G^{\textrm{orig}}$ because \Cref{lem:S is LR sep} implies that $S$ is $(L,R)$-separator in $\Gorig$ and $X_L \subseteq L$ and $X_R \subseteq R$.%
 
  Since $S$ is an  $(X_L,X_R)$-separator in $\Gorig$, we have that for every vertex $v$ in $\Gorig$ if $v \in N_{\Gorig}(X_L) \cap N_{\Gorig}(X_R)$, then $v \in S$. In particular, there are at most $|S| \leq c$ vertices in $N_{\Gorig}(X)$ that can be both neighbors of $X_L$ and $X_R$. By \Cref{eq:ngorigxlxr} and the fact that at most $|S| \leq c$ of $N_{\Gorig}(X)$ can be both neighbors of $X_L$ and $X_R$, we conclude that  $|N_{\Gorig}(X_L)| + |N_{\Gorig}(X_R)| \leq (|N_{\Gorig}(X)|+ |S|) +2|S|= k_0 + 3|S| \leq k_0 + 3c.$
 \end{proof}
Therefore, \Cref{claim:neighbor split small} imply that the potential can increase at most $3c$, and we are done.%
\end{proof}

\subsection{Fast Implementation for Few Terminals} \label{sec:fast impl few termminals}

We describe the implementation of \Cref{alg:sparsify1} that is fast when the number of terminal is small (e.g., $|T| = \log^{O(1)}(n)$).

\begin{lemma} \label{lem:alg1 sim slow}
 \Cref{alg:sparsify1} can be simulated by another algorithm that runs in $\ot(k^3mc +km2^{O(c^2)})$ time where $m$ is the number of edges and $k$ is the number of terminals of the input graph for \Cref{alg:sparsify1}. 
\end{lemma}

We first describe the key challenge
for fast implementation.  Given $G$ with $n$ vertices and
$m$ edges, recall that $G_L$ ($G_R$) is constructed by adding clique
edges between neighbors of $\hat R$ ($\hat L$).  The key challenge is
that the size of the graph $G_L$ and $G_R$ can be $O(n^2)$ even if the
$G$ is sparse. Our approach is to simulate \Cref{alg:sparsify1} using hypergraphs
instead. The intuition is that each clique in the graph can be viewed as a
hyperedge in the corresponding hypergraph. That is, adding clique edges in $G$ can be viewed as
adding a hyperedge that is formed by the union of existing
hyperedges. We next formalize the intuition. 

\paragraph{Hypergraphs.} We use bipartite representation %
of a hypergraph $H=(V_{\mathsf{v}},V_{\mathsf{e}},E)$ where $V_{\mathsf{v}}$ is a set of vertices, and $V_{\mathsf{e}}$ is a
set of hyperedges, and the set of edges in bipartite graph $E$
represents vertex and hyperedge incidence in $H$. The size of
hypergraph $H$ is $|E|$. Given an hyperedge  $e \in V_{\mathsf{e}}$, we denote $N_{H}(e) \subseteq V_{\mathsf{v}}$ to be the set of vertices that are incident to $e$. %

\begin{definition}
  Given a hypergraph $H =(V_{\mathsf{v}},V_{\mathsf{e}},E)$,
  $\textsf{Clique}(H)$ is a graph $G = (V_{\mathsf{v}},E')$ where $E'$
  is obtained by adding clique edges of $N_{H}(e)$ for every hyperedge
  $e \in V_{\mathsf{e}}$.  We say that a graph $G$ \textit{represents}
  a hypergraph $H$ if $G = \textsf{Clique}(H)$.
\end{definition}

\begin{observation}\label{obs:hypergraph clique equiv}
  Let $H$ and $G$ be a hypergraph and a graph such that $G =
  \textsf{Clique}(H)$. Then, a vertex cut $(L,S,R)$ in $G$ is a vertex cut in $H$ and vice versa. 
 \end{observation}
 \begin{proof}
 If $(L,S,R)$ is a vertex cut in $G$, then  there is no edge between $L$ and $R$. Since $\textsf{Clique}(H) = G$, every hyperedge cannot be incident to both $L$ and $R$. Therefore, $(L,S,R)$ is a vertex cut in $H$. If $(L,S,R)$ is a vertex cut in $H$, then every hyperedge cannot be incident to both $L$ and $R$. Since $\textsf{Clique}(H) = G$, there is no edge that is incident to $L$ and $R$ in $G$, and thus $(L,S,R)$ is a vertex cut in $G$. \qedhere
 \end{proof}

\begin{proposition} [\Cref{line: construct GL GR}] \label{pro:hypergraph basic gl gr}
  Given a hypergraph $H$ of size $p$  such that $G = \textsf{Clique}(H)$, a
  terminal set $T$, and a min Steiner cut $(L,S,R)$, there is an
  algorithm that runs in $O(p)$ time and outputs two hypergraphs $H_L$
  and $H_R$ 
  such that $G_L = \textsf{Clique}(H_L)$ and $G_R =
  \textsf{Clique}(H_R)$, respectively where $G_L$ and $G_R$ are left and right graphs of $G$ with respect to $(L,S,R)$ (\Cref{def:left right graphs}). Furthermore, the size of each hypergraph  is  no more than that of $H$. %
\end{proposition}
\begin{proof}
  We explain the construction of $H_R$. The construction of $H_L$ is
  similar. Let $\hat L = L \cup (S-T)$. We view $H$ as a bipartite
  graph, and construct $H_L$ as follows. Given $(L,S,R)$, we first
  compute $\hat E = N_H(\hat L)$ the set of hyperedges incident to $\hat
  L$. Then, we compute $N_H(\hat E)$ the set of vertices that are
  incident to $\hat E$. Then, we add a terminal $t_L$ to
  $V_{\textsf{v}}$, and add a new hyperedge $\hat e$  to
  $V_{\textsf{e}}$ where $\hat e$ is incident to every vertex in $\{t_L\} \cup N_H(\hat E)$. Observe that $\hat e$
  represents a clique between the neighbors of $\hat L$ in the original
  graph. Finally, we clean up by deleting all hyperedges in $\hat E$ and
  vertices in $\hat L$. By design, $G_R = \textsf{Clique}(H_R)$, and we can
  implement in $O(p)$ time.  
\end{proof}

We state three subroutines that implement \Cref{line:base 4c2,line:min
   steiner cut slow,line:maximal isolating cut} of \Cref{alg:sparsify1} using
 hypergraphs.
 
 \begin{proposition} \label{pro:hypergraph implementation}
   Let $H$ be a hypergraph of size $p$ and  $T$ be a set of $k$ terminals. There are
   the following deterministic algorithms whose proofs are almost the same as the graph version (where we can view hypergraphs as bipartite graphs).%
   \begin{enumerate}
   \item  (\Cref{line:base 4c2}) Given $H,T$, compute a min Steiner
     cut of size at most $c$ or certifies that min  Steiner cut is at
     least $c$ in $O(k^2 pc)$  time. 
   \item (\Cref{line:min steiner cut slow}) Given $H$ and $A,B \subseteq T$, compute a minimum $(A,B)$-weak separator of size at most
     $c$, or certifies that $\mu(A,B) \geq c+1$ in $O(pc)$ time. 
   \item (\Cref{line:maximal isolating cut}) Given $H$ and $t \in T$,  compute a maximal $t$-isolating
     cut of size at most $c$ (therefore, a good $t$-isolating cut), or correctly report that the
     size of $t$-isolating cut is at least $c+1$ in $O(pc)$ time. 
   \end{enumerate}
 \end{proposition}

We are now ready to prove \Cref{lem:alg1 sim slow}.

 \begin{proof}[ Proof of \Cref{lem:alg1 sim slow}]

 We simulate \Cref{alg:sparsify1} by using a hypergraph $H$ instead of
 a graph $G$ as an input such that $G = \textsf{Clique}(H)$. We apply
 \Cref{pro:hypergraph implementation} to implement \Cref{line:base 4c2,line:min
   steiner cut slow,line:maximal isolating cut} of
 \Cref{alg:sparsify1}. Next, we describe how to implement \Cref{line: construct
   GL GR}. Given a min Steiner cut $(L,S,R)$, we 
 apply \Cref{pro:hypergraph basic gl gr} to construct $H_L$ and $H_R$
 such that  $G_L = \textsf{Clique}(H_L)$ and $G_R =
 \textsf{Clique}(H_R)$, and recurse.

We now analyze the modified algorithm based on hypergraphs. Consider the recursion tree of
the modified algorithm where we use hypergraphs instead of graphs as
inputs.  Each node in the recursion tree is slightly changed from
$(G,T)$ to $(H,T)$ where $H$ is a hypergraph and $G$. By induction on the depth
 of recursion tree and \Cref{pro:hypergraph basic gl gr}, we have that
 at any time $G = \textsf{Clique}(H)$. Furthermore, the size of
 hypergraphs is at most $m$, and thus the time to construct $H_L$ and
 $H_R$ according to \Cref{pro:hypergraph basic gl gr} is also $O(m)$.  Since $G
 = \textsf{Clique}(H)$ at any time, \Cref{lem:recursion tree: total
   number of internal nodes} implies that the number of  balanced min Steiner cuts, good $t$-isolating cuts, and unbalanced non-isolating cuts (all in hypergraph versions) are at most $O(k/c),
 O(kc)$ and $O(kc)$ respectively. Thus, the number of calls to min
 Steiner cuts is at most $O(kc)$, and the number of calls to maximal $t$-isolating Steiner cut is at most
 $O(kc)$. The number of calls to the base case is at most
 $O(kc)$.  Therefore, the running time is $O(kc (k^2 mc+ mc) + k c
 \cdot c m 2^{O(c^2)}) = O(k^3mc + km 2^{O(c^2)})$.
\end{proof}

\begin{corollary} \label{cor:fast impl few terminals}
Given a graph $G$ with terminal set $T$ and a parameter $c>0$, there is an algorithm that computes a  $(T,c)$-\roc \psetpair $((Z,X_1,\ldots,X_\ell),C)$  for $G$ such that  %
   \begin{itemize}
       \item $|Z| = O(|T|c^2)$,
       \item $|C| = O(|T|\cdot 2^{O(c^2)})$, and
       \item $\sum_i |N(X_i)| = O(|T|c^2)$. 
   \end{itemize}     
The algorithm takes $\ot(|T|^3mc + |T|m2^{O(c^2)})$ time.  
\end{corollary}
\begin{proof} %
Let $(\calS,C)$ be a $(T,c)$-\roc obtained from \Cref{lem:alg1 sim slow} with $G,T,c$ as inputs. We compute $X_1,\ldots, X_{\ell}$ from the recursion tree of the algorithm in \Cref{lem:alg1 sim slow} according to \Cref{def:output partition}. This takes at most $O(n)$ time per internal node in the recursion tree to label every node according to the Steiner cut. Since there are $O(|T|c)$ internal nodes, it takes total $O(n|T|c)$ additional time to compute $X_1,\ldots, X_{\ell}$, which is subsumed by the running time \Cref{lem:alg1 sim slow}.
\end{proof}

\subsection{Fast Implementation for Expanders with Many Terminals} \label{sec:fast impl expanders}

This section is devoted to proving the following lemma. 
\begin{lemma} \label{lem:fast expander many terminals roc}
Given a $\phi$-vertex expander graph $G$ with terminal set $T$ and a parameter $c>0$, there is an algorithm that computes a  $(T,c)$-\roc \psetpair $((Z,X_1,\ldots,X_\ell),C)$  for $G$ such that %
   \begin{itemize}
       \item $|Z| = O(|T|c^2)$,
       \item $|C| = O(|T|\cdot 2^{O(c^2)})$, and
       \item $\sum_i |N(X_i)| = O(|T|c^2)$. 
   \end{itemize}     
The algorithm takes $\ot(mc\phi^{-1} + \phi^{-5} |T| 2^{O(c^2)})$ time. %
\end{lemma}

We first describe a simpler goal of simulating \Cref{alg:sparsify1} fast as stated in \Cref{lem:fast expander with many terminals}. Then, we prove \Cref{lem:fast expander many terminals roc} by showing that the algorithm in \Cref{def:output partition} for extracting a $(T,c)$-\roc \psetpair from \Cref{alg:sparsify1} can be implemented efficiently.

\begin{lemma} \label{lem:fast expander with many terminals}
 If the input graph is a $\phi$-vertex expander, then \Cref{alg:sparsify1}
 can be simulated by another algorithm that runs in $\ot(mc\phi^{-1} + \phi^{-5}
k 2^{O(c^2)})$ time where $m$ is the number of edges and $k$ is the number of terminals of the input graph for \Cref{alg:sparsify1}. 
\end{lemma}

We describe the challenges for  \Cref{lem:fast expander with many terminals}. Observe that the number of calls to min Steiner cut is $O(kc)$ (\Cref{lem:recursion tree: total number of internal nodes}), and
thus the previous algorithm can take up to $O(mkc)$ even if we assume
linear-time algorithm for computing min Steiner cut. To get a better
running time, we will exploit the crucial fact that $\phi$-expander can only have unbalanced
small separators. Therefore, we can compute min
Steiner cut, and good $t$-isolating cut in sublinear time using local algorithms.  We now make the intuition precise.  

\begin{definition}
Let $G = (V,E)$ be a $\phi$-expander graph.  For any seed vertex $v \in V$,  define $\cL_G(v) = \{ L
\subseteq V \colon v \in L, \text{ and } |L| \leq \frac{c}{\phi}\}$, and $\textlocal_G(v) = \min_{L \in \cL_G(v)}
|N_G(L)|$.  A minimizer vertex set $L \ni v$ is called \textit{local
  cut} containing the seed vertex $v$. If $|N(L)| \leq c$, we say that
$L$ is a \textit{local mincut} of size at most $c$. %
\end{definition}

We describe a different implementation of
\Cref{alg:sparsify1} that is based on local algorithms. Let $t_{\textlocal}(c,\phi)$ be the running time to compute the local
cut given a seed vertex or certify that min local cut has size at least $c+1$.  We
first modify the base case. That is, if the number of terminals is
smaller than $2c\phi^{-1}$ then we run the previous algorithm. From now
we assume that the number of terminal is at least $2c\phi^{-1}$. The
key property is that every local cut is Steiner cut and vice versa.

\begin{proposition} \label{pro:expander steiner cut}
  If $G$ is a $\phi$-expander, and $|T| \geq 2c\phi^{-1}$, then every local
  cut of size at most $c$ is a Steiner cut of size at most $c$ and vice versa.   
\end{proposition}

\paragraph{Computing min Steiner cuts.}  We maintain a list of value
$\ell(t)$ for each terminal $t$ with the invariant that $\ell(t) \leq
\textlocal_G(t)$. With this invariant and \Cref{pro:expander steiner cut}, the local cut containing the
seed node $t'$ such that $\textlocal_G(t') = \min_t\ell(t)$ is a min
Steiner cut.  By monotonicity operations of constructing $G_L$ and
$G_R$, $\textlocal_G(t)$ does not decrease for all $t$.  Therefore, to
maintain the invariant, it is enough to record $\ell(t)$ to be the latest value of
$\textlocal_G(t)$ checked by computing local cut containing $t$ so
far. The value $\ell(t)$ represents the lower bound of the current
value of $\textlocal_G(t)$. To find a Steiner cut, we compute a
min local cut given a seed vertex on the terminal $t$ whose $\ell(t)$ is
minimized, and update the value of $\ell(t)$. If the new local cut is
not the smallest among $\ell(t)$, we keep that cut and recompute the
next terminal whose $\ell$ is minimized. To find smallest $\ell(t)$,
we use priority queue. Observe that the value of $\ell(t)$
can changed at most $c$ time, and so the number of increase key
operations is at most $O(kc)$, each operation takes $O(\log n)$ time.  Since the number of Steiner cuts is at
most $O(kc)$ over the entire algorithm, the total time due to computing min
Steiner cuts is:

\begin{align}  \label{eq:local steiner cut}
     O(kc \cdot t_{\textlocal}(c,\phi) + kc \log n)
\end{align}

\paragraph{Computing good $t$-isolating cuts.} Suppose we are given a local mincut $L$
such that $(L, N(L), V - N[L])$ is a min Steiner cut and
$t$-isolating. We describe how to get a good $t$-isolating cut from
$L$. We construct $G_L$ and $G_R$ according to \Cref{lem:closure recursion lemma}. Observe that $G_L$
immediately becomes a base case. So, we work on $G_R$. Observe
that $t = t_L$ in this case. Then, we compute the next local mincut
using $t$ as a seed vertex and repeat until one of the following
happens: (1) the local cut is $t$-isolating and $\isolate(t)$
increases or (2) the local cut is $t$-isolating and $N(L) \subseteq T$
or (3) the local cut is not $t$-isolating. For case (1), the previous
local cut is a good $t$-isolating cut. For case (2), current local cut
is a good $t$-isolating cut.  Finally, for case (3), we obtain a min
Steiner cut and it reduces to the previous step.

Next, we analyze the running time. Given a local mincut, the number of
iterations until $\isolate(t)$ increases or $N(L) \subseteq T$ is at
most $2c\phi^{-1}$ since $G$ is $\phi$-expander. Each iteration takes
$t_{\textlocal}(c,\phi)$ time. The value of $\isolate(t)$ can increase
at most $c$ times, and there are at most $k$ terminals. The case
$N(L) \subseteq T$ and case (3) can also occur at most $O(kc)$ times. Therefore, the total time due to computing good $t$-isolating cuts (excluding the time to compute $G_L$ and $G_R$) is: 
\begin{align} \label{eq:good t-isolating cut} 
     O(kc^2\phi^{-1} \cdot t_{\textlocal}(c,\phi))
\end{align}

Next, we explain how to efficiently compute $G_L$ and $G_R$ using
hypergraphs. For brevity, we will focus on constructing left and right graphs of $G$ with respect to $(L,S,R)$. The construction of strictly left and strictly right graphs is similar.  

\paragraph{Simulating with Hypergraphs.} For the same reason as in the previous section, in order to construct
$G_L$ and $G_R$ according to \Cref{line: construct GL GR}, we simulate the
algorithm using hypergraph $H$ such that $G = \textsf{Clique}(H)$ as
follows.  Throughout this section, we view the hypergraph as a bipartite graph $H =
(V_{\textsf{v}},V_{\textsf{e}},E)$ and denote $G_{\textinput}$ as the
initial input graph. Initially, $H$ is a
$\phi$-expander, and furthermore for all $v \in V_{\textsf{v}}$,
$\textdeg_H(v) = \textdeg_G(v)$. To support efficient construction of
$G_L$ and $G_R$, we need a hypergraph that is equipped with additional
operation as follows. 

\begin{definition} \label{def:merge hypergraph}
   \textit{Mergable hypergraph} $H$ is a hypergraph $(V_{\textsf{v}},
   V_{\textsf{e}},E)$ that supports the following operation:
   $\textsc{Merge}( \hat E \subseteq V_{\textsf{e}})$: contract $\hat
   E$ into a single node in $V_{\textsf{e}}$ (i.e., taking the union of
   hyperedges in $\hat E$), there may be parallel edges from this
   operation.  
 \end{definition}
 
\begin{lemma} [Hypergraph With Efficient Hyperedge
  Merging] \label{lem:merging hypergraph}
  There is a mergable hypergraph that supports $\textsc{Merge}(\hat E)$
  operation in amortized $O(\alpha(m) \cdot |\hat E|) = \ot(|\hat E|)$ where $\alpha(m)$ is the inverse
  Ackermann function and $m$ is the number of hyperedges, and the time
  to access an edge in $E$ is $O(\alpha(m)) = \ot(1)$ per access. 
\end{lemma}

Intuitively, we can view mergable hypergraphs as hypergraphs but with
$O(\alpha(m)) = \ot(1)$ overhead for accessing an individual edge in the
bipartite representation. For the rest of the section, when we say
hypergraphs we mean mergable hypergraphs, and the cost for accessing
edges in its bipartite representation is $O(\alpha(m)) = \ot(1)$.

We maintain the invariant that at any time, $H = \textsf{Clique}(G)$
and $H$ is an expander with low volume in the following sense.  

\begin{definition}
We say that a hypergraph $H$ is a $(c,\phi)$-\textit{low-volume expander} if for all
vertex cut $(L,S,R)$ such that $|L| \leq |R|$ and $|S| \leq c$, we
have  $\vol_H(L) \leq (\frac{2c}{\phi})^2$ and $H$ is a
$\phi$-expander. 
\end{definition}

A key feature of low-volume expander is that we can compute local cuts
in sublinear time.

\begin{lemma} \label{lem:hypergraph local cut}
 If $H$ is a $(c,\phi)$-low volume expander, then for any seed vertex
 $t$, we can compute min local cut containing $t$ or correctly report
 that local min cut containing $t$ has size at least $c+1$ in $\ot(c^{O(c)}
 \phi^{-2})$ time. Therefore, $t_{\textlocal}(c,\phi) =  \ot(c^{O(c)}\phi^{-2})$. 
\end{lemma}

Observe that initially the hypergraph is indeed a $(c,\phi)$-low-volume
expander.  We will show that at any time every hypergraph during the recursion is $(c,\phi)$-low-volume expanders.

\begin{definition}
  Given a hypergraph $H$ and a graph $G$, we say that $H$ is
  \textit{degree-dominated} by $G$, denoted $H \leq_{\textdeg} G$ if
  for all $v \in V_H$,  if $v \in V_G$, then $\textdeg_{H}(v) \leq
  \textdeg_G(v)$; otherwise, $\textdeg_H(v) \leq 1$.
\end{definition}

\begin{lemma} \label{lem:fast GLGR hypergraph}
  Let $H = (V_{\textsf{v}},V_{\textsf{e}},E)$ be a hypergraph that is a $(c,\phi)$-low-volume expander and
  degree-dominated by $G_{\textinput}$ and $G =
  \textsf{Clique}(H)$, and $T$ be a terminal set. Let $G_L, G_R,
  T_L,T_R,t_L,t_R$ be the graphs and terminal sets  defined in \Cref{def:left right graphs}. Given a pair of vertex sets $(L,S)$ of $H$ such
  that $N_H(L) = S$ and $|S| \leq c$, there is an algorithm that
  outputs a hypergraph $H_L$ with terminal set $T_L$ and modify $H$ to
  be another hypergraph $H_R$ with terminal set $T_R$  satisfying the following properties: 
  \begin{enumerate}
  \item  $G_L = \textsf{Clique}(H_L)$ and $G_R =
    \textsf{Clique}(H_R)$,
  \item $H_L \leq_{\textdeg} G_{\textinput}$ and $H_R \leq_{\textdeg}
    G_{\textinput}$, 
  \item $H_L$ and $H_R$  are both $(c,\phi)$-low-volume vertex
    expanders, and
  \item $H_L$ has at most $2c\phi^{-1}$ vertices and is of size at
    most $O(c^3\phi^{-2})$. 
  \end{enumerate}
   The algorithm runs in $\ot( (c\phi^{-1})^3+ (c\phi^{-1}) \cdot \vol_H(S-T))$ time. 
\end{lemma}

We prove \Cref{lem:merging hypergraph},  \Cref{lem:hypergraph local cut} and \Cref{lem:fast GLGR hypergraph} in \Cref{sec:impl detials}. %
We are now ready to prove \Cref{lem:fast expander with many
  terminals}. We start with algorithm, followed by analysis. 
  
\paragraph{Modified Algorithm.} We describe the modified version of \Cref{alg:sparsify1}.  Let $G_{\textinput}$ be the initial input graph that is
  $\phi$-vertex expander, and $T$ be the set of terminals for $G_{\textinput}$.  We first construct a mergable hypergraph
  $H$ such that $G_{\textinput} = \textsf{Clique}(H)$. This step takes
  $O(m)$ time. Observe that $H$ is $(c,\phi)$-low volume vertex
  expander.  The new base case is that whenever the number of terminals is smaller than
  $2c\phi^{-1}$. If that happens, we apply \Cref{lem:alg1 sim slow}.
  Otherwise, we compute min Steiner cuts and good $t$-isolating cuts as discussed above. When a local mincut $L$ is found, we apply \Cref{lem:fast GLGR hypergraph} to implement \Cref{line: construct GL GR}, and recurse on $H_L$ and $H_R$. Observe that $H_L$ will always go to the base case
  since the number of vertices is at most $2c\phi^{-1}$. By
  \Cref{lem:fast GLGR hypergraph},
  we have that at any time the hypergraph $H_R$ is
  $(c,\phi)$-low-volume vertex expander, and $H_R \leq_{\textdeg}
  G_{\textinput}$. 

\paragraph{Recursion Tree.}
   We analyze the modified algorithm. We have a recursion tree where each
 node is a pair of hypergraph and its terminal set $(H,T)$. If the
 current node is $(H,T)$ with $(L,S,R)$ as a min Steiner cut, the left child is $(H_L,T_L)$
 and the right child is $(H_R,T_R)$.  Note that
 the recursion tree is quite simple: the left node always go to the
 base case.  By induction, we can show that $G = \textsf{Clique}(H)$
 at any time. By the correctness of the base case (\Cref{lem:alg1 sim
   slow}), the modified algorithm indeed simulates  \Cref{alg:sparsify1}.

\begin{proof} [Proof of \Cref{lem:fast expander with many terminals}]

It remains to analyze the running time of the modified algorithm. Observe that at any time $H$
is $(c,\phi)$-low volume vertex expander. By \Cref{lem:hypergraph
  local cut}, we have  $t_{\textlocal}(c,\phi) = \ot(c^{O(c)} \phi^{-2})$.  The running time due to
computing min Steiner cuts is given by \Cref{eq:local steiner cut} which is:
\begin{align} \label{eq:fast steiner cut}
    \ot(kc \cdot t_{\textlocal}(c,\phi) + kc \log n) = \ot( k\phi^{-2}
  c^{O(c)}). 
\end{align} The
running time due to computing good $t$-isolating cut is given by
\Cref{eq:good t-isolating cut} (excluding the time to construct $G_L,G_R$) which is:
\begin{align} \label{eq:fast t isolating cut}
  O(kc^2\phi^{-1} \cdot
  t_{\textlocal}(c,\phi))  = \ot(k\phi^{-3}
  c^{O(c)}).
\end{align}
 We now describe the total cost for constructing subproblems $H_L$ and $H_R$ given the current hypergraph
$H$. Note that the time includes the construction of $G_L,G_R$ when we
compute good $t$-isolating cuts.  By \Cref{lem:fast GLGR hypergraph}, the running time is $\ot(
(c\phi^{-1})^3 + (c\phi^{-1})\cdot \vol_H(S-T))$ per internal node in the recursion tree. Since the
number of internal nodes is at most $kc$, the first term sums to $\ot(
kc^4 \phi^{-3})$. The bound the second term, observe that each node $v
\in S-T$ is non-terminal, and $v$ corresponds to the
original vertex in $G_{\textinput}$. Since $H \leq_{\textdeg}
G_{\textinput}$, we have that $\textdeg_H(v) \leq \textdeg_{G_{\textinput}}(v)$ for
all $v \in S - T$.  Since each $v$ appears in $S - T$ at most once, we
have that the total cost over the entire recursion is at most
$\ot( (c\phi^{-1}) \cdot \sum_{v \in V}\textdeg_{G_{\textinput}}(v)) =\ot(
c\phi^{-1} \cdot m)$.   Therefore,  the running time due to \Cref{line: construct
  GL GR} is:
\begin{align} \label{eq:split GLGR}
  \ot( kc^4 \phi^{-3} + c\phi^{-1} \cdot m).
\end{align} 
Finally, we bound the running time due to the base cases.  By
\Cref{lem:fast GLGR hypergraph}, each $H_L$ becomes a base case of
size $O(c^3\phi^{-2})$.  Since there are at most $O(kc)$ min Steiner
cuts throughout the recursion, the total size of the base cases is $m'
= O(kc^4 \phi^{-2})$. Let $k'$ be the number of terminal in the base
case. By definition, $k' \leq 2c\phi^{-1}$. By \Cref{lem:alg1 sim
  slow}, the total running time due to the base cases is given by
\begin{align} \label{eq:fast basecase}
 \ot(k'^3m' c + k'm' 2^{O(c^2)}) = \ot(c^8 \phi^{-5} k + c^5
\phi^{-3} 2^{O(c^2)} k) = \ot( \phi^{-5} 2^{O(c^2)} k).  
\end{align}
Summing all the cost from \Cref{eq:fast steiner cut,eq:fast t isolating cut,eq:split GLGR,eq:fast basecase}, we obtain the
desired running time. 
\end{proof}

\paragraph{Extracting A $(T,c)$-\Roc \Psetpair.} Finally, we prove \Cref{lem:fast expander many terminals roc}.

\begin{proof} [Proof of \Cref{lem:fast expander many terminals roc}]
By \Cref{lem:fast expander with many terminals}, the modified algorithm simulates \Cref{alg:sparsify1} efficiently. Given a recursion tree of the modified algorithm, our next goal is to compute the partition $(Z, X_1,\ldots,X_{\ell})$ according to \Cref{lem:output partition} efficiently. In order to extract \roc \psetpair, we apply \Cref{cor:fast impl few terminals} on the new base cases. More precisely, we define \textit{the right spine} of the recursion tree as the path from root to the first right-only descendant of the root that is not a base case.  The right spine of the recursion tree can be viewed as an execution on the original graph until it becomes a base case. Indeed, observe that at the root node we start with the entire graph. Then, whenever we find a local cut, construct $H_L$ and recurse on the left (which immediately goes into the base case that calls \Cref{cor:fast impl few terminals}), and continue on the remaining graph.  

We now describe the desired output $((Z,X_1,\ldots ,X_{\ell}),C)$. Observe that a node $v$ is a base case if and only if $v$ is not in the right spine of the recursion tree. For each base case $v$, we apply \Cref{cor:fast impl few terminals} to obtain a \roc \psetpair $(\Pi_v,C_v)$ at the subproblem in node $v$ in the recursion tree. Recall that $\Pi_v$ consists of the reducer and a non-reducer sequence. The final output is $\Pi = (Z,X_1,\ldots,X_{\ell})$ where the reducer of $\Pi$ is $Z =  Z_{\text{base}} \cup Z_{\text{spine}}$ where $Z_{\text{base}}$ is the union of the sets of the reducer of the partition $\Pi_v$ over all base case $v$, and $Z_{\text{spine}}$ is the union of Steiner cuts along the right spine of the recursion tree, and the non-reducer of $\Pi$ can be obtained by the union of of all non-reducer sequence of all partition $\Pi_v$ over all the base cases $v$. Finally, $C$ is a union of $C_v$ over all base cases $v$.  It follows by design that the pair $(\Pi,C)$ is a $(T,c)$-\roc of $G_{\textinput}$ with the desired properties because we compute the pair $(\Pi,C)$ according to \Cref{lem:output partition}. The running time to compute the pair $(\Pi,C)$ is dominated by the algorithm to complete the recursion tree. 
\end{proof}

\subsection{Implementation Details} \label{sec:impl detials} %

We prove \Cref{lem:merging hypergraph},  \Cref{lem:hypergraph local cut} and \Cref{lem:fast GLGR
  hypergraph} in this section.  For convenience, we treat mergable
hypergraphs as normal hypergraphs and hide the $O(\alpha(m))$ factor
overhead in the polylog factors. We may call merge operation when
needed.

\subsubsection{Proof of \Cref{lem:merging hypergraph}}

To implement mergable hypergraph, we view $H =
(V_{\textsf{v}},V_{\textsf{e}},E)$ as a bipartite graph. Furthermore,
instead of viewing $V_{\textsf{e}}$ as a set of vertices, we view
$V_{\textsf{e}}$ as a collection of sets. More precisely,
recall that initially our hypergraph satisfies the following
$G_{\textinput} = \textsf{Clique}(H)$. Suppose initially
$V_{\textsf{e}} = \{e_1,\ldots,e_m\}$. We view $V_{\textsf{e}} = \{
\{e_1\},\ldots, \{e_m\}\}$. To implement $\textsc{Merge}(\hat E)$, we
replace each set in $\hat E$ with $\bigcup_{C \in \hat E} C$. We
implement set union operation using union-find data structures
$\mathcal{U}$ where the ground set is $V_{\textsf{e}}$ in the initial hypergraph. For any
edge $(u,e) \in E$ in the original bipartite graph, we have the
corresponding edge $(u, \mathcal{U}.\text{find}(e))$ in the mergable hypergraph. Note that this allows
parallel edges in bipartite graph after merging operations.     Also, union and find operations take amortized $O(\alpha(m))$ time where $m$ is the size of the hypergraph. 

\subsubsection{Proof of  \Cref{lem:hypergraph local cut}}

 Recall that $H = (V_{\textsf{v}},V_{\textsf{e}},E)$. We view $H$ interchangably as a
     bipartite graph and a hypergraph.  The local cut algorithm is the
     following: Given a seed vertex $t$, we
     apply LocalVC algorithm on the bipartite graph $H$ using volume parameter $\nu =
     (c+1)(c\phi^{-1})^2$ and cut parameter $c$. If LocalVC outputs a
     local cut $L$, then we binary search on value $c$ until we find
     the $c = \textlocal(t)$.  Otherwise, we report that there is no
     local mincut of size at most $c$.  Since LocalVC runs in $O(k^k
     \nu)$ where $k \leq c$, and $\nu = O(c^3 \phi^{-2})$, the total
     time is $\ot(c^{O(c)} \phi^{-2})$. 

     We argue the correctness. Suppose there is a vertex cut $(L,S,R)$
     in the hypergraph $H$ such that $t \in L, |S| \leq c, |L| \leq  |R|$. The
     neighbor of $L$, denoted as $N_B(L)$ in bipartite graph represents all hyperedges
     incident to $L$. Denote $L' = N_B(L) \cup L$ the set of vertices
     $L$ and their hyperedges in the bipartite graph. By definition,
     $N_B(L') = S$ and $|N_B(L')| = |S| \leq c$. Therefore,
     $\vol_B(L') = \vol_B(L) + \vol_B(L) \cdot |S| \leq (c\phi^{-1})^2(c+1)$.
     The last inequality follows since $H$ is $(\phi,c)$-low volume
     expander  (so we have $\vol_B(L) = \vol_H(L) \leq
     O((c\phi^{-1})^2)$).  Since $\vol_B(L') \leq (c+1)(c\phi^{-1})^2$
     and $N_B(L') = |S| \leq c$, LocalVC will output a local cut. If
     there is no such $(L,S,R)$, then LocalVC always reports there is
     no local cut. 
     
\subsubsection{Proof of \Cref{lem:fast GLGR hypergraph}}

 Since $H$ is $(\phi,c)$-low volume vertex expander, we have
 \begin{align}  \label{eq:low vol vertex exp}
   \vol_H(L) \leq O((c\phi^{-1})^2) \mbox{ and } |L |+|S| \leq
   2c\phi^{-1}.
   \end{align}
 
\paragraph{Construction of $H_L$.}   We first show that we can
 construct $H_L$ satisfying aforementioned properties in $O( (c\phi^{-1})^3 + (c\phi^{-1}) \cdot
 \vol_H(S-T))$ time.  Recall the set $\hat R = R \cup (S -  T)$ as defined in \Cref{lem:closure recursion lemma}. Our algorithm is to identify neighborhood
 of $\hat R$ in $L \cup (S - T)$, denoted as $N(\hat R)$, without reading the entire set of
 $\hat R$.  Let $H'_L$ be the hypergraph induced by $H[L \cup
 S]$.

Next, we describe the algorithm that compute $N(\hat R)$. Observe that if $v \in L$, then any hyperedge
 containing $v$ must incident to at most $|L| +|S|$
 vertices. Therefore, we can check for every vertex $v \in L$ if $v$ is a neighbor of $\hat R$ in total $O(\vol_H(L) \cdot (|L|+|S|))
 \overset{(\ref{eq:low vol vertex exp})}{=}O((c\phi^{-1})^3)$ time.  Next, observe that for every hyperedge $e$ containing a vertex
 $v$ in $S$, if $e$ is incident to more than $2c\phi^{-1}  \overset{(\ref{eq:low vol vertex exp})}{\geq} |L|+|S|$
 vertices, $v$ must be incident to $\hat R$. So, we need to check only
 hyperedges incident to $S - T$ that incident at most $O(c\phi^{-1})$
 vertices. This takes $O(\vol_H(S-T) \cdot c\phi^{-1})$ time.  

 Once $N(\hat R)$ is found, we modify $H'_L$ as follows. We
 remove $S - T$, remove  hyperedges incident to $N(\hat R)$, and add a vertex $t_R$ and a new hyperedge to that incident to every vertex in $N(\hat R) \cup \{t_R\}$. The modified
 hypergraph corresponds to $H_L$ where $G_L =
 \textsf{Clique}(H_L)$. The number of vertices in $H_L$ is
 at most $|L|+|S|  \overset{(\ref{eq:low vol vertex exp})}{\leq}
 2c\phi^{-1}$. Furthermore, every old hyperedge in $G_L$ must be
 incident to $L \cup (S-T)$. Let $E'$ be the set of hyperedges that
 incident to $L$. Since $|S| \leq c$, the total size of $E'$ is at
 most $O(c\cdot \vol_H(L))$.  Therefore, the size of $G_L$ is at most $O(\vol_H(L)
 + \vol_H(E')) = O(c \cdot (c\phi^{-1})^2) = O(c^3 \phi^{-2})$. Also, observe that the degree of each vertex does not
 increase. The terminal $t_R$ has degree exactly 1. Combining with
 the fact that $H \leq_{\textdeg} G_{\textinput}$, we conclude that
 $H_L \leq_{\textdeg} G_{\textinput}$.  
 
 \paragraph{Construction of $H_R$.}  We next describe the construction
 of $H_R$ in $\ot((c\phi^{-1})^2 + \vol_H(S-T))$ time. Recall the set
 $\hat L = L \cup (S -  T)$ as defined in \Cref{lem:closure recursion
   lemma}. Let $E_1$ be the set of hyperedges incident to $S - T$, and $E_2$ be the set of
 hyperedges incident to vertices only in $\hat L$. We now describe the construction.  
 \begin{enumerate}
 \item Compute $\hat E = E_1 - E_2$. This step takes
   $\ot(\vol_H(L)+\vol_H(S-T))  \overset{(\ref{eq:low vol vertex
       exp})}{=}\ot( (c\phi^{-1})^2 + \vol_H(S-T))$ time. 
 \item Apply $\textsc{Merge}(\hat E)$ operation and obtain a new
   hyperedge $\hat e$ ($\hat e$ is added to $V_{\textsf{e}}$). By \Cref{lem:merging hypergraph}, this
   step takes $\ot(|\hat E|) = \ot(\vol_H(S-T))$ time.
 \item Remove $\hat L$ from $H$. This step takes $\ot(\vol_H(L)) \overset{(\ref{eq:low vol vertex
       exp})}{=} \ot( (c\phi^{-1})^2)$ time. 
 \item Add a  terminal $t_L$ and add $(t_L, \hat e)$ edge to $E$.
 \item Add for all $s \in S\cap T$, add $(s, \hat e)$ edge to $E$. 
 \end{enumerate}

  By design, every vertex in $N(\hat L)$ is incident to the
  hyperedge $\hat e$. Combining with the fact that $G =
  \textsf{Clique}(H)$, we have $G_R = \textsf{Clique}(H_R)$. We next
  show that $H_R \leq_{\textdeg} G_{\textinput}$. By \Cref{def:merge hypergraph},
  $\textsc{Merge}$ operation does not reduce the degree in step 2.  Observe that the 
  vertex $t_L$ has degree exactly 1 by step 3. The
  only potential increase in degree (by one) is the vertices in $S \cap T$ in
  step 4. However, since $S$ is a min Steiner cut, every node $v \in
  S$ must have an hyperedge that is incident to $L$. Since we remove
  $\hat L$ in step 3, the degree must be reduced by at least
  1. Therefore, the degree of $v$ does not increase at the end of step
  5. Combining with the fact that $H \leq_{\textdeg} G_{\textinput}$ we
  conclude that $H_R \leq_{\textdeg} G_{\textinput}$. 

  \paragraph{$H_L$ and $H_R$ are both $(c,\phi)$-low volume vertex
    expander.} We now prove the last remaining property for $H_L$ and
  $H_R$. We prove for $H_L$ (the case $H_R$ is similar).  Let
  $(L,S,R)$ be a vertex cut in $H_L$ where $|S| \leq c$. By
  monotonicity property (\Cref{pro:monotonicity GLGR}), we
  have that $S$ must be a separator in $H$.  By similar argument in
  \Cref{pro:monotonicity GLGR}, $H_L$ and $H_R$ are $\phi$-vertex
  expanders.  Therefore, there is a
  vertex cut $(L',S,R')$ in $H$ such that $L'$ becomes $L$ and $R'$
  becomes $R$ after transforming $H$ to $H_L$. Since the degree of
  each vertex does not increase, we have $\vol_{H_L}(L) \leq
  \vol_{H}(L')$.  Since $H$ is a $(c,\phi)$-low volume expander, we
  also have $\vol_H(L') \leq (c\phi^{-1})^2$. Therefore,   $\vol_{H_L}(L) \leq
  \vol_{H}(L') \leq (c\phi^{-1})^2$. We conclude that $H_L$ is a
  $(c,\phi)$-low volume vertex expander. 

\subsection{The Final Algorithm} \label{sec:final reducing set}

  This section is devoted to proving \Cref{thm:reducing-or-covering set system}.
  We start with an important tool: vertex expander decomposition. 
  
\begin{theorem}[\cite{LongS22}] \label{thm:vertex expander decomp}  
	There is a deterministic algorithm that, given a graph $G=(V,E)$
	with $n$ vertices and $m$ edges and a parameter $\phi\in(0,1/10\log n)$,
	computes in $m^{1+o(1)}/\phi$ time a partition $\{Z,X_{1},\dots,X_{\ell}\}$
	of $V$ and $\{Y_1,\dots,Y_\ell\}$ such that 
	\begin{itemize}
                       \item $Z = \bigcup_i (Y_i - X_i)$,
        \item for all $i$,  $N[X_i] \subseteq Y_i \subseteq X_i \cup Z$,
		\item for all $i$, $G[Y_i]$ is a $\phi$-vertex expander, and
		\item $\sum_{i}|Y_i - X_i|\le\phi n^{1+o(1)}$.
	\end{itemize}
\end{theorem}

The first item implies that $N(X_i) \subseteq Z$ for all $i$. That is, all $X_i$ are separated from each other by $Z$.
The set $Y_i$ is the set containing $N[X_i]$ such that $G[Y_i]$ is a vertex expander, but $Y_i$ is not too big in the sense that the ``additional'' part $Y_i \setminus X_i \subseteq S$ and also $\sum_i |Y_i \setminus X_i| \le \phi n^{1+o(1)}$. Note also that a vertex could be counted in $\sum_i |Y_i \setminus X_i|$ many times. That is, we actually bound the total number of ``occurrences'' of all vertices in $Y_i \setminus X_i$ overall $i$.

Observe that each $G[X_i]$ is possibly not connected. This is
unavoidable, for example, when $G$ is a star.

Our plan is to apply expander decomposition as a preprocessing step before applying \Cref{lem:fast expander many terminals roc}  on each expander. Once we obtain \roc \psetpair for each expander, we need to "compose" them. To do so, we need the notion of composibility of $(T,c)$-\roc \psetpair for the decomposition. We state it in a general form as follows. Recall the definitions of reducer and non-reducer sequence  in \Cref{def:roc partition}.

\begin{lemma} \label{lem:roc composable}
Let $(Z,X_1,\ldots,X_{\ell})$ be a partition of $V(G)$ such that $Z$ is $(X_i, X_j)$-separator for all distinct $i,j$. For $i \in [\ell]$, let $Y_i$ be a set of vertices such that $N_G[X_i] \subseteq Y_i \subseteq X_i\cup Z$. Let $G_i = G[Y_i]$ and $T_i = (T\cap X_i) \cup (Y_i - X_i)$.  For $i \in [\ell]$, let $(\Pi_i,C_i)$ be $(T_i,c)$-\roc \psetpair in $G_i$. We define   $\textsc{Compose} (Z, (\Pi_1, C_1), \ldots , (\Pi_{\ell},C_i))$ that outputs $((Z',X_1',\ldots, X'_{\ell'}),C')$ where each set in the output is defined as follows. 
\begin{itemize}
    \item $Z' = Z \cup \bigcup_{i \leq \ell}Z_i$ where $Z_i$ is the reducer of $\Pi_i$, 
    \item $C' = \bigcup_{i \leq \ell} C_i$, and 
    \item The sequence $X'_1, \ldots X'_{\ell'}$ is obtained as follows. We first subtract $Z$ from every set in the non-reducer sequences of the partition $\Pi_i$. Then, for all partition $\Pi_i$, we append all the resulting non-reducer sequences. 
\end{itemize}

 Then, $((Z',X_1',\ldots, X'_{\ell'}),C') =  $ \textsc{Compose}$(Z,(\Pi_1,C_1), \ldots, (\Pi_\ell,C_{\ell}))$ is a $(T,c)$-\roc \psetpair in $G$ where
   \begin{itemize}
       \item $|Z'| \leq  |Z| + \sum_i |Z_i|$ where $Z_i$ is the reducer of $\Pi_i$,
       \item $|C'| \leq \sum_{i} |C_i|$, and %
       \item $\sum_{i' \in [\ell']} |N_G(X'_{i'})| \leq \sum_{i \in [\ell]} \sum_j |N_{G_i}(Q_{i,j})| + \sum_{i \in [\ell]}|Y_i - X_i| $ where $Q_{i,j}$ is the $j$-th set in the non-reducing sequence of $\Pi_i$. %
   \end{itemize}  
\end{lemma}

 We prove \Cref{lem:roc composable} in \Cref{sec:roc composable} since the proof is similar to that of \Cref{lem:closure recursion lemma}. 

We are now ready to prove \Cref{thm:reducing-or-covering set system}. %
We describe the algorithm and analysis. 

\paragraph{Algorithm.} We describe the algorithm in \Cref{alg:sparsifyfast}. Note that the input graph has arboricity $c$. First, we apply expander decomposition. For each expanders, we apply \Cref{alg:sparsify1} using terminal set in the boundary $Y_i - X_i$ and $T \cap X_i$. Since each graph is an expander, we can implement \Cref{alg:sparsify1} fast using \Cref{lem:fast expander many terminals roc}. Finally, we compose all the  \roc \psetpair to output the final $(T,c)$-\roc \psetpair in $G$.  

\begin{algorithm}[H]
  \DontPrintSemicolon
  \KwIn{ A graph $G = (V,E)$, a terminal set $T \subseteq V$ and parameters $c >0, \phi \in (0, 1/10\log |V|)$ where $G$ has arboricity $c$.} 
  \KwOut{ A $(T,c)$-\roc \psetpair in $G$.}
  \BlankLine
  Apply \Cref{thm:vertex expander decomp} using $G$ and $\phi$ to obtain $S,  X_1,\ldots,
  X_\ell$ and $Y_1, \ldots Y_\ell$. \;
  \For{$i \in [\ell]$}
  { Let $T_i = (T \cap X_i) \cup
    (Y_i - X_i)$. \;
    $(\Pi_i,C_i) = \textsc{ROC}(G[Y_i],T_i,c)$ \tcp*{\Cref{alg:sparsify1}}  
  }
  \Return{\normalfont{} \textsc{Compose}$(Z,(\Pi_1,C_1),\ldots, (\Pi_\ell,C_\ell))$  where \textsc{Compose} is defined in \Cref{lem:roc composable}.}%
\caption{\textsc{FastROC}$(G,T, c,\phi)$}
\label{alg:sparsifyfast}
\end{algorithm}

We denote $|T| = k$ and $m = |E|$. For $i \leq \ell$, denote $|T_i| = k_i$ and $|E(G[Y_i])| = m_i$. The total size of terminals is: 
\begin{align} \label{eq:sumki}
  \sum_i k_i = \sum_i  ( |T \cap X_i| + |Y_i - X_i|  ) = O(k + \phi n^{1+o(1)})
\end{align}  
Next, the total instances size  is:
\begin{align} \label{eq:summi}
  \sum_i m_i = \sum_{i}|E_G(Y_i,Y_i)| \leq \sum_i c|Y_i| =
  \sum_i(c|Y_i - X_i| + c|X_i|) =  O(c\phi n^{1+o(1)} + nc)
\end{align}  
The first inequality follows because $G$ has arboricity $c$ (i.e., for all $S \subseteq V(G)$, $|E(S,S)| \leq c|S|$).

\paragraph{Correctness.}  Let $(\Pi' = (Z',X'_1,\ldots,X'_{\ell'}), C')$ be the output of \Cref{alg:sparsifyfast}. For $i \leq \ell$, let $Z_i$ be the reducer of partition $\Pi_i$, and let $Q_{ij}$ be $j$-th set of the non-reducer sequence of $\Pi_i$. Since $(\Pi_i,C_i)$ is $(T_i,c)$-\roc for all $i$, \Cref{lem:roc composable} implies that $(\Pi',C')$ is $(T,c)$-\roc where the sizes of $|Z'|, |C'|, \sum_{i' \in [\ell']} |N_G(X'_{i'})|$ are stated in \Cref{lem:roc composable}. We now argue each one accordingly. First, we bound the size of $Z'$.
\begin{align*}
   |Z'| \leq |Z| + \sum_{i} |Z_i| \leq  \phi n^{1+o(1)} + O((\sum_i k_i)c^2) \overset{(\ref{eq:sumki})}{=} O((k+\phi n^{1+o(1)})c^2).
\end{align*}
The first inequality follows by \Cref{lem:roc composable}. For the second inequality, $|Z| \leq \phi n^{1+o(1)}$ because $Z = \bigcup_i (Y_i - X_i)$, and $\sum_{i}|Y_i - X_i|\le\phi n^{1+o(1)}$ (\Cref{thm:vertex expander decomp}), and the term $\sum_i|Z_i| \leq O(\sum_ik_ic^2)$ follows from \Cref{lem:fast expander many terminals roc}. Next, we bound the size of $C'$.
\begin{align*}
    |C'| \leq \sum_{i} |C_i| \leq (\sum_i k_i)2^{O(c^2)}   \overset{(\ref{eq:sumki})}{=}   O((k + \phi n^{1+o(1)}) \cdot
    2^{O(c^2)}) = O((\phi n^{1+o(1)} + k) \cdot 2^{O(c^2)}).
\end{align*}
The first inequality follows from \Cref{lem:roc composable}. The second inequality follows from \Cref{lem:fast expander many terminals roc}. Finally, we bound the total neighbors. 
\begin{align*}
    \sum_{i' \in [\ell']} |N_G(X'_{i'})| & \leq \sum_{i \in [\ell]} \sum_j |N_{G_i}(Q_{i,j})| + \sum_{i \in [\ell]}|Y_i - X_i|  \\
    & \leq O(\sum_{i}k_ic^2) + \phi n^{1+o(1)}\\ 
  & \overset{(\ref{eq:sumki})}{=} O((k+\phi n^{1+o(1)})c^2). 
\end{align*}
 The first inequality follows from \Cref{lem:roc composable}. The second inequality follows from \Cref{lem:fast expander many terminals roc}.

\paragraph{Running Time.} By \Cref{thm:vertex expander decomp}, the expander decomposition
algorithm takes $O(m^{1+o(1)}\phi^{-1})$ time and
guarantees that $G[Y_i]$ is a $\phi$-vertex expander for all
$i$. By \Cref{lem:fast expander with many terminals}, the
time to compute $R_i$ is $\ot(m_ic\phi^{-1} + \phi^{-5}k 2^{O(c^2)})$.
Therefore, the total running time is:
\begin{align*}
\ot(m^{1+o(1)}\phi^{-1} + (\sum_i m_i) \cdot c\phi^{-1} + \phi^{-5} (\sum_i
  k_i) \cdot 2^{O(c^2)})  \\
   \overset{(\ref{eq:sumki}\text{ and }\ref{eq:summi})}{=} \ot(m^{1+o(1)}c\phi^{-1} + (\phi^{-5}k + \phi^{-4}
                          n^{1+o(1)}) \cdot 2^{O(c^2)}) \\
  = \ot(m^{1+o(1)}\phi^{-4} \cdot 2^{O(c^2)}).
\end{align*}

\section{ Open Problems} \label{sec:open problems}

\paragraph{Deterministic Vertex Connectivity Algorithms.} It remains an outstanding open problem whether there exists a linear-time deterministic vertex connectivity algorithm as asked by Aho, Hopcroft, and Ullman in 1974 (\cite{AhoHU74} Problem 5.30).  Our results answer this question affirmatively up to a subpolynomial factor in the running time when the connectivity $c = o(\sqrt{ \log n})$.  
Can one improve the dependency on $c$ in our running time to $m^{1+o(1)}\poly(c)$? This would follow immediately by our framework if one can improve \Cref{thm:overview sparsifier} so that  a $(T,c)$-sparsifier of size $|T|\poly(c)$ can be obtained in $m^{1+o(1)}\poly(c)$ time.
For general connectivity $c$, even a deterministic algorithm with  $\tilde{O}(mn)$ time bound is still open.

How fast can we compute vertex connectivity in directed graphs? 
The algorithm by \cite{ForsterNYSY20} runs in $\tilde{O}(mc^2)$ time, which is near-linear for all $c = \polylog(n)$. When $c$ can be big, the algorithm by \cite{li2021vertex} takes $n^{2+o(1)}$ time when we apply the almost-linear-time max flow algorithm by \cite{chen2022maximum}. However, both of these algorithms are Monte-Carlo. In contrast, the fastest deterministic algorithm \cite{Gabow06} still requires in  $O(mn + m \cdot \min\{c^{5/2}, c \cdot n^{3/4}\})$ time.  No almost linear time deterministic algorithm is known even when we promise that $c=O(1)$.

\paragraph{Weighted Vertex Connectivity.}
Given a vertex-weighted graph, the \emph{weighted vertex connectivity} problem is to compute a vertex cut with minimum total weight. The state-of-the-art of this natural extension is by Henzinger Rao and Gabow \cite{HenzingerRG00} who gave a $\tilde{O}(mn)$-time Monte-Carlo algorithm and a $\ot(\kappa_0 mn)$-time deterministic algorithm where $\kappa_0$ is the unweighted vertex connectivity. It is an outstanding open problem whether this bound can be improved.

\paragraph{Mimicking Networks.} Our algorithm for mimicking network is of independent interest. An immediate open problem is to compute a mimicking network of size $k \poly(c)$ in $m^{1+o(1)} \poly(c)$ time  where $k$ is the number of terminals. The same question holds for the edge version of the mimicking networks.  For the purpose of vertex connectivity algorithm in our framework, it is enough for the sparsifier to preserve only a separator of size at most $c$ that splits terminal set $T$ (instead of all possible pair $A,B \subseteq T$ such that $\mu_G(A,B) < c$). Using a weaker notion of sparsifier, is it possible to obtain such a sparsifier of size $k\poly(c)$ in $m^{1+o(1)} \poly(c)$ time?

\section*{Acknowledgement}  
This project has received funding from the European Research Council (ERC) under the European Union's Horizon 2020 research and innovation programme under grant agreement No 715672 and No 759557.  
 We thank Yu Gao, Jason Li, Danupon Nanongkai, and Richard Peng for discussion during the early stage of this project. We thank the 2021 HIM program Discrete Optimization during which part of this work was developed. We thank the anonymous reviewer for suggesting a simple proof of \Cref{lem:static neighborhood oracle}.

        \bibliographystyle{alpha}
        \bibliography{references,dp-refs} 

        \appendix

          \section{Omitted Proofs}

\subsection{Proof of \Cref{claim:boring induction recurrence}} \label{sec:boring induction recurrence}

We prove by induction on $c$ and $\cbar$.  We omit the subpolymonial factor for simplicity of the presentation. The base cases follow trivially.  Next, we want to prove that 
\begin{align} \label{eq:recurrence appendix}
s(n,k,c,\cbar) \leq  (k+n \phi)  (4+c+\cbar)^{3(c+\cbar)} \cdot (1+\cbar) \cdot  2^{(\tau  c^2+\cbar)}
\end{align}
Assuming that \Cref{eq:recurrence appendix} holds for parameters $(c-1,\cbar)$ and $(c,\cbar -1)$, we prove that \Cref{eq:recurrence appendix} holds for the parameter $(c,\cbar)$.
Observe that
\begin{align} \label{eq:sum kini}
  \sum_i (k_i+\phi n_i ) \overset{(\ref{eq:sum ki and sum ni})}{\leq} (k + n\phi) c^2 + \phi n + \phi(k+n\phi)c^2 \leq 3(k+n\phi)c^2
\end{align}
Applying the induction hypothesis on \Cref{eq:recurrence snkc}, we obtain: 
\begin{align*} 
  s(n,k,c,\cbar) &\leq (k+\phi n) 2^{\tau c^2} +\sum_i (k_i + \phi n_i)(3+c+\cbar)^{3(c+\cbar)-3} \cdot (1 + \cbar) \cdot 2^{(\tau   (c-1)^2+\cbar)}  \\
                 & + \sum_i (k_i + \phi n_i)  (3+c+\cbar)^{3(c+\cbar)-3} \cdot ((1+\cbar) - 1) \cdot 2^{(\tau   c^2+\cbar-1)} \\
                 & \overset{(\ref{eq:sum kini})}{\leq}  \sum_i (k_i + \phi n_i)(3+c+\cbar)^{3(c+\cbar)-3} \cdot (1 + \cbar) \cdot 2^{(\tau   (c-1)^2+\cbar)}  \\
                 & + \sum_i (k_i + \phi n_i)  (3+c+\cbar)^{3(c+\cbar)-3} \cdot (1+\cbar) \cdot 2^{(\tau   c^2+\cbar-1)} \\
                 & \leq  \sum_i (k_i + \phi n_i)  (3+c+\cbar)^{3(c+\cbar)-3} \cdot (1+\cbar) \cdot( 2^{(\tau   (c-1)^2+\cbar)} +  2^{(\tau   c^2+\cbar-1)}  ) \\
                 & \leq    \sum_i (k_i + \phi n_i)  (3+c+\cbar)^{3(c+\cbar)-3} \cdot (1+\cbar) \cdot 2^{(\tau   c^2+\cbar)}  \\
                &\overset{(\ref{eq:sum kini})}{\leq} (k+n \phi) (3c^2) (3+c+\cbar)^{3(c+\cbar)-3} \cdot (1+\cbar) \cdot 2^{(\tau   c^2+\cbar)}  \\
  & \leq (k+n \phi) (4+c+\cbar)^{3(c+\cbar)}\cdot (1+\cbar) \cdot 2^{(\tau   c^2+\cbar)} .
\end{align*}
Observe $2^{(\tau   (c-1)^2+\cbar)} +  2^{(\tau   c^2+\cbar-1)} \leq  2\cdot 2^{(\tau   c^2+\cbar-1)}=  2^{(\tau   c^2+\cbar)}$ since $\tau(c^2 -2c+1) + \cbar \leq \tau c^2 + \cbar -1$.

\subsection{Proof of \Cref{lem:roc composable}} \label{sec:roc composable}
 The fact that $|Z'| \leq  |Z| + \sum_i |Z_i|$ where $Z_i$ is the reducer of $\Pi_i$ and that $|C'| \leq \sum_{i} |C_i|$ follow by design. Next,  we show that $\sum_{i' \in [\ell']} |N_G(X'_{i'})| \leq \sum_{i \in [\ell]} \sum_j |N_{G_i}(Q_{i,j})| + \sum_{i \in [\ell]}|Y_i - X_i|$. 
Observe that, for all $i,j$, 
\begin{align} \label{eq:ngqijz}
    N_G(Q_{ij} -Z) = N_{G_i}(Q_{ij} - Y_i) \subseteq  N_{G_i}(Q_{ij})\cup (Y_i \cap Q_{ij}). 
\end{align}
Since $X'_{i'} = Q_{ij} - Z$ for some $i,j$, and $Q_{ij}$ are disjoint among those with the same $i$, we have $\sum_{i' \in [\ell']} |N_G(X'_{i'})| \leq \sum_{i \in [\ell]} \sum_j |N_{G_i}(Q_{i,j})| + \sum_{i \in [\ell]}|Y_i - X_i|$. 

Furthermore, \Cref{eq:ngqijz} implies that $N_G(Q_{ij} - Z) \subseteq Z_i \cup Z \subseteq Z'$ for all $i,j$. Therefore, $Z'$ is an $(X'_i,X'_j)$-separator in $G$ for all $i,j$.  It remains to prove the following claim.

We prove that $(\Pi' = (Z',X'_1,\ldots,X'_{\ell'}),C')$ is a $(T,c)$-\roc \psetpair in $G$. Fix a pair $A,B \subseteq T$ such that $\mu_G(A,B) \leq c$. If $C'$ cover some min $(A,B)$-weak separator, then we are done. If $Z'$ contains a non-terminal vertex of some brittle min $(A,B)$-weak separator, then we are also done. Assume otherwise. Let $S$ be a brittle min $(A,B)$-weak separator in $G$. If $Z$ is an $(x,y)$-separator for some distinct $x,y \in S$, then we are done with $(A,B)$. Assume otherwise. There is $i$ such that $S\cap X_i \neq \emptyset$.  WLOG, we assume that $S \subseteq X_i \cup Z$ and $S \cap Z \subseteq T$. If $S \not \subseteq Y_i$, then there are $x,y \in S$ such that $x \in X_i$ and $y \not \in Y_i$, and thus $\Pi'$ splits $S$ and we are done. Now, assume $S \subseteq Y_i$. 

Let $A_i = A \cap X_i, B_i = B\cap X_i$.  Denote $\partial X_i = Y_i - X_i$ as the set of
  \textit{boundary} vertices.   In $G_i$, we
    say that a boundary vertex $v \in \partial X_i$ is
    $A$-\textit{proxy} if $A - A_i \neq \emptyset$, and there is an
    $(A - A_i,v)$-path in $G$ that does not use $X_i$.
    Similarly, $v$ is $B$-\textit{proxy}  if $B - B_i \neq \emptyset$ and there is an $(B -
    B_i,v)$-path in $G$ that does not use $X_i$.  Let
    $\partial A_i$ be the set of $A$-proxy boundary vertices, and
    $\partial B_i$ be the set of $B$-proxy boundary vertices. Let
    $A' = \partial A_i \cup A_i$ and $B' = \partial B_i \cup
    B_i$. %

\begin{claim} \label{claim:sa'b'gi}
$S$ is an $(A',B')$-weak separator in $G_i$. 
\end{claim}
\begin{proof}
The proof is similar to that in \Cref{claim:new terminal set}.
\end{proof}

\begin{claim} \label{claim:a'b'giabg}
An $(A',B')$-weak separator in $G_i$ is an $(A,B)$-weak separator in $G$.
\end{claim}
\begin{proof}
Since $S$ is an $(A,B)$-separator in $G$, every $(A,B)$-path in $G$ must use $S$. Since $S \subseteq Y_i$, every $(A,B)$-path in $G$ contains an $(A',B')$-subpath in $G_i$. 
Let $S'$ be an $(A',B')$-weak separator in $G_i$.  Suppose $S'$ is not  an $(A,B)$-weak separator in $G$. There must be an $(A,B)$-path $P$ in $G$ that does not use $S'$. Therefore, $P$ contains a subpath from $A'$ to $B'$ in $G_i$ that does not use $S'$, a contradiction. 
\end{proof}

By \Cref{claim:sa'b'gi} and \Cref{claim:a'b'giabg}, $ \mu_{G_i}(A',B')  = \mu_G(A,B) \leq c$. Since $S \cap Z \subseteq T$, $S - T_i = S - T$. Hence, \Cref{claim:sa'b'gi} implies that  $\mu^T_{G_i}(A',B') \leq |S - T_i| = |S - T| = \mu^T_G(A,B)$. Since $S$ is a min $(A',B')$-weak separator in $G_i$ and $C_i$ does not cover any min $(A,B)$-weak separator in $G_i$, $\Pi_i$ either splits or $T_i$-hits some min $(A',B')$-weak separator $S'$ in $G_i$. By \Cref{claim:a'b'giabg} and the fact that $\mu_{G_i}(A',B') = \mu_G(A,B)$, $S'$ is a min $(A,B)$-weak separator in $G$.  We prove that $\Pi'$ either $T$-hits or splits $S'$ in $G$, and we are done. If $\Pi_i$ $T_i$-hits $S'$, then $|(Q_{ij} \cap S') - T_i| \leq \mu_{G_i}^T(A,B) -1  \leq \mu^T_G(A,B) -1$ for all $j$. Therefore, $|((Q_{ij} - Z) \cap S') - T| \leq \mu^T_G(A,B) -1$ for all $j$. Since $S' \subseteq Y_i$, $|((Q_{i'j} - Z) \cap S') - T| \leq \mu^T_G(A,B) -1$ for all $i' \neq i$ and for all $j$. Therefore, $\Pi'$ $T$-hits $S'$. If $\Pi_i$ splits $S'$ in $G_i$, then, for all $j$, $|S' \cap N_{G_i}[Q_{ij}]| \leq |S'|-1$, and thus $|S' \cap N_G[Q_{ij}-Z]| \leq |S'|-1$. Observe that $S' \subseteq Y_i$. Therefore, $\Pi'$ splits $S'$ in $G$.

\end{document}